\documentclass[pdflatex,sn-mathphys-num]{sn-jnl}

\usepackage[markup=underlined]{changes}
\definechangesauthor[name={Jos\'{e} Rodal}, color=red]{author}

\usepackage{tcolorbox}

\newtcolorbox{highlighted}{colback=yellow!20, colframe=yellow!80!black, sharp corners, boxrule=0.5mm}

\usepackage{amsmath,amssymb,amsfonts}%
\usepackage{fix-cm}%
\usepackage{amsthm}%
\usepackage[cal=rsfs]{mathalfa}
\usepackage{graphicx}%
\usepackage{multirow}%
\usepackage[title]{appendix}%
\usepackage{xcolor}%
\usepackage{textcomp}%
\usepackage{manyfoot}%
\usepackage{booktabs}%
\usepackage{algorithm}%
\usepackage{algorithmicx}%
\usepackage{algpseudocode}%
\usepackage{listings}%
\usepackage{bm}%
\usepackage{mathtools}
\usepackage{caption}%
\usepackage{subcaption}
\usepackage{hyperref}
\usepackage{physics}
\usepackage{adjustbox} 
\usepackage{enumitem}
\usepackage{cleveref}
\usepackage{url}
\usepackage[utf8]{inputenc}
\usepackage{makecell}
\DeclareUnicodeCharacter{202F}{} %
\usepackage{booktabs}   
\usepackage{array}      
\usepackage{tabularx}

\newcommand{\eqBetaAlc}{eq:BetaAlc}

\newcommand{\fIrrotational}{fIrrotational}

\newcommand{\eqBetaIrr}{eq:BetaIrr}
\newcommand{\Lie}{\mathcal{L}}
\newcommand{\vA}{\mathcal{A}}
\newcommand{\vAphi}{\mathcal{A}_{\phi}}

\newtheorem{theorem}{Theorem}
\newtheorem{proposition}[theorem]{Proposition}%

\theoremstyle{remark}

\begin{document}

\title[A Warp Drive with Predominantly Positive Invariant Energy Density and Global Hawking–Ellis Type I]{A Warp Drive with Predominantly Positive Invariant Energy Density and Global Hawking–Ellis Type I}


\author*[1]{\fnm{Jos\'{e}} \sur{Rodal}}\email{jrodal@alum.mit.edu}

\affil*[1]{\orgname{Rodal Consulting}, \orgaddress{\street{205 Firetree Ln.}, \city{Cary}, \state{NC}, \postcode{27513}, \country{USA}}}





\abstract{We present the first fully explicit, continuous, analytically derived warp--drive spacetime within General Relativity whose shift--vector flow is kinematically irrotational. Building on Santiago \emph{et al.} that scalar--potential, zero--vorticity warp fields are Hawking--Ellis Type~I for unit lapse and flat spatial slices, we supply a closed--form scalar potential and smooth shift components with proper boundary behavior, together with a Cartan--tetrad analytic pipeline and high--precision eigenanalysis. Compared with the Alcubierre and Natário models (evaluated at identical \((\rho,\sigma,v/c)\)), our irrotational solution exhibits \emph{significantly reduced} local NEC/WEC stress: its peak proper--energy deficit is reduced by a factor of \(\approx 38\) relative to Alcubierre and \(\approx 2.6\times10^{3}\) relative to Natário, and its peak NEC violation is more than \(60\times\) smaller than Natário. Crucially, the stress--energy is \emph{globally} Hawking--Ellis Type~I, with a well--defined timelike eigenvalue (proper energy density) everywhere. A fixed--smoothing vortical ablation confirms that this improvement is causally due to irrotational, curl--free kinematics rather than profile shaping: adding modest vorticity collapses the \(E_{+}/E_{-}\) balance and drives large increases in the negative--energy magnitude \(E_{-}\). We quantify the negative--energy requirement via a \emph{slice-integrated} (on $\Sigma_t$) negative--energy volume and tabulate global measures. A far--field extrapolation to \(R\!\to\!\infty\) yields tail--corrected totals \(|E_+-E_-|/(E_++E_-)=0.04\%\). Thus the net \emph{proper} energy is consistent with zero to four decimal places (in fractional units). We also establish regularity at \(r=0\) for the irrotational construction.
}


%




\maketitle
\newpage
\tableofcontents

\section{Introduction}
\label{sec:intro}

The theoretical investigation of warp–drive spacetimes in General Relativity began with Alcubierre’s proposal~\cite{Alcubierre_1994}, which produces apparent faster-than-light travel by contracting space ahead of a “bubble’’ and expanding it behind. While mathematically consistent, the construction requires exotic matter violating the weak and null energy conditions (WEC/NEC), challenging physical plausibility.

Natário~\cite{Natário} introduced zero-expansion warp drives, characterized (under unit lapse and time-independent slicing) by vanishing spatial divergence of the shift, $K=0\Leftrightarrow D_i\beta^i=0$. These models still exhibit NEC violations.

Kinematically irrotational drives (vanishing spatial vorticity, $\boldsymbol{\omega}=0$) were analyzed by Santiago et al.~\cite{SantiagoVisser}, who showed that for unit lapse ($\alpha=1$) and flat, time-independent slices the momentum density vanishes and the stress–energy is Hawking–Ellis Type~I; they also derived bounds on local energy-condition deficits. Schuster et al.~\cite{SchusterSantiagoVisser2023} complemented this with global ADM-mass analyses, identifying generic pathologies for typical bubbles. Other approaches include discontinuous or distributional set-ups (e.g., Lentz’s phase-engineered profiles~\cite{Lentz_2021} and thin-shell models~\cite{Huey2023}), which entail junction terms; a “tilted’’ warp paradigm with nonzero vorticity and acceleration~\cite{Barzegar2024}; and broader explorations of backgrounds, numerics, and matter models~\cite{garattini2024black,garattini2025positive,Fuchs2024,Helmerich2024,AbellánVasilev,Abellán2023lapse,AbellánVasilevGRG}.

Despite these advances, the literature has lacked a \emph{fully explicit}, continuous, analytically derived construction of a kinematically irrotational warp drive with a closed-form scalar potential and continuous shift components. \textit{Critically, this gap in specific solutions is compounded by a methodological limitation:} many prior studies emphasize observer- or slice-dependent energy diagnostics rather than invariants derived from the stress–energy eigenstructure.

\textbf{This work.} We present, to our knowledge, the first explicit scalar potential $\Phi(r,\theta,t)$ that realizes a continuous, globally regular, kinematically irrotational warp-drive flow with correct asymptotics, satisfying $\boldsymbol{\omega}=*\,\dd\boldsymbol{\beta}^{\flat}=0$. Irrotationality enforces $\boldsymbol{\beta}^{\flat}=-\dd\Phi$ on each slice and, in the flat static triads adopted here, implies $G_{\hat 0\hat i}=0$; consequently $G^{\hat\mu}{}_{\ \hat\nu}$ is block-diagonal with a unique timelike eigenvalue (the proper energy density), cf.\ Prop.~\ref{prop:irrot_class}. This algebraic structure accounts for the observed reduction in local NEC/WEC deficits relative to vortical flows. Our results are consistent with and extend prior analyses—recovering the Santiago et al.\ Type-I mechanism and complementing the ADM-mass considerations of Schuster et al.—while adding an \emph{explicit} construction together with slice-integrated global energy measures (Table~\ref{tab:global_measures}).

\paragraph*{Roadmap and contributions.}
We first specify the coordinate system (Fig.~\ref{Fig1_SpherCoord}) and baseline notation; full units and extended conventions are in Appendix~\ref{Notation}. We then set up the invariant eigenstructure framework for the Einstein tensor—used to extract the proper energy density and Hawking–Ellis type—and summarize the high-precision Cartan–tetrad pipeline employed throughout. Next, we present the explicit scalar-potential construction of the kinematically irrotational warp drive and its kinematic rationale ($\dd\boldsymbol{\beta}^\flat=0 \Rightarrow G_{\hat0\hat i}=0$), followed by algebraic and energy-condition maps, whole-slice global budgets, and head-to-head comparisons with the Alcubierre and Natário models.  The main contributions are:
(i) an explicit, continuous analytic construction of a kinematically irrotational warp drive (closed-form $\Phi$ and continuous $\beta^i$);
(ii) a high-precision Cartan–tetrad pipeline that extracts invariant proper energy density via the timelike eigenvalue and cross-checks it against a Hessian invariant;
(iii) full-spacetime Hawking–Ellis classification maps and quantitative energy-condition diagnostics, showing substantially reduced local NEC/WEC deficits relative to vortical flows;
(iv) global, slice-integrated energy budgets with tail modeling, indicating near-cancellation of proper energy (baseline window and tail-extrapolated totals);
(v) a consolidated algebraic type picture under the stated assumptions ($\alpha=1$, flat, time-independent slices): the irrotational geometry is globally Type~I in this setup; Alcubierre and Natário exhibit a mix of Types~I and IV, with no Type~II/III regions observed.

\begin{figure}[htbp]
    \centering
    \begin{subfigure}[b]{0.7\textwidth}
        \includegraphics[width=\textwidth]{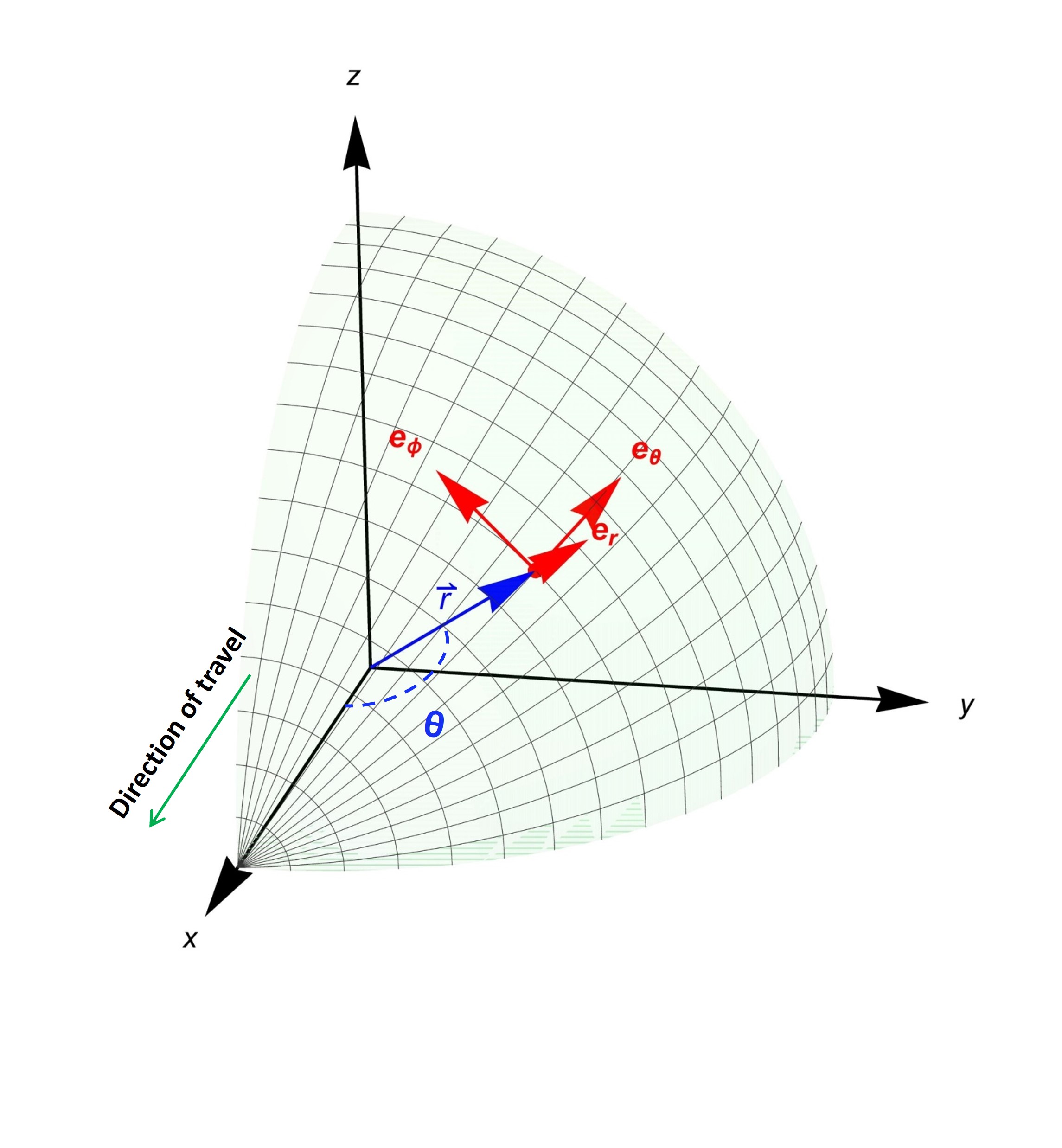}
       \caption{$\mathit{P}$ at $(r=\rho,\;0<\theta<\pi/2,\;0<\phi<\pi/2)$, with $\theta$ measured from $+x$ and $\phi$ in the $yz$ (equatorial) plane from $+y$ toward $+z$}

        \label{fig:1a}
    \end{subfigure}
    \hfill
    \begin{subfigure}[b]{0.4\textwidth}
        \includegraphics[width=\textwidth]{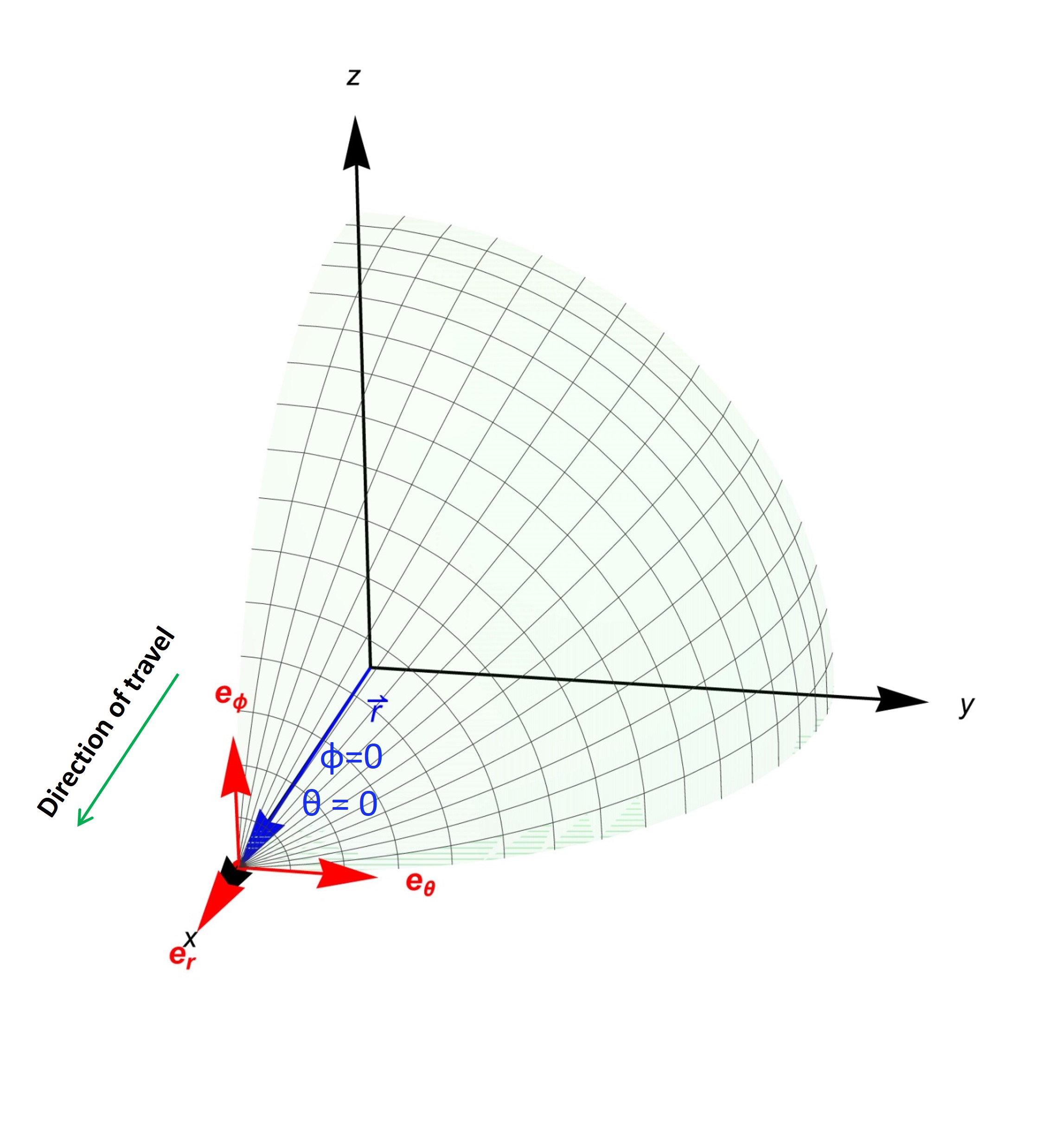}
        \caption{$\mathit{P}$ at $(r=\rho,\,\theta=0,\,\phi=0)$}
        \label{fig:1b}
    \end{subfigure}
    \hfill
    \begin{subfigure}[b]{0.4\textwidth}
        \includegraphics[width=\textwidth]{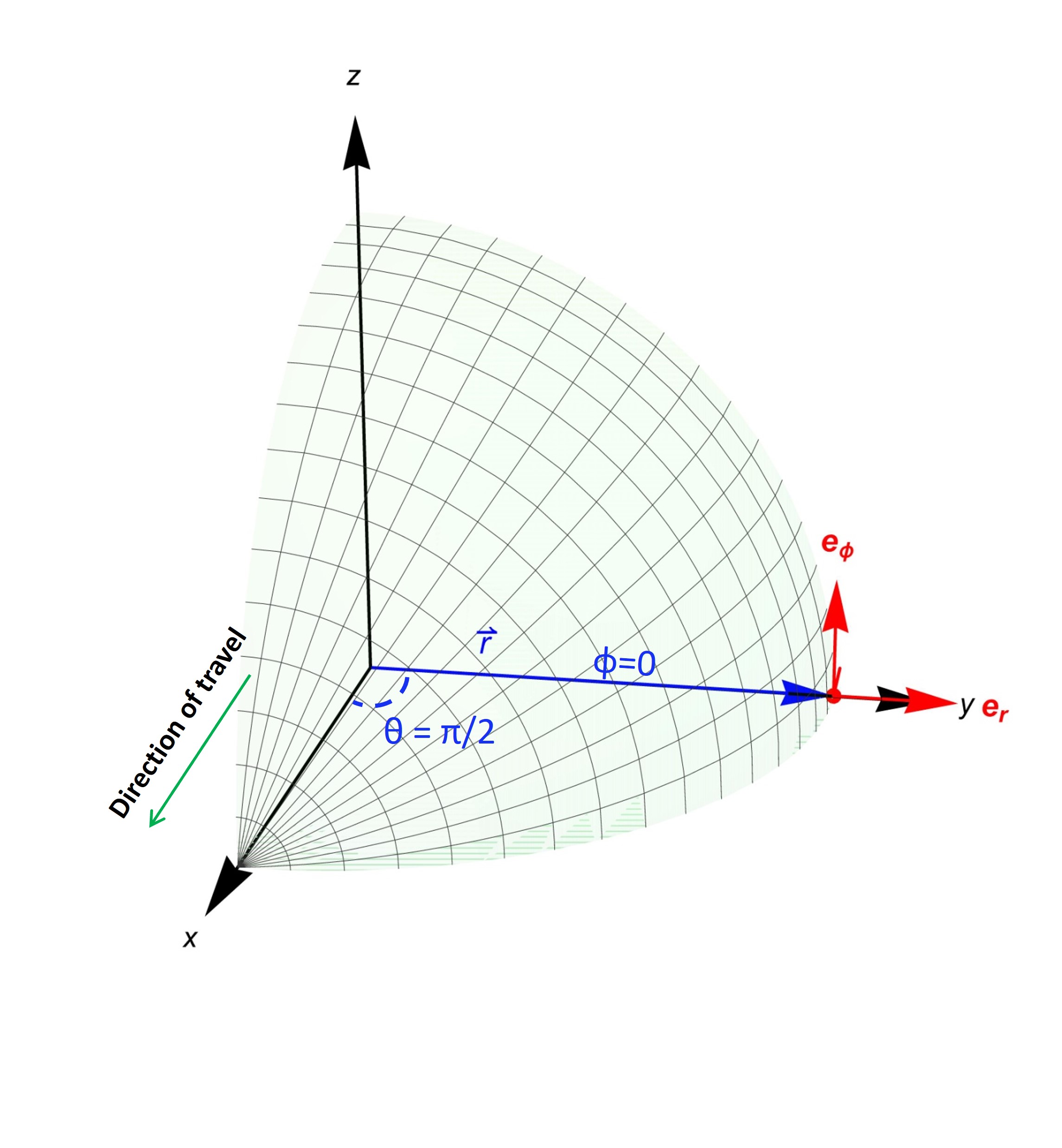}
        \caption{$\mathit{P}$ at $(r=\rho,\,\theta=\pi/2,\,\phi=0)$}
        \label{fig:1c}
    \end{subfigure}
    \caption{Natário's spherical coordinate system $(r, \theta, \phi)$ for the warp drive moving in the $x$ direction.}
    \label{Fig1_SpherCoord}
\end{figure}

\section{Timelike eigenvalue definition of proper energy density and rest frame}\label{BBsecEnergyDensity}

\noindent\emph{Quick map.} For a compact signature/tensor-form crosswalk see Table~\ref{tab:compact_summary}; for a side‐by‐side comparison of algebraic features relevant here (diagonalizability, eigenvalue content, existence of a real rest frame, and characteristic energy‐condition behavior), see Table~\ref{tab:type_comparison}.

\medskip

\noindent\textbf{Convention and use.}
We adopt signature $(-,+,+,+)$ and define the proper energy density via the timelike eigenvalue of the mixed stress--energy tensor $T^\mu{}_\nu$ in its principal orthonormal frame \cite{Lichnerowicz_1967,synge1966,PlebanskiKrasinski,LandauLifshitz,hawking_ellis_1973}:
\begin{equation}
T^{(\mu)}{}_{(\nu)} \;=\; \mathrm{diag}\!\big(-\varrho_p,\; p_1,\; p_2,\; p_3\big).
\label{eq:principal_frame}
\end{equation}
For comparison with prior literature one may also state the standard eigenvalue problem in the Lichnerowicz \cite{Lichnerowicz_1967} convention (often written for $(+,-,-,-)$):
\begin{equation}
T^\mu{}_\nu\, u^\nu \;=\; \varrho_p\, u^\mu ,
\label{LichstdEigProb}
\end{equation}
while in our adopted $(-,+,+,+)$ signature this reads
\begin{equation}
T^\mu{}_\nu\, u^\nu \;=\; -\,\varrho_p\, u^\mu .
\label{stdEigProb}
\end{equation}
Equivalently, the covariant non-standard form becomes the generalized eigenproblem
\begin{equation}
T_{\mu\nu}\,u^\nu \;=\; -\,\varrho_p\, g_{\mu\nu}\,u^\nu .
\label{genEigProb}
\end{equation}
When a real timelike eigenvector exists the tensor is Hawking–Ellis Type~I and admits a unique rest frame; in Type~IV a complex conjugate pair appears and no real rest frame exists. In the principal frame the null and weak energy conditions take the standard forms (for all $i$):
\begin{equation}
\mathrm{NEC:}\qquad \varrho_p + p_i \;\ge\; 0,
\label{eq:NEC}
\end{equation}
\begin{equation}
\mathrm{WEC:}\qquad \varrho_p \;\ge\; 0\quad \text{and}\quad \varrho_p + p_i \;\ge\; 0.
\label{eq:WEC}
\end{equation}

\begin{table}[htbp]
\centering
\begin{tabular}{|c|c|c|c|}
\hline
\thead{Tensor Form} & \thead{Signature} & \thead{Proper Energy \\ Density ($\varrho_p$)} & \thead{Principal \\ Pressures ($p_i$)} \\
\hline\hline
Mixed ($T^{(\mu)}{}_{(\nu)}$) & $(-,+,+,+)$ & $-\;T^{(0)}{}_{(0)}$ & $T^{(i)}{}_{(i)}$ \\
\hline
Mixed ($T^{(\mu)}{}_{(\nu)}$) & $(+,-,-,-)$ & $T^{(0)}{}_{(0)}$ & $-\;T^{(i)}{}_{(i)}$ \\
\hline\hline
Covariant ($T_{(\hat{\mu})(\hat{\nu})}$) & $(-,+,+,+)$ & $T_{(\hat{0})(\hat{0})}$ & $T_{(\hat{i})(\hat{i})}$ \\
\hline
Covariant ($T_{(\hat{\mu})(\hat{\nu})}$) & $(+,-,-,-)$ & $T_{(\hat{0})(\hat{0})}$ & $T_{(\hat{i})(\hat{i})}$ \\
\hline\hline
Contravariant ($T^{(\hat{\mu})(\hat{\nu})}$) & $(-,+,+,+)$ & $T^{(\hat{0})(\hat{0})}$ & $T^{(\hat{i})(\hat{i})}$ \\
\hline
Contravariant ($T^{(\hat{\mu})(\hat{\nu})}$) & $(+,-,-,-)$ & $T^{(\hat{0})(\hat{0})}$ & $T^{(\hat{i})(\hat{i})}$ \\
\hline
\end{tabular}
\caption{Proper energy density \(\varrho_p\) and principal pressures \(p_i\) in the principal orthonormal frame, showing their relation to diagonal tensor components for various tensor forms and metric signatures (\emph{assuming real eigenvalues exist}).}
\label{tab:compact_summary}
\end{table}

\begin{table}[htbp]
\caption{Algebraic Comparison of Hawking–Ellis Type I and Type IV Stress-Energy Tensors.}
\label{tab:type_comparison}
\centering
\begin{tabular*}{\textwidth}{@{\extracolsep{\fill}} | l | p{0.35\textwidth} | p{0.35\textwidth} |} 
\hline
\textbf{Property} & \textbf{Type I} (e.g., Classical Fluids \newline and Non-null Fields) & \textbf{Type IV} (e.g., Alcubierre \& \newline Natário-Zero-Exp. Warp Drives) \\
\hline \hline
\begin{tabular}[t]{@{}l@{}}\textbf{Eigenvalues}\\ \textbf{of $T^\mu\,_\nu$}\end{tabular} & 
ONE real timelike ($-\varrho_p$),\newline THREE real spacelike ($p_i$) & 
ONE pair complex conjugate, \newline TWO real spacelike
\\ \hline
\begin{tabular}[t]{@{}l@{}}\textbf{Diagonaliz.}\\ \textbf{\& Rest}\\ \textbf{Frame}\end{tabular} &
Diagonalizable over $\mathbb{R}$. \newline Admits a unique, real orthonormal tetrad frame (rest frame) where $T^{\hat{\mu}}\,_{\hat{\nu}} = \text{diag}(-\varrho_p, p_1, p_2, p_3)$. &
Diagonalizable only over $\mathbb{C}$. \newline Does not admit a real orthonormal tetrad rest frame; eigenvectors associated with complex eigenvalues are complex conjugates.
\\ \hline
\begin{tabular}[t]{@{}l@{}}\textbf{Physical}\\ \textbf{Nature}\\ \textbf{\& Energy}\\ \textbf{Conditions}\end{tabular} &
Typically represents classical matter/energy distributions. \newline Can satisfy standard energy conditions (e.g., NEC, WEC). &
Represents non-classical, \newline exotic distributions. \newline Characteristically violates all\newline standard energy conditions (e.g.,\newline NEC, WEC)~\cite{Martin-MorunoBook, Maeda2020}.\newline  Requires non-static spacetime \newline in $f(R)$ gravity~\cite{Maeda2021}. \\
\hline
\end{tabular*}
\end{table}

\section{Warp-drive spacetimes in a spherical coordinate tetrad framework}\label{secWarpDrive}

In this section, we systematically analyze warp-drive spacetimes within a spherical coordinate orthonormal tetrad framework, providing a clear and physically insightful decomposition of their kinematic and dynamic properties. We begin by discussing the generalized warp-drive spacetime, establishing a unified mathematical setting suitable for encompassing various specific warp-drive geometries and facilitating comparisons among them. Subsequently, we focus on the well-known Alcubierre warp drive, explicitly recasting its geometry and stress-energy content into spherical orthonormal tetrad components, thereby clarifying its structure and physical interpretation. We then explore Natário's kinematically divergenceless (zero expansion) warp drive, highlighting its distinct geometric features. Lastly, we address kinematically irrotational warp drives, examining the rationale for considering vorticity-free conditions in warp-drive models, and analyzing their physical consequences within our tetrad formalism.  For a compact side-by-side comparison of the three families we study, including defining kinematic constraints and shift components, see Table~\ref{tab:focused_comparison_tabularstar}.

\subsection{The generalized warp drive}\label{secGeneralized}

In this subsection, we introduce Natário's generalized warp-drive spacetime within the 3+1 ADM decomposition framework. We establish the mathematical foundations necessary to describe arbitrary warp-drive configurations, emphasizing the role of the lapse and shift functions in characterizing the geometry and evolution of spacetime. We then define the line element of the generalized warp drive, clearly delineating the distinction between coordinate and tetrad-based descriptions. This sets the stage for comparing specialized warp-drive models, such as Alcubierre's original spacetime, Natário's zero-expansion configuration, and a new irrotational warp drive which will be discussed in subsequent subsections.

Natário \cite{Natário} characterized a generalized warp-drive spacetime, whose geometry is associated with a vector field $\mathbf{X}$, using the 3+1 Arnowitt-Deser-Misner (ADM) formalism (Appendix~\ref{Notation}) \cite{Misner1973,Gourgoulhon,AlcubierreBook,baumgarte_shapiro_2010}. In this framework, spacetime is foliated by a sequence of spacelike hypersurfaces $\Sigma_t$. The \textit{Eulerian observers} are defined by the condition that their 4-velocity $n^\alpha$ is the future-pointing unit vector field normal to these hypersurfaces. The vector field $\mathbf{X}$ generating the warp effect is identified with the ADM shift vector $\boldsymbol{\beta}$, as discussed below \eqref{eq:shift_vector}. A key consideration is the causal structure: while such spacetimes may satisfy global hyperbolicity for subliminal speeds, superluminal warp drives generally exhibit Cauchy horizons \cite{Finazzi}, restricting the \textit{global} applicability of the 3+1 decomposition, although it typically remains valid \textit{locally}.

We use the standard symbols $\alpha$ (lapse) and $\boldsymbol{\beta}$ (shift), but adopt the \emph{minus-sign} shift convention in the line element, i.e.\ $(dx^{i}-\beta^{i}dt)$ rather than $(dx^{i}+\beta^{i}dt)$. Relative to Gourgoulhon \cite{Gourgoulhon} and Alcubierre \cite{AlcubierreBook}, this corresponds to $\beta^{\text{(here)}}_{\,i}=-\,\beta^{\text{(there)}}_{\,i}$. For $\alpha=1$ the slice normal has $n_\mu=(-1,0,0,0)$ and $n^\mu=g^{\mu\nu}n_\nu=(1,\beta^i)$ (since $g^{0i}=-\beta^i$ in this convention), and the time vector decomposes as $\partial_t=\alpha n-\beta^{i}\partial_{i}=n-\beta^{i}\partial_{i}$. (Misner et al. \cite{Misner1973} use $N$ and $N^{i}$ for lapse and shift.)

Following Natário \cite{Natário} and common practice for the warp-drive models considered here, we set the lapse function $\alpha=1$, simplifying the relation between coordinate time $t$ and proper time $\tau$ along the normal direction ($d\tau = \alpha \, dt = dt$). 

The shift vector $\boldsymbol{\beta}$ measures how spatial coordinates are displaced along $\boldsymbol{\beta}$ between adjacent hypersurfaces $\Sigma_t$ relative to the normal vector~\cite{Gourgoulhon}. The sign of its components determines the orientation (e.g., positive or negative $x$-direction) of this coordinate shift, representing a local velocity field relative to Eulerian observers at rest in the slicing. Following convention~\cite{SantiagoVisser}, we relate the generating vector field $\mathbf{X}$ directly to the shift vector as
\begin{equation}
    \boldsymbol{\beta}=\mathbf{X} .
    \label{eq:shift_vector}
\end{equation}

This choice corresponds to a \emph{flow} interpretation. The common negative sign convention $\beta^i = X^i = -v^i$ often seen in the literature~\cite{SantiagoVisser} typically relates the shift vector $\beta^i$ (local coordinate flow velocity relative to Eulerian observers) to the bubble's coordinate velocity $v^i$ as measured by distant observers, with the sign arising from the model's specific dynamics.

With these choices ($\alpha=1$, $\gamma_{ij}=\delta_{ij}$ for flat spatial slices in Cartesian coordinates), the line element with signature $(-,+,+,+)$ for a \textit{generalized warp-drive spacetime} takes the form \cite{Natário}:
\begin{equation}
\label{eq:lineelementGeneral}
\begin{aligned}
    ds^2 &= -dt^2 + \delta_{ij}\bigl(dx^i - \beta^{i} dt\bigr)\bigl(dx^j - \beta^{j} dt\bigr) \\
         &= -dt^2 + \bigl(dx - \beta^{x} dt\bigr)^2 + \bigl(dy - \beta^{y} dt\bigr)^2 + \bigl(dz - \beta^{z} dt\bigr)^2 .
\end{aligned}
\end{equation}
where $(x^1, x^2, x^3) = (x,y,z)$ are orthogonal Cartesian coordinates, and $\beta^i = (\beta^{x},\beta^{y},\beta^{z})$ are the corresponding Cartesian coordinate components of the shift vector field $\boldsymbol{\beta}$. 

\medskip

To facilitate physical interpretation and analysis, particularly in scenarios with spherical symmetry or specific directional properties, it is advantageous to work with local orthonormal frames (tetrads or vierbein), wherein the metric components are locally Minkowski ($\eta_{\hat{\alpha}\hat{\beta}}$) \cite{Misner1973, Wald, deFelice, PlebanskiKrasinski, Chandrasekhar}. In spherical coordinates $(r, \theta, \phi)$ (defined in Appendix~\ref{Notation}), the spatial unit orthonormal basis vectors $\boldsymbol{e}_{\hat{i}}$ are related to the coordinate basis vectors $\boldsymbol{g}_{i} = \partial/\partial x^i$ by:
\begin{align}
\boldsymbol{e}_{\hat{r}} \equiv \frac{\partial}{\partial r}=\boldsymbol{g}_{r}\; , \;
\boldsymbol{e}_{\hat{\theta}} \equiv \frac{1}{r}\, \frac{\partial}{\partial \theta}= \frac{1}{r}\, \boldsymbol{g}_{\theta}\; , \;
\boldsymbol{e}_{\hat{\phi}} \equiv \frac{1}{r\sin\theta}\, \frac{\partial}{\partial \phi}= \frac{1}{r\sin\theta}\,\boldsymbol{g}_{\phi}\, . \label{eq:tetrad_basis_def} 
\end{align}

Relative to this basis, the shift vector $\boldsymbol{\beta}$ is then expressed as:
\begin{equation}
\boldsymbol{\beta} = \beta^{\hat{r}}\, \boldsymbol{e}_{\hat{r}} + \beta^{\hat{\theta}}\, \boldsymbol{e}_{\hat{\theta}} + \beta^{\hat{\phi}}\, \boldsymbol{e}_{\hat{\phi}} \, ,\label{beta3comp}
\end{equation}
where $\beta^{\hat{i}}$ are the tetrad (physical) components.

We adopt the convention of using carets (hats) for tetrad indices \cite{Misner1973}. The corresponding dual basis one-forms are $\boldsymbol{e}^{\hat{r}} = \dd r$, $\boldsymbol{e}^{\hat{\theta}} = r\,\dd \theta$, $\boldsymbol{e}^{\hat{\phi}} = r \sin\theta\,\dd \phi$. While the radial basis one-form $\boldsymbol{e}^{\hat{r}}$ is exact, the angular one-forms are not closed ($\dd \boldsymbol{e}^{\hat{\theta}} \neq 0, \dd \boldsymbol{e}^{\hat{\phi}} \neq 0$), reflecting the fact that this standard spherical orthonormal basis is non-holonomic, meaning it cannot be integrated to define a coordinate chart on the differentiable manifold. Equivalently, the basis vector fields do not commute, as shown by the Lie bracket: e.g., $[\boldsymbol{e}_{\hat{\theta}}, \boldsymbol{e}_{\hat{\phi}}] = -\frac{\cot\theta}{r}\, \boldsymbol{e}_{\hat{\phi}} \neq 0$ (transporting a vector around an infinitesimal loop in this plane results in a net rotation of the vector).

The components of the shift vector $\boldsymbol{\beta}$ in the coordinate bases $\boldsymbol{g}_{i}$ and orthonormal bases $\boldsymbol{e}_{\hat{i}}$ are related via the definitions in \eqref{eq:tetrad_basis_def}. For instance, $\beta^{\hat{r}} = \beta^r$, $\beta^{\hat{\theta}} = r \beta^\theta$, $\beta^{\hat{\phi}} = r \sin\theta \beta^\phi$. Since the spatial part of the tetrad metric is $\delta_{\hat{i}\hat{j}}$, the distinction between covariant and contravariant spatial tetrad components is trivial ($\beta_{\hat{i}} = \delta_{\hat{i}\hat{j}} \beta^{\hat{j}} = \beta^{\hat{i}}$).

Tetrad orthonormal basis components represent the measurements an observer in free fall (subject only to gravity) would record—that is, the \emph{physical components} in a local Minkowski frame. No single inertial frame can globally encompass the effects of spacetime curvature; however, in a sufficiently small region (a local Lorentz frame), the equivalence principle ensures that Special Relativity applies.

Moreover, although tetrad components are the components of tensors relative to a local orthonormal basis, they do not transform under general coordinate transformations in the same way as coordinate basis components. Instead, when the tetrad basis is changed, the components transform via local Lorentz transformations.

Finally, the tetrad formalism employs the Ricci rotation coefficients (or spin connection), which encode how the tetrad fields rotate relative to the Levi–Civita connection as one moves through spacetime. These coefficients reflect the geometric obstruction to defining a global inertial frame in a curved spacetime.

\subsection{Alcubierre's warp drive in a spherical coordinate orthonormal tetrad}
\label{sec:Alcubierre}

Alcubierre’s solution \cite{Alcubierre_1994} was originally written in
Cartesian $(t,x,y,z)$ coordinates.  Because the geometry is axially symmetric
about the $+x$ axis (direction of travel), it is convenient to adopt
the spherical chart\,
$(t,r,\theta,\phi)$ shown in Fig.~\ref{Fig1_SpherCoord} and to pass to an
orthonormal tetrad
$
  \{\boldsymbol{e}_{\hat{0}},\boldsymbol{e}_{\hat{r}},
    \boldsymbol{e}_{\hat{\theta}},\boldsymbol{e}_{\hat{\phi}}\},
$
with dual one-forms
$
  \{\boldsymbol{e}^{\hat{0}},\boldsymbol{e}^{\hat{r}},
    \boldsymbol{e}^{\hat{\theta}},\boldsymbol{e}^{\hat{\phi}}\}.
$
We work in geometric units with \(c=1\) unless stated otherwise.  We leave Einstein's coupling constant \(\kappa\) explicit.

\paragraph*{Shift vector}

Rotational symmetry implies $\beta^{\hat\phi}=0$; hence
\begin{equation}
  \boldsymbol{\beta}
  =\beta^{\hat r}\,\boldsymbol{e}_{\hat r}
  +\beta^{\hat\theta}\,\boldsymbol{e}_{\hat\theta} \, .
  \label{eq:beta2comp}
\end{equation}
Following Alcubierre we choose the spatial covector
$
  \boldsymbol{\beta}^{\flat}
  =v(t)\,f_{\text{Alc}}(r)\,\dd x
$
with $x=r\cos\theta$.  Using
$\dd x=\cos\theta\,\dd r-r\sin\theta\,\dd\theta$
and the spherical dual basis
$\boldsymbol{e}^{\hat r}=\dd r$, $\boldsymbol{e}^{\hat\theta}=r\,\dd\theta$,
one finds
\begin{equation}
  \boldsymbol{\beta}^{\flat}
  =v(t)\,f_{\text{Alc}}(r)
  \bigl(\cos\theta\,\boldsymbol{e}^{\hat r}
        -\sin\theta\,\boldsymbol{e}^{\hat\theta}\bigr).
  \label{eq:betaflatCoordDiff}
\end{equation}

Because the triad $\{\boldsymbol{e}_{\hat r},\boldsymbol{e}_{\hat\theta},
\boldsymbol{e}_{\hat\phi}\}$ is \emph{spatial} and orthonormal, raising
indices changes only the sign of the time component; therefore
$\beta^{\hat i}=\beta_{\hat i}$ for $i=r,\theta,\phi$:
\begin{equation}
\begin{pmatrix}
\beta^{\hat r}\\[2pt]
\beta^{\hat\theta}\\[2pt]
\beta^{\hat\phi}
\end{pmatrix}
=v(t)\,f_{\text{Alc}}(r)
\begin{pmatrix}
\cos\theta\\[4pt]
-\sin\theta\\[4pt]
0
\end{pmatrix}.
\label{eq:BetaAlc}
\end{equation}

\paragraph*{Form function}

The smooth \emph{top-hat} function
\begin{equation}
  f_{\text{Alc}}(r)
  =\frac{\tanh[\sigma(r+\rho)]-\tanh[\sigma(r-\rho)]}
         {2\,\tanh(\sigma\rho)}
  \label{eq:fAlc} ,
\end{equation}
obeys \(f_{\text{Alc}}(0)=1\) and
$\lim_{r\to\infty}f_{\text{Alc}}(r)=0$.
Its transition layer of width $\simeq \pm 3/\sigma$ is centered on $r\simeq\rho$.
Here $\rho$ is the warp bubble radius and $\sigma^{-1}$ its wall thickness.

\paragraph*{Boundary behaviour (summary)}
Inside the bubble the shift approaches $+v(t)$ along $+x$; far away it decays to zero. Detailed small/large--$r$ limits at $\theta=0$ and $\theta=\pi/2$ are listed in Appendix~\ref{app:boundary-details}.

\subsection{Natário’s kinematically divergence-free warp drive}
\label{sec:Natario}

Natário \cite{Natário} sought a warp-drive spacetime with \emph{vanishing spatial
expansion}. This coincides with zero kinematic divergence $\operatorname{div}\boldsymbol{\beta}=0$ under unit lapse, flat time‑independent 3+1 slicing.  Working on each
constant-time slice $(\mathbb R^{3},\gamma)$ one may write
\begin{equation}
  \operatorname{div}\boldsymbol{\beta}
  =\frac1{\sqrt{\gamma}}\partial_i\!\bigl(\sqrt{\gamma}\,\beta^{i}\bigr)
  =-\delta\boldsymbol{\beta}^{\flat}=0,
\label{eq:NatKinDiv}
\end{equation}
where $\delta$ is the three–dimensional exterior codifferential.
Hence $\boldsymbol{\beta}^{\flat}$ must be \emph{co-exact}:
\begin{equation}
  \boldsymbol{\beta}^{\flat}=*
  \dd\boldsymbol{A}\qquad(\exists\,\text{1-form } \boldsymbol{A}).
\end{equation}
Choosing
\(
  \boldsymbol{A}= -\tfrac12\,f(r)\,r^{2}\sin^{2}\theta\,\dd\phi
\)
and multiplying by an arbitrary bubble velocity $v(t)$ gives
\begin{equation}
  \boldsymbol{\beta}^{\flat}
  =-v(t)\,*\dd\!
     \bigl[f(r)\,\tfrac12 r^{2}\sin^{2}\theta\,\dd\phi\bigr].
  \label{eq:beta_flat_def_natario}
\end{equation}

A short calculation produces covariant components
\begin{align}
  \beta_{\hat r}&=-v(t)\,f(r)\cos\theta,\\
  \beta_{\hat\theta}&=v(t)\Bigl(f(r)+\tfrac r2 f'(r)\Bigr)\sin\theta,\\
  \beta_{\hat\phi}&=0,
\end{align}
so that in our spatial orthonormal triad
\begin{equation}
\!\!
\begin{pmatrix}
\beta^{\hat r}\\[2pt]
\beta^{\hat\theta}\\[2pt]
\beta^{\hat\phi}
\end{pmatrix}
=
\begin{pmatrix}
-\,v(t)\,f(r)\cos\theta\\[4pt]
\;v(t)\left(f(r)+\tfrac r2 f'(r)\right)\sin\theta\\[6pt]
0
\end{pmatrix}.
\label{eq:NatarioShiftComponents}
\end{equation}

\paragraph*{Form function for Natário}

To facilitate comparison with Alcubierre we define
\begin{equation}
  f(r):=1-f_{\text{Alc}}(r),
  \label{eq:fNatario}
\end{equation}
so that $f(0)=0$ and $\lim_{r\to\infty}f(r)=1$.  With this choice the
Eulerian observer at the bubble center ($r=0$) sees distant matter move
with speed $-v(t)$, opposite to Alcubierre’s convention.

\paragraph*{Boundary behaviour (summary)}
The center of the bubble is at rest while the distant universe drifts with velocity $-v(t)$, consistent with Nat\'ario’s convention. Detailed limits appear in Appendix~\ref{app:boundary-details}.

\paragraph*{Kinematic constraint and global energy conditions}
Natário’s divergence–free ansatz
($K\equiv\gamma^{ij}K_{ij}=0$, equivalently
$\operatorname{div}_{\Sigma_t}\boldsymbol{\beta}=0$ for unit lapse and
time‑independent spatial metric)
renders every time slice $\Sigma_t$ \emph{instantaneously} isochoric.
Yet if the driving function $v(t)$ varies in time the resulting shift
field $\beta^i(t,x^{k})$ makes the spacetime \emph{non‑stationary}.
Although $K=0$ on each slice, the individual components
\begin{equation}
K_{ij}\;=\;\tfrac12\!\bigl(D_i\beta_j+D_j\beta_i\bigr)
\label{eq:extrinsic_curvature}
\end{equation}
acquire explicit $t$–dependence, and so do the curvature scalars that
enter the mixed Einstein tensor
$G^{\mu}{}_{\nu}$.

The averaged null energy condition (ANEC) probes the \emph{integrated}
null projection,
\begin{equation}
    \mathcal{I}_{\mathrm{ANEC}}
    \;=\;
    \int_\gamma
        T_{\mu\nu}k^\mu k^\nu\,
        \mathrm{d}\lambda
    =
    \frac1{\kappa}
    \int_\gamma
        G_{\mu\nu}k^\mu k^\nu\,
        \mathrm{d}\lambda ,
\label{eq:ANEC_integral}
\end{equation}
taken along complete null geodesics~$k^\mu$.
Because such geodesics thread a sequence of slices that may evolve in time,
a quantity that vanishes \emph{instantaneously} on each
$\Sigma_t$—for example the local expansion $K$—may still yield a non-zero
contribution to the integral, with sign depending on the detailed evolution.
Thus, imposing $K=0$ on each slice does not by itself establish compliance with
averaged energy conditions. In practice, one should assess the full
four-dimensional evolution of $\bigl(\alpha,\beta^i,\gamma_{ij}\bigr)$ along the
null geodesics when evaluating ANEC (and related \emph{averaged} conditions).

\smallskip

\begin{proposition}[Trace--energy condition for the Nat\'ario warp drive]
\label{prop:TEC_Natario}
\textbf{Assumptions:} $\alpha\equiv 1$, flat \emph{and static} spatial slice ($R^{(3)}=0$), vanishing expansion $K=0$, and $\Lie_{n}K=0$ \emph{($K=0$ fixes only the slice value, whereas $\Lie_{n}K$ is its normal–time derivative and can be nonzero).}

Then the Einstein--tensor trace satisfies $G\le 0$ everywhere; equivalently, the trace--energy condition $T\le 0$ holds throughout the spacetime.
\end{proposition}

\begin{proof}
Work in a $3{+}1$ decomposition with $(-,+,+,+)$ signature, $\gamma_{ij}\simeq\mathbb{R}^{3}$ and $K_{ij}=-\tfrac12\Lie_{n}\gamma_{ij}$.
The Gauss--Codazzi scalar identity gives
\begin{equation}
{}^{(4)}R
= {}^{(3)}R + K^{2} + K_{ij}K^{ij}
-\frac{2}{\alpha}\,\Lie_{n}K
-\frac{2}{\alpha}\,D_iD^{\,i}\alpha .
\label{eq:GaussCodazziScalar}
\end{equation}
With $R^{(3)}=0$, $K=0$, $D_i\alpha=0$ and $\Lie_{n}K=0$, we obtain $R^{(4)}=K_{ij}K^{ij}\ge0$.
Since $G=-R^{(4)}$, one has $G\le0$, and thus $T=G/\kappa\le0$.
\end{proof}

\subsection{Kinematically irrotational warp drives}
\label{sec:Irrotational}

We now construct a warp-drive spacetime whose \emph{shift} is irrotational in the
\((3+1)\) sense, i.e.\ whose spatial vorticity tensor $\omega_{ij}
  :=\tfrac12\bigl(D_{j}\beta_{i}-D_{i}\beta_{j}\bigr)$ vanishes on every constant-time slice.  The discussion is organized in
two parts: first a justification for adopting a \emph{kinematic} (purely
geometric) definition of vorticity, in preference to fluid-dynamical
potential-flow ideas; second the actual derivation of an irrotational
shift vector.

\subsubsection{Justification for a kinematic–vorticity approach}

In warp geometries the “flow’’ is the geometric shift $\boldsymbol{\beta}$, not a material four–velocity, so fluid–enthalpy analogies do not apply. Under our slice assumptions (unit lapse $\alpha\equiv1$, flat, time-independent spatial slices), imposing $\dd\boldsymbol{\beta}^{\flat}=0$ (i.e., $\boldsymbol{\beta}^{\flat}=-\dd\Phi$) sets $G_{\hat0\hat i}=0$ (see Eq.~(47)) and block-diagonalizes $G^{\hat\mu}{}_{\ \hat\nu}$ in the Eulerian tetrad, yielding a unique timelike eigenvector and Hawking–Ellis Type~I everywhere for our construction (Prop.~\ref{prop:irrot_class}). This algebraic mechanism—eliminating momentum density and the associated Type~IV patches—is the reason to enforce \emph{kinematic} irrotationality.

A fixed–smoothing ablation that adds only a vortical component (with $g(r)$, masks, and tolerances held fixed) shows the causal role: the global budgets degrade sharply (e.g., $E_-$ grows by $>500\times$ at $\eta=1$), proving the improved NEC/WEC profile is due to curl-free kinematics rather than profile smoothing (Appendix~\ref{app:ablation}, Table~\ref{tab:abl_sweep}).

We deliberately \emph{do not} also impose Natário’s $\operatorname{div}\boldsymbol{\beta}=0$ (Sec.~\ref{sec:Natario}). Enforcing both $\dd\boldsymbol{\beta}^{\flat}=0$ and $\operatorname{div}\boldsymbol{\beta}=0$ would force a harmonic potential (\(\nabla^2\Phi=0\)), severely restricting localized bubbles; allowing $\operatorname{div}\boldsymbol{\beta}\neq0$ enlarges the design space and, in our numerics, reduces negative-energy magnitudes.

In summary, within these slice assumptions, irrotationality is a \emph{sufficient} and conceptually clean algebraic constraint for achieving global Type~I; we do not claim it is necessary in general.

\subsubsection{Kinematically irrotational warp drive}\label{secIrrotational}

In this section, we construct an irrotational warp-drive spacetime by defining a velocity field using an orthonormal tetrad framework in spherical coordinates. The tetrad components are chosen to exploit the continuous rotational symmetry around the $x$ direction of travel, following a similar approach to that used by Natário in his zero-expansion warp-drive model~\cite{Natário}.

In General Relativity, the vorticity characterizes the local rotation of the congruence of timelike worldlines (paths of observers) in spacetime. On a 4-D Lorentzian pseudo-Riemannian manifold \((M^4,g)\), the spacetime vorticity tensor $\omega^{(4)}_{\mu\nu}$ characterizes dynamic rotation associated with a 4-velocity field $u^\mu$. In the 3+1 decomposition, we focus on the kinematic Euclidean (\(\mathbb{R}^3\),$\gamma$) spatial vorticity tensor on each slice, defined using the spatial covariant derivative $D_i$ acting on the shift vector $\boldsymbol{\beta}$:
\begin{equation}
\omega_{ij} = D_{[i} \beta_{j]} = \frac{1}{2} \left( D_i \beta_j - D_j \beta_i \right). \label{eq:vorticity_tensor} 
\end{equation}

The \emph{physical spacetime vorticity} of a timelike congruence with 4–velocity \(u^\mu\) is the spatially projected antisymmetric part
$\omega^{(4)}_{\mu\nu} \;=\; h_{\mu}{}^{\alpha} h_{\nu}{}^{\beta}\,\nabla_{[\alpha}u_{\beta]}, 
\;
h_{\mu\nu}=g_{\mu\nu}+u_\mu u_\nu.$
In exterior–calculus form this reads
$
\omega^{(4)} =(1/2)(\dd_{4} u^\flat \;-\; u^\flat \wedge a^\flat) ,
\;
a^\flat:= (u\!\cdot\!\nabla)u_\mu\,\dd x^\mu,$.

By contrast, the \emph{kinematic (purely spatial) vorticity} of the shift on each slice \(\Sigma_t\) is
\begin{equation}
\omega \;=\; \dd \boldsymbol{\beta}^\flat,
\label{eq:vorticity_form}
\end{equation}
where \(\dd\) is the 3D exterior derivative. Our irrotationality condition is \(\omega=0\), which (on simply connected slices) implies \(\boldsymbol{\beta}^\flat=-\dd\Phi\). This constraint involves only spatial derivatives; therefore \(\omega=0\) does \emph{not} imply \(\omega^{(4)}=0\) if \(\boldsymbol{\beta}\) depends on time, because the acceleration/tilt term \(u^\flat\wedge a^\flat\) can remain nonzero.

In the \( 3+1 \) decomposition formalism with lapse function \( \alpha = 1 \), the coordinate time \( t \) measures proper time along the worldlines of observers moving orthogonally to the spatial hypersurfaces (normal observers). Consequently, the lapse function does not contribute to the vorticity, and the dynamics of the flow are entirely encoded in the behavior of the shift vector \( \boldsymbol{\beta} \).

Computing the spatial exterior derivative $\dd \boldsymbol{\beta}^{\flat}$:
\begin{align}
\omega &= \left( \dd \beta_{\hat{r}} \wedge \boldsymbol{e}^{\hat{r}} \right) + \left( \dd \beta_{\hat{\theta}} \wedge \boldsymbol{e}^{\hat{\theta}} + \beta_{\hat{\theta}}\, \dd \boldsymbol{e}^{\hat{\theta}} \right) + \left( \dd \beta_{\hat{\phi}} \wedge \boldsymbol{e}^{\hat{\phi}} + \beta_{\hat{\phi}}\, \dd \boldsymbol{e}^{\hat{\phi}} \right).
\end{align}

The spatial exterior derivatives of the shift vector components \(\beta_{\hat{r}}, \beta_{\hat{\theta}}, \beta_{\hat{\phi}}\) in the tetrad (vierbein) spherical components are:
\begin{subequations}
\label{eq:dd_beta_all} 
\begin{align}
\dd \beta_{\hat{r}} &=\frac{\partial \beta_{\hat{r}}}{\partial r}\, \boldsymbol{e}^{\hat{r}} + \frac{1}{r}\, \frac{\partial \beta_{\hat{r}}}{\partial \theta}\, \boldsymbol{e}^{\hat{\theta}} + \frac{1}{r \sin\theta}\, \frac{\partial \beta_{\hat{r}}}{\partial \phi}\, \boldsymbol{e}^{\hat{\phi}}, \\[2ex]
\dd \beta_{\hat{\theta}} &=  \frac{\partial \beta_{\hat{\theta}}}{\partial r}\, \boldsymbol{e}^{\hat{r}} + \frac{1}{r}\, \frac{\partial \beta_{\hat{\theta}}}{\partial \theta}\, \boldsymbol{e}^{\hat{\theta}} + \frac{1}{r \sin\theta}\, \frac{\partial \beta_{\hat{\theta}}}{\partial \phi}\, \boldsymbol{e}^{\hat{\phi}}, \\[2ex]
\dd \beta_{\hat{\phi}} &= \frac{\partial \beta_{\hat{\phi}}}{\partial r}\, \boldsymbol{e}^{\hat{r}} + \frac{1}{r}\, \frac{\partial \beta_{\hat{\phi}}}{\partial \theta}\, \boldsymbol{e}^{\hat{\theta}} + \frac{1}{r \sin\theta}\, \frac{\partial \beta_{\hat{\phi}}}{\partial \phi}\, \boldsymbol{e}^{\hat{\phi}}.
\end{align}
\end{subequations}

Collecting all terms, we obtain the components of the spatial vorticity 2-form \(\omega\) in the tetrad basis:
\begin{equation}
\omega = \omega_{\hat{r}\hat{\theta}}\, \boldsymbol{e}^{\hat{r}} \wedge \boldsymbol{e}^{\hat{\theta}} + \omega_{\hat{r}\hat{\phi}}\, \boldsymbol{e}^{\hat{r}} \wedge \boldsymbol{e}^{\hat{\phi}} + \omega_{\hat{\theta}\hat{\phi}}\, \boldsymbol{e}^{\hat{\theta}} \wedge \boldsymbol{e}^{\hat{\phi}} .
\end{equation}

To obtain the spatial vorticity vector \(\boldsymbol{\omega}\), we apply the spatial Hodge star operator to the spatial vorticity 2-form $\omega$:
\begin{equation}
\boldsymbol{\omega} =* \omega= * (\dd \boldsymbol{\beta}^{\flat}). \label{vorticity_vector}
\end{equation}

(We use only the spatial (3D) exterior calculus on each slice $\Sigma_t$: $\dd$ denotes the spatial exterior derivative and $*$ the 3D Hodge star. When spacetime derivatives are mentioned, we write $\dd_4$ explicitly.)

Therefore the components of the spatial vorticity vector \(\omega^{\hat{k}}\) are given by:
\begin{equation}
\omega^{\hat{k}} =(*\omega)^{\hat{k}} = \frac{1}{2}\, \epsilon^{\hat{k}\hat{i}\hat{j}}\, \omega_{\hat{i}\hat{j}} \, ,
\end{equation}
where \(\epsilon^{\hat{r}\hat{\theta}\hat{\phi}} = +1\) is the Levi-Civita symbol in the tetrad basis, fully antisymmetric in its indices. Here we have expressed the result in terms of the contravariant components $\beta^{\hat{i}}$, using the numerical equality $\beta_{\hat{i}}=\beta^{\hat{i}}$ for spatial components in the orthonormal basis. Substituting the expressions for \(\omega_{\hat{i}\hat{j}}\), we obtain:
\begin{align}
\omega^{\hat{r}} &= \omega_{\hat{\theta}\hat{\phi}} = \frac{1}{r}\, \frac{\partial \beta^{\hat{\phi}}}{\partial \theta} - \frac{1}{r\sin\theta}\, \frac{\partial \beta^{\hat{\theta}}}{\partial \phi} + \frac{\cos\theta}{r\sin\theta}\, \beta^{\hat{\phi}}, \label{eqOmegaR} \\[2ex]
\omega^{\hat{\theta}} &= \omega_{\hat{\phi}\hat{r}} = \frac{1}{r\sin\theta}\, \frac{\partial \beta^{\hat{r}}}{\partial \phi} - \frac{\partial \beta^{\hat{\phi}}}{\partial r} - \frac{1}{r}\, \beta^{\hat{\phi}}, \label{eqOmegaTheta} \\[2ex]
\omega^{\hat{\phi}} &= \omega_{\hat{r}\hat{\theta}} = \frac{\partial \beta^{\hat{\theta}}}{\partial r} + \frac{\beta^{\hat{\theta}}}{r} - \frac{1}{r} \frac{\partial \beta^{\hat{r}}}{\partial \theta}. \label{eqOmegaPhi}
\end{align}

As expected, these tetrad components correspond to the standard physical vorticity components in spherical coordinates as used in fluid mechanics \cite{Verdiere1994, krishnamurthy2019}.

We seek an irrotational shift vector field $\boldsymbol{\beta}$ exhibiting continuous rotational symmetry about the $x$-axis (the axis of travel). This axial symmetry implies that the components of $\boldsymbol{\beta}$ in the spherical orthonormal basis must be independent of the azimuthal angle $\phi$ ($\partial/\partial\phi = 0$) and that the azimuthal component itself must vanish ($\beta^{\hat{\phi}}=0$). We therefore adopt the following ansatz for the tetrad components, separating variables in $(t, r, \theta)$:
\begin{equation}
\begin{pmatrix}
\beta^{\hat{r}} \\
\beta^{\hat{\theta}} \\
\beta^{\hat{\phi}}
\end{pmatrix}
=
\begin{pmatrix}
- v(t)\, f(r)\, \cos\theta \\
v(t)\, g(r)\, \sin\theta\\
0
\end{pmatrix} .
\label{eq:BetaIrr}
\end{equation}
Here, \(v(t)\) is the velocity along the \(x\)-axis, while $f(r)$ and $g(r)$ define the radial profile. The relative signs ensure consistency with motion along the $+x$ direction (as discussed for the Alcubierre case in Sec.~\ref{sec:Alcubierre}).

Substituting this axially symmetric form ($\partial/\partial\phi = 0$, $\beta^{\hat{\phi}}=0$) into the general expressions for the vorticity components \eqref{eqOmegaR}-\eqref{eqOmegaTheta}, we immediately find $\omega^{\hat{r}} = 0$ and $\omega^{\hat{\theta}} = 0$. The kinematic irrotationality condition, $\boldsymbol{\omega}=0$, thus reduces solely to requiring the azimuthal component $\omega^{\hat{\phi}}$ to vanish. From \eqref{eqOmegaPhi}, this condition becomes:
\begin{equation}
    \omega^{\hat{\phi}} = \frac{\partial \beta^{\hat{\theta}}}{\partial r} + \frac{\beta^{\hat{\theta}}}{r} - \frac{1}{r}\frac{\partial \beta^{\hat{r}}}{\partial \theta} = 0.
    \label{eq:vorticity_condition}
\end{equation}

 By Frobenius's theorem, this condition is equivalent to the spatial flow being surface-orthogonal \emph{on each slice}: since $\dd\beta^\flat=0$ on $\Sigma_t$, there exists a scalar potential $\Phi$ with $\beta^\flat=-\dd\Phi$, so $\boldsymbol{\beta}$ is orthogonal to the 2-surfaces $\Phi=\text{const}$ within the slice (this is distinct from spacetime vorticity). This requirement leads to a differential equation for the tetrad components \(\beta^{\hat{i}}\), which can be solved to construct the desired warp-drive spacetime.

Substituting \eqref{eq:BetaIrr} into \eqref{eq:vorticity_condition}, we obtain:

\begin{equation}
    v(t) \sin\theta \left(  \frac{\partial g(r)}{\partial r}  + \frac{g(r)}{r} - \frac{f(r)}{r} \right) = 0.
\end{equation}

Since the prefactor $v(t) \sin\theta \neq 0$ for $\theta \neq 0, \pi$; the irrotational condition $\omega^{\hat{\phi}} =0$ reduces to a linear first-order ordinary differential equation (ODE) $ \frac{\partial g(r)}{\partial r}  + \frac{g(r)}{r} = \frac{f(r)}{r}$ relating $g(r)$ and $f(r)$. 
This linear first-order ODE for $g(r)$ can be solved using the integrating factor $\exp(\int \frac{1}{r} \, dr) = r$. Multiplying by $r$ yields $r \frac{dg}{dr} + g(r) = f(r)$, or $\frac{d}{dr}(r g(r)) = f(r)$. Integrating with respect to $r$ gives $r g(r) = \int f(r) \, dr + C$, leading to:

\begin{equation}
    g(r) = \frac{1}{r} \left( \int f(r) \, dr + C \right). \label{eq:g integral}
\end{equation}

This expression provides $g(r)$ in terms of $f(r)$ and the integration constant $C$, to be determined from boundary conditions.

In the spatially divergence-free warp drive, the shift vector component in the polar angle direction involves the \textit{derivative} \eqref{eq:NatarioShiftComponents} of the form function $f(r)$ with respect to \(r\). Conversely, in the irrotational warp drive, this component involves the \textit{integral} \eqref{eq:g integral} of the form function $f(r)$ with respect to \(r\). The choice of $f(r)$ is guided by physical considerations, particularly the desired boundary conditions for the shift vector. Adopting the same convention as for the Natário drive regarding the observer perspective (Sec.~\ref{sec:Natario}, observer stationary at $r=0$, distant stars recede at $-v$), we require the form function $f(r)$ to satisfy $\lim_{r \to \infty} f(r) = 1$ and $\lim_{r \to 0} f(r) = 0$. As in Sec.~\ref{sec:Natario}, we achieve this using a definition based on Alcubierre's function \eqref{eq:fAlc}:
\begin{equation}
f(r) := 1 - f_{\text{Alc}}(r) . \label{fIrrotational}
\end{equation}

With this form function it follows from \eqref{eq:g integral} and \eqref{fIrrotational}, that evaluating the integral analytically yields $g(r)$ as follows:

\begin{equation}
\begin{aligned}
    g(r) &= \frac{ \operatorname{csch}\left( \dfrac{\rho \sigma}{2} \right) \operatorname{sech}\left( \dfrac{\rho \sigma}{2} \right) }{4\, r\, \sigma} \Bigg[ 2 r \sigma \sinh\left( \rho \sigma \right) \\
    &\quad + \cosh\left( \rho \sigma \right) \Bigg( \ln\Big( \cosh\left( \sigma (r - \rho) \right) \Big) - \ln\Big( \cosh\left( \sigma (r + \rho) \right) \Big) \Bigg) \Bigg] + \frac{C}{r}. \label{integral of f}
\end{aligned}
\end{equation}

The boundary condition $\lim_{r \to 0} g(r) = 0$ requires the integration constant $C$ to be set to zero, i.e., $C = 0$. This boundary condition determines $g(r)$ to be:

\begin{equation}
\begin{aligned}
    g(r) &= \frac{ 2 r \sigma \sinh\left( \rho \sigma \right)  + \cosh\left( \rho \sigma \right) \Bigg( \ln\Big(\frac{\cosh\left( \sigma (r - \rho) \right)}{\cosh\left( \sigma (r + \rho) \right)} \Big) \Bigg)  }{4\, r\, \sigma  \sinh\left( \dfrac{\rho \sigma}{2} \right) \cosh\left( \dfrac{\rho \sigma}{2} \right)} . \label{g(r)}
\end{aligned}
\end{equation}

\smallskip

This function \(g(r)\), given by \eqref{g(r)} with $C=0$, by construction ensures zero kinematic vorticity when used with \(f(r) = 1 - f_{\text{Alc}}(r)\). This $g(r)$ satisfies $\lim_{r\to 0} g(r)=0$ and $\lim_{r\to\infty} g(r)=1$. A short, rigorous proof of the $r\to 0$ behavior for $f(r)=1-f_{\rm Alc}(r)$ is provided in Appendix~\ref{app:gr0}. Consequently, the full shift vector defined by $f(r)$ and $g(r)$ in \eqref{eq:BetaIrr} satisfies the desired asymptotic conditions shown in the Nat\'ario rows of Table~\ref{tab:boundary-behaviour} (Appendix~\ref{app:boundary-details}, at $\theta=0$ and $\theta=\pi/2$). These boundary conditions are satisfied due to the properties of the hyperbolic functions and the chosen forms of \(f(r)\) and \(g(r)\) provided that: $\sigma > 0, \, \left( v(t),\, \rho \right) \in \mathbb{R}$.

\paragraph*{Scalar potential}
The kinematic irrotationality condition $\dd\boldsymbol{\beta}^\flat = 0$ ensures that the shift covector field can, on topologically simple spatial slices, be derived globally from a scalar potential $\Phi(r,\theta,t)$. Adopting the convention $\boldsymbol{\beta}^\flat = -\dd\Phi$ inherently satisfies the irrotationality condition, as $\dd\boldsymbol{\beta}^\flat = -\dd^2\Phi = 0$. For the specific potential consistent with the integration constant $C=0$ and the relationship $\frac{d}{dr}(r g(r)) = f(r)$,
\begin{equation}
    \Phi(r,\theta,t) = v(t) \bigl(r \, g(r)\bigr) \cos\theta,
    \label{eq:irrot_potential}
\end{equation}
where $g(r)$ is given by \eqref{g(r)}, we compute the covariant components of $\boldsymbol{\beta}^\flat = \beta_{\hat\alpha}\boldsymbol{e}^{\hat\alpha}$ in the orthonormal tetrad basis $\{\boldsymbol{e}^{\hat\alpha}\}$. The components of the exterior derivative $\dd\Phi$ in this basis are given by the projections of the gradient onto the tetrad vectors, leading to $\beta_{\hat\mu} = - e^{\nu}{}_{\hat\mu} \partial_\nu \Phi$. For the standard spherical orthonormal tetrad where $\boldsymbol{e}^{\hat r} = \dd r$ and $\boldsymbol{e}^{\hat \theta} = r\,\dd \theta$, this yields:
\begin{align}
    \beta_{\hat r} &= -e^{\nu}{}_{\hat r} \partial_\nu \Phi = -\partial_r \Phi = -v(t) f(r) \cos\theta, \label{eq:beta_hat_r_deriv}\\
    \beta_{\hat \theta} &= -e^{\nu}{}_{\hat \theta} \partial_\nu \Phi = -\frac{1}{r}\partial_\theta \Phi  = v(t) g(r) \sin\theta, \label{eq:beta_hat_theta_deriv}\\
    \beta_{\hat \phi} &= -e^{\nu}{}_{\hat \phi} \partial_\nu \Phi = -\frac{1}{r\sin\theta}\partial_\phi \Phi = 0.
\end{align}
Recalling that $\beta^{\hat{i}} = \eta^{\hat i \hat j}\beta_{\hat{j}} = \delta^{\hat i \hat j}\beta_{\hat{j}} = \beta_{\hat{i}}$ for the spatial components ($i,j \in \{r,\theta,\phi\}$) in the chosen orthonormal tetrad, these derived components \eqref{eq:beta_hat_r_deriv}--\eqref{eq:beta_hat_theta_deriv} correspond precisely to the ansatz components $\beta^{\hat{r}}, \beta^{\hat{\theta}}, \beta^{\hat{\phi}}$ defined in \eqref{eq:BetaIrr}. As previously discussed, however, no classical thermodynamic or fluid conservation interpretation attaches to this potential $\Phi$ in the context of warp-drive geometry.
\smallskip

\begin{proposition}[Hawking–Ellis classification of the irrotational drive]
\label{prop:irrot_class}
\textbf{Assumptions:} $\alpha\equiv 1$, flat \emph{and static} spatial slices (${}^{(3)}R_{ij}=0$, $\partial_t\gamma_{ij}=0$), and an irrotational shift $\boldsymbol{\beta}^{\flat}=-\dd\Phi$.
Then the stress--energy tensor of the derived irrotational warp--drive spacetime is Hawking–Ellis Type~I at every point.
\end{proposition}

\begin{proof}
With $\alpha=1$ and a flat static slice $(\gamma_{ij}\simeq\mathbb{R}^3)$,
the extrinsic curvature is the spatial Hessian
\(K_{\hat{i}\hat{j}}=-D_{\hat{i}}D_{\hat{j}}\Phi\).
The momentum constraint
\(G_{\hat{0}\hat{i}}\propto
  D_{\hat{k}}\bigl(K^{\hat{k}}{}_{\hat{i}}-\delta^{\hat{k}}_{\hat{i}}K\bigr)
      = -D_{\hat{k}}
        \bigl(D^{\hat{k}}D_{\hat{i}}\Phi
              -\delta^{\hat{k}}_{\hat{i}}D_{\hat{m}}D^{\hat{m}}\Phi\bigr)\)
vanishes identically because $[D_{\hat{k}},D_{\hat{i}}]=0$ on the
Riemann‑flat slice.
Hence $G^{\hat{\alpha}}{}_{\hat{\beta}}$ is block-diagonal in the orthonormal tetrad adapted to the Eulerian normal $n^\mu$ (in the coordinate basis $\{\partial_t,\partial_i\}$), with $n^\mu=(1,\beta^i)$ for $\alpha=1$, and coframe $\omega^{\hat 0}=dt$, $\omega^{\hat i}=dx^i-\beta^i dt$; equivalently,
$\partial_t=\alpha\,n^\mu\partial_\mu-\beta^i\partial_i$ so: $\partial_t=n-\beta^i\partial_i$. For flat, time–independent spatial slices $({}^{(3)}R_{ij}=0,\ \partial_t\gamma_{ij}=0)$:
\begin{equation}
G^{\hat{\alpha}}{}_{\hat{\beta}}
=\operatorname{diag}\!\bigl(G^{\hat{0}}{}_{\hat{0}},
                            G^{\hat{i}}{}_{\hat{j}}\bigr),
\qquad
G^{\hat{i}}{}_{\hat{j}}=G^{\hat{j}}{}_{\hat{i}} \, .
\end{equation}

The spatial block is real symmetric, so its three eigenvalues
\(\lambda_{G(k)}\;(k=1,2,3)\) are real; \(\lambda_{G(0)}=G^{\hat{0}}{}_{\hat{0}}\)
is also real.
Because $e_{\hat{0}}$ is timelike and is the eigenvector for
$\lambda_{G(0)}$, the set
\(\{\lambda_{G(\alpha)}\}_{\alpha=0}^{3}\) comprises four real eigenvalues
with a timelike eigenvector.  By the Hawking–Ellis scheme this is precisely
Type~I~\cite{synge1966,Lichnerowicz_1967}. \qedhere
\end{proof}

\paragraph*{Mechanism for Peak Reduction}
The observed reduction in stress–energy extrema is a direct consequence of kinematic irrotationality and its algebraic implications. In warp drives that permit vorticity (e.g., Alcubierre and Natário \cite{Alcubierre_1994,Natário}), a nonzero vorticity (\(\dd\boldsymbol{\beta}^{\flat}\neq 0\)) typically yields a nonzero momentum density \(G_{\hat 0 \hat i}\), introducing energy–flux terms that can lead to patches of Hawking–Ellis Type~IV structure \cite{Martín-Moruno_2018,Maeda2020}.

By contrast, under the present assumptions (\(\alpha\equiv1\), flat and time–independent spatial slices, and \(\dd\boldsymbol{\beta}^{\flat}=0\)), the \(3{+}1\) momentum constraint enforces vanishing momentum density:
\begin{equation}
G_{\hat 0 \hat i}=0,
\label{eq:momentum_zero}
\end{equation}
which block–diagonalizes the mixed Einstein tensor in the Eulerian orthonormal tetrad:
\begin{equation}
G^{\hat\mu}{}_{\ \hat\nu}=\mathrm{diag}\!\big(G^{\hat 0}{}_{\ \hat 0},\ G^{\hat i}{}_{\ \hat j}\big).
\label{eq:blockdiag}
\end{equation}
This guarantees a global Hawking–Ellis Type~I structure and is the algebraic origin of the improved energy–condition profile.

This causal link is not merely theoretical; it is confirmed by a rigorous numerical ablation (Appendix~\ref{app:ablation}). In that test, the baseline smoothing profile \(g(r)\) and numerical tolerances are held fixed while a tunable vortical component is introduced, isolating vorticity as the sole variable. The results (Table~\ref{tab:abl_sweep}) show a catastrophic degradation of the energy metrics: the negative–energy magnitude \(E_-\) increases by a factor exceeding \(500\), and the crucial \(E_+/E_-\) balance collapses. This demonstrates that the favorable energy properties of the baseline model are a direct consequence of its irrotational, curl–free kinematics, not an artifact of profile smoothing.

\paragraph{Comparative summary of warp-drive models}  Table~\ref{tab:focused_comparison_tabularstar} provides a comparative summary of the Alcubierre \cite{Alcubierre_1994}, Natário zero-expansion \cite{Natário}, and kinematically irrotational warp drives, based on the orthonormal tetrad components derived in Secs.~\ref{secGeneralized}, \ref{sec:Alcubierre}, \ref{sec:Natario} and \ref{secIrrotational} (assuming lapse function $\alpha=1$, flat spatial slices, axial symmetry with $x$-axis polar). The table contrasts the defining kinematic characteristic imposed on the shift vector $\boldsymbol{\beta}$, the resulting tetrad components $\beta^{\hat{i}}$, the utilized radial form functions ($f(r)$, $g(r)=\frac{1}{r}(\int f(r) \, \mathrm{d}r)$ derived from $f_{\text{Alc}}(r)$ \eqref{eq:fAlc}), and the inferred observer perspective. This highlights how the specific constraints ($\operatorname{div}\boldsymbol{\beta}=0$ for Natário, $\boldsymbol{\omega}=*\dd\boldsymbol{\beta}^\flat=0$ for the irrotational model) manifest in the shift vector components relative to the original Alcubierre spacetime.


\begin{table}[htbp] 

\caption{Warp-drive models} 
\label{tab:focused_comparison_tabularstar} 

\centering

\begin{tabular*}{\textwidth}{@{\extracolsep{\fill}}|l|l|l|l|} 
\hline
\textbf{Model} & \textbf{Alcubierre \cite{Alcubierre_1994}} & \textbf{Natário (No Exp.) \cite{Natário}} & \textbf{Kinem. irrotational} \\ \hline \hline
\textbf{Charact.} & $\boldsymbol{\beta}^{\flat} = v f_{\text{Alc}} \dd x$ \eqref{eq:betaflatCoordDiff} & $\operatorname{div}\boldsymbol{\beta} = 0$ \eqref{eq:NatKinDiv} & $\boldsymbol{\omega} = * (\dd \boldsymbol{\beta}^{\flat}) = 0$ \eqref{vorticity_vector} \\ \hline
\textbf{Shift Vr.} & from \eqref{\eqBetaAlc}: & from \eqref{eq:NatarioShiftComponents}: & from \eqref{\eqBetaIrr} and  \eqref{eq:g integral}: \\
($\beta^{\hat{i}}$ Tetrad) & $\beta^{\hat{r}} = v f \cos\theta$ & $\beta^{\hat{r}} = -v f \cos\theta$ & $\beta^{\hat{r}} = -v f \cos\theta$ \\
& $\beta^{\hat{\theta}} = -v\, f \sin\theta$ & $\beta^{\hat{\theta}} = v (f + \frac{r}{2} f') \sin\theta$ & $\beta^{\hat{\theta}} =v \frac{1}{r}(\int f \, \mathrm{d}r)\sin\theta $ \\
& $\beta^{\hat{\phi}} = 0$ & $\beta^{\hat{\phi}} = 0$ & $\beta^{\hat{\phi}} = 0$ \\ \hline
\textbf{Form Fn.} & $f =f_{\text{Alc}}$ \eqref{eq:fAlc} & $f = 1 - f_{\text{Alc}}$ \eqref{eq:fNatario} & $f = 1 - f_{\text{Alc}}$ \eqref{\fIrrotational}\\ \hline
\textbf{Observer} & Rest $r \to \infty$; Bubble $v(t)$ & Rest $r = 0$; Stars $-v(t)$ & Rest $r = 0$; Stars $-v(t)$ \\ \hline
\end{tabular*}

\end{table}

\section{Interpretation and visualization of Einstein-tensor eigenvalues}
\label{subsec:eigenvalue_visualization}

This section connects the algebraic framework developed above with the
numerical results for the \textit{Alcubierre},
\textit{irrotational} and \textit{Natário} warp-drive spacetimes.
Subsection~\ref{sec:classification_protocol_final_v5} reviews the Cartan
construction of \(G^{\hat\alpha}{}_{\;\hat\beta}\), applies the
Hawking–Ellis classification (emphasizing Types I and IV), and presents
the eigensolver, tolerance schemes and validation checks that lead to the
physically-motivated ordering of the eigenpairs.

Section~\ref{sec:eigenvalue_interpretation} analyzes the plotted eigen‐data, maps regions where the Einstein tensor is Hawking–Ellis Type I (fully diagonalizable) or Type IV (complex–conjugate null pair), and quantifies violations of both the null and weak energy conditions.  It also computes the scalar invariant $G=-R=\sum_{\alpha=0}^{3}\lambda_{G(\alpha)}$
which enters the trace energy condition, thereby enabling a direct comparison of the stress–energy requirements for the three warp‐drive spacetimes.

Numerical values of the mixed Einstein tensor \(G^{\hat\mu}{}_{\;\hat\nu}\), expressed in an orthonormal tetrad, are computed in Wolfram \textit{Mathematica\textsuperscript{\textregistered}}~\cite{Mathematica} at 50-digit precision (unless stated otherwise).  We examine a warp bubble moving at light‐speed with the parameters \(\rho=5\,[\mathrm{m}]\) and \(\sigma=4\,[\mathrm{m}^{-1}]\).  These choices reproduce those of Refs.~\cite{Rodal,Rodal2024}, which employed a curvilinear‐coordinate tensor formalism; by using the same parameters in our Cartan orthonormal‐tetrad approach (with spin coefficients), we can directly compare the computed \(G^{\hat\mu}{}_{\;\hat\nu}\) for the Alcubierre and Natário drives, thus thoroughly validating our high‐precision numerical implementation.  To economize space, we restrict attention to the light‐speed case here; superluminal warp‐bubble velocities produce analogous qualitative behavior (as follows from the analytic expressions) and have been numerically detailed for the Alcubierre and Natário drives in Ref.~\cite{Rodal2024}.  All figures sample the radial band
\(
  r\in\bigl[\rho-3/\sigma,\;\rho+3/\sigma\bigr],
\)
which contains the bubble walls where the form function \(f(r)\) departs from its asymptotic limits \(f\to1\) and \(f\to0\).

\noindent\emph{Methods note:} All eigenvalue computations used 50-digit precision with symmetry, trace, and eigen-sum checks; results were insensitive to tolerance variations over two orders of magnitude.

\paragraph*{Global energy metrics}
We summarize the slice–integrated measures defined in Appendix~\ref{app:global_energy}: the negative/positive volumes \(V_\pm\) and energy magnitudes \(E_\pm\).  These are integrals over a spatial slice \(\Sigma_t\) and hence are foliation-dependent; they are invariant under spatial coordinate changes and tetrad rotations on \(\Sigma_t\), but not under changes of slicing. They are computed with the disjoint masks of Eqs.~\eqref{eq:regions}–\eqref{eq:tolerances}, using the axially reduced volume element \eqref{eq:vol_element}, and the definitions in Eqs.~\eqref{eq:Epm_magnitudes}–\eqref{eq:volumes}. Headline values for the fiducial \((\rho,\sigma,v/c)=(5~\mathrm{[m]},\,4~\mathrm{[m^{-1}]},\,1)\) on the baseline window \(R_{\rm integ}=12\rho\) appear in Tables~\ref{tab:globals_SI} and \ref{tab:ratios_checks}; localization diagnostics are reported in Tables~\ref{tab:band_captures}–\ref{tab:shell_captures} and the meridional map in Fig.~\ref{fig:eq_density}. These global measures complement the local extrema (e.g.\ \(\min \varrho_p\)) and quantify the spatial extent and integrated magnitude of NEC/WEC–violating regions. A compact snapshot is given in Table~\ref{tab:global_measures}; for far–field completion using a two–point \(1/R\) tail model and corresponding totals, see Sec.~\ref{sec:tail_model} and Table~\ref{tab:baseline_vs_tail}, which drive the net proper energy to
\(|E_{\rm net}|/E_{\rm abs}=0.04\%\) (i.e.\ a \(4\times10^{-4}\) fractional level), consistent with zero.

\begin{table}[htbp]
\centering
\small
\caption{Local and global \emph{energy} measures for the irrotational warp drive at the fiducial \((\rho,\sigma,v/c)\).
Baseline values are on the window \(R_{\rm integ}=12\rho\); tail–corrected values use a two–point \(1/R\) extrapolation to \(R\!\to\!\infty\).
 Numbers are reported with 3\,s.f.\ .}
\label{tab:global_measures}

\begin{tabular}{|l|c|c|c|}
\hline
\textbf{Case} & $E_{-}$ [J] & $E_{+}$ [J] & $E_{+}/E_{-}$ \\
\hline
Baseline (\(12\rho\)) &
$1.21\times 10^{44}$ &
$1.29\times 10^{44}$ &
$1.07$ \\
\hline
Tail–corrected (\(R\!\to\!\infty\), \(1/R\)) &
$1.33\times 10^{44}$ &
$1.33\times 10^{44}$ &
$1.00$ \\
\hline
\end{tabular}

\vspace{0.6ex}

\begin{tabular}{|l|c|c|c|c|}
\hline
\multicolumn{5}{|c|}{\textbf{Local extrema and baseline volumes (at \(R_{\rm integ}=12\rho\))}} \\
\hline
$\max \varrho_p$ [J\,m$^{-3}$] & $\min \varrho_p$ [J\,m$^{-3}$] & $V_{-}$ [m$^{3}$] & $V_{+}$ [m$^{3}$] & $V_{0}$ [m$^{3}$] \\
\hline
$2.23\times 10^{42}$ & $-1.29\times 10^{41}$ & $6.38\times 10^{5}$ & $2.34\times 10^{5}$ & $3.40\times 10^{4}$ \\
\hline
\end{tabular}
\end{table}

\medskip

\textbf{Limitations.} All Type~I and peak-reduction results reported here assume $\alpha\equiv 1$ and flat, \emph{time-independent} spatial slices. Relaxing these conditions (e.g., allowing $\nabla\alpha\neq 0$ or evolving $\gamma_{ij}$) can reintroduce momentum density $G_{\hat0\hat i}$ and spoil global Type~I behavior.

\subsection{\texorpdfstring{Algebraic classification, interpretation, and numerical protocol for \( G^\mu{}_\nu \)}{Algebraic classification, interpretation, and numerical protocol for G mu nu}}
\label{sec:classification_protocol_final_v5} 

The Einstein tensor \(G^\mu{}_\nu=g^{\mu\alpha}G_{\alpha\nu}\) shares the same
algebraic structure as the stress–energy tensor \(T^\mu{}_\nu\).
 With the cosmological constant set to zero \(\Lambda=0\), Einstein’s field equations \(G^\mu{}_\nu=\kappa\,T^\mu{}_\nu\) imply a pointwise proportionality between these two tensors, so the
Hawking–Ellis (Segre–Plebański) classification
\cite{hawking_ellis_1973,Stephani_Kramer_MacCallum_Hoenselaers_Herlt_2003,Plebanski1964,Segre}
can be applied directly to \(G^\mu{}_\nu\).
In practice this means that computing the eigenvalues and causal character of \(G^\mu{}_\nu\) directly reveals the type of stress-energy tensor or effective source required.  This is particularly useful in warp‐drive metrics, where one specifies \(g_{\mu\nu}\) \emph{a priori} and then reads off the implied energy–momentum content rather than solving for it.

\subsubsection{Einstein tensor via Cartan formalism}

The mixed components \(G^{\hat\alpha}{}_{\;\hat\beta}\) are computed
entirely in Cartan’s tetrad language, using Wald’s sign conventions~\cite{Wald}
and avoiding Christoffel symbols~\cite{deFelice}.  We choose an orthonormal coframe
\(\{\omega^{\hat\alpha}\}\equiv
   \{e^{\hat\alpha}{}_{\;\mu}\,\mathrm{d}x^{\mu}\}\)
with
\(
  g_{\mu\nu}
    =\eta_{\hat\alpha\hat\beta}\,
      e^{\hat\alpha}{}_{\;\mu}\,
      e^{\hat\beta}{}_{\;\nu},
    \quad
  \eta_{\hat\alpha\hat\beta}=\mathrm{diag}(-,+,+,+).
\)
The dual frame vectors
\(\mathbf{e}_{\hat\alpha}=e_{\hat\alpha}{}^{\;\mu}\partial_\mu\)
satisfy the usual completeness relations.

From the first Cartan structure equation (zero torsion, \(T^{\hat\alpha}=0\)):
\begin{equation}
  \mathrm{d}\,\omega^{\hat\alpha}
  +\omega^{\hat\alpha}{}_{\;\hat\beta}
   \wedge\omega^{\hat\beta}
  =0,
  \label{eq:Cartan1}
\end{equation}
one reads off the antisymmetric spin‐connection 1-forms
\(\omega^{\hat\alpha}{}_{\;\hat\beta}
  =-\omega_{\hat\beta}{}^{\;\hat\alpha}.\)

The second Cartan structure equation defines the curvature 2-forms:
\begin{equation}
  \Omega^{\hat\alpha}{}_{\;\hat\beta}
    =\mathrm{d}\,\omega^{\hat\alpha}{}_{\;\hat\beta}
     +\omega^{\hat\alpha}{}_{\;\hat\gamma}
      \wedge\omega^{\hat\gamma}{}_{\;\hat\beta}
    =\tfrac12
      R^{\hat\alpha}{}_{\;\hat\beta\hat\gamma\hat\delta}\,
      \omega^{\hat\gamma}\wedge\omega^{\hat\delta}.
  \label{eq:Cartan2}
\end{equation}
Contracting \(R^{\hat\alpha}{}_{\;\hat\beta\hat\gamma\hat\delta}\) yields
\(
  R_{\hat\alpha\hat\beta}
    =R^{\hat\gamma}{}_{\;\hat\alpha\hat\gamma\hat\beta},
  \quad
  R=\eta^{\hat\alpha\hat\beta}R_{\hat\alpha\hat\beta}.
\)
Finally, the Einstein tensor is
\(
  G_{\hat\alpha\hat\beta}
    =R_{\hat\alpha\hat\beta}
     -\tfrac12\,R\,\eta_{\hat\alpha\hat\beta},
  \quad
  G^{\hat\alpha}{}_{\;\hat\beta}
    =\eta^{\hat\alpha\hat\gamma}G_{\hat\gamma\hat\beta},
\)
and we feed this \(4\times4\) matrix to the eigenanalysis.

\subsubsection{Hawking–Ellis Type IV: geometry and physical meaning}\label{sec:IVgeometry_and_physical_meaning}

Type IV occurs when the mixed Einstein tensor \(G^{\hat\alpha}{}_{\hat\beta}\) possesses one complex–conjugate eigen‑pair and two distinct real spacelike eigenvalues; consequently, no real causal (timelike) eigenvector exists. 
Segre \cite[§20, p.~330]{Segre} records this four-simple-divisor spectrum as \([1\,1\,1\,1]\) when viewed over \(\mathbb{C}\). Petrov adopts the same bracket \cite[§47, pp.\,324–325]{Petrov}, while Stephani et al. \cite[Table 5.1]{Stephani_Kramer_MacCallum_Hoenselaers_Herlt_2003} write it as \([1\,1,\,Z\bar Z]\). 
Following Hall \cite[Table 7.1]{Hall} we use the label \([z\,\bar z\,11]\).

Let the complex eigenvalues be
\begin{equation}\label{eq:TypeIV_spectrum}
\Lambda_G=\xi+i\mu,\quad
\bar\Lambda_G=\xi-i\mu,\quad
\lambda_{2},\lambda_{3}\in\mathbb{R},\quad
\lambda_{2}\neq\lambda_{3},\quad
\mu\neq0 .
\end{equation}
with eigenvectors \(z^\mu\) and \(\bar z^\mu\).
Because \(G^{\hat\alpha}{}_{\hat\beta}\) is self-adjoint,
\(g(G(z),\bar z)=g(z,G(\bar z))\).
Substituting \(G(z)=\Lambda_G z\) and \(G(\bar z)=\bar\Lambda_G \bar z\),
this gives \((\Lambda_G-\bar\Lambda_G)\,g(z,\bar z)=0\).
Since \(\Lambda_G\neq\bar\Lambda_G\), we have
\(g(z,\bar z)=0\).

\paragraph*{Null two-plane.}
Write \(z^\mu=x^\mu+i\,y^\mu\).  Then \(g(x,x)+g(y,y)=0\).
Without loss of generality take \(x^\mu\) timelike and \(y^\mu\)
spacelike, orthogonal; the real plane
\(\mathcal N=\mathrm{span}\{x^\mu,y^\mu\}\) is Lorentzian and contains
two independent null directions.  Every genuine Type-IV tensor shares
this intrinsic null two-plane; no additional branch with
\(g(z,\bar z)\neq0\) is algebraically possible. This impossibility traces directly to the
\(g\)-self-adjointness of \(G^{\hat\alpha}{}_{\hat\beta}\).

\paragraph*{Canonical \(2\times2\) block.}
Let \(\{t^\mu,s^\mu\}\subset\mathcal N\) be an orthonormal basis of the
invariant null two–plane, with
\(g(t,t)=-1,\; g(s,s)=+1,\; g(t,s)=0\).
Then
\begin{subequations}\label{eq:G_on_plane}
\begin{align}
  G^{\mu}{}_{\nu}\,t^{\nu} &= \xi\,t^{\mu}-\mu\,s^{\mu},\\
  G^{\mu}{}_{\nu}\,s^{\nu} &= \mu\,t^{\mu}+\xi\,s^{\mu},
\end{align}
\end{subequations}
so on the \(t\!-\!s\) plane the mixed Einstein tensor is represented by
\[
\begin{pmatrix}\xi & \mu\\[2pt]-\mu & \xi\end{pmatrix}.
\]

\emph{Canonical forms across the literature.} Our mixed \(2\times2\) Type-IV block matches the canonical forms in \cite{hawking_ellis_1973,Maeda2020,Martín-Moruno_2018,Petrov} after the appropriate index-position conversion (covariant or contravariant \(\leftrightarrow\) mixed) and signature bookkeeping; explicit parameter maps are given in App.~\ref{app:typeIV_maps}.


\paragraph*{Physical measurables.}
With Einstein’s coupling constant \(\kappa\),
an observer whose four–velocity is \(t^\mu\) (aligned with the null
plane) measures
\[
\varrho_{\text{obs}}=-\xi/\kappa,\quad
p_s=\xi/\kappa,\quad
q_s=\mu/\kappa,
\]
where \(q_s\) is an unavoidable energy- (or momentum-) flux
in the \(s^\mu\) direction.  As shown in
Prop.\,\ref{prop:mu_invariance}, the quantities
\(\xi\) and \(|\mu|\) are Lorentz scalars, but
\(\varrho_{\text{obs}}\) itself is frame-dependent because no real
causal eigenvector exists.

\smallskip

\begin{proposition}[Invariant magnitudes of \(\xi\) and \(\mu\)]
\label{prop:mu_invariance}
Let \(G^{\hat\alpha}{}_{\hat\beta}\) at a spacetime point be
Hawking–Ellis Type IV, with complex eigenvalues
\(\Lambda_G=\xi\pm i\mu\).
Then \(\xi\) and \(\lvert\mu\rvert\) are Lorentz scalars
(only the sign of \(\mu\) can flip).
Consequently
\(\lvert q_s\rvert=\lvert\mu\rvert/\kappa\) is invariant, and for
observers whose 4-velocities lie in the null two-plane \(\mathcal N\)
one has \(\varrho_{\text{obs}}=-\xi/\kappa\), also invariant within
that class.
\end{proposition}

\begin{proof}
In an orthonormal tetrad adapted to \(\mathcal N\),
\[
B \;=\; G^{\hat\alpha}{}_{\hat\beta}\big|_{\mathcal N}
 =\begin{pmatrix}\xi & \mu\\[2pt]-\mu & \xi\end{pmatrix},
 \qquad \mu\neq0.
\]
Any Lorentz transformation that preserves \(\mathcal N\) acts by real
similarity on \(B\), hence
\[
\tr B = 2\xi,
\qquad
\det B = \xi^2 + \mu^2
\]
are invariants.  Since
\(\mu^2=\det B-\tfrac14(\tr B)^2\) is invariant,
only \(|\mu|\) is well defined (orientation reversals send
\(\mu\to-\mu\)).
Dividing by \(\kappa\) gives the stated invariants for
\(\lvert q_s\rvert\) and, for null-plane observers,
\(\varrho_{\text{obs}}\). \qedhere
\end{proof}

\smallskip
Although \(\xi=\tfrac12\,\tr B\) is a true scalar, it is \emph{not} a
proper energy density: no real timelike eigenvector exists.  The energy
density measured by an arbitrary observer \(t^{\hat\gamma}\) is
\[
\varrho_{\text{obs}}[t^{\hat\gamma}]
   \;=\;\frac{1}{\kappa}\,
        G_{\hat\alpha\hat\beta}\,t^{\hat\alpha}t^{\hat\beta},
\]
and varies with the frame.  Only observers whose four-velocities lie in
\(\mathcal N\) obtain the invariant value
\(\varrho_{\text{obs}}=-\xi/\kappa\).
By contrast, in a Type I region the unique timelike eigenvalue yields
the proper (frame-invariant) energy density
\(\varrho_p=-\lambda_{G(0)}/\kappa\).

\paragraph*{Null–energy–condition violation}

Because the energy–flux parameter \(\mu\) attached to the complex
eigenvalues is non-zero (\(\mu\neq0\)),
\emph{every} Hawking–Ellis Type IV tensor violates the null energy
condition,
\(G_{\mu\nu}K^\mu K^\nu\ge0\),
for a suitable null vector \(K^\mu\).
Write the complex eigenvector as \(z^\mu=x^\mu+i y^\mu\) with
\(g(x,x)=-g(y,y)\) and \(g(x,y)=0\); the invariant two–plane
\(\mathcal N=\mathrm{span}\{x,y\}\) then contains the two null
directions
\(K_\pm^\mu = x^\mu \pm y^\mu\).
A direct calculation gives
\[
G_{\mu\nu}K_\pm^\mu K_\pm^\nu
      = \mp\,2\,|\mu|\;|g(x,x)| ,
\]
so one choice of sign is strictly negative, establishing NEC
violation.

\paragraph*{Structural stability}

Type IV is \emph{structurally unstable}.
Its defining algebraic constraint
\(g(z,\bar z)=0\) is an equality that is generically spoiled by an
arbitrarily small perturbation of the tensor, causing the eigen-system
to move into another Hawking–Ellis class (typically Type I, or,
at special degeneracy points, Type II/III).
By contrast, Type I—with four distinct real eigenvalues—is an open,
stable class.
(Types II and III are likewise unstable because they rely on
eigenvalue degeneracies.)

\subsubsection{Numerical protocol and validation\label{sec:numerics}}

We analyze the mixed Einstein tensor $G^{\hat\alpha}{}_{\;\hat\beta} = G^{\mu}{}_{\nu}\,e^{\hat\alpha}{}_{\mu}\,e^{\nu}{}_{\hat\beta}$ numerically in Wolfram \textit{Mathematica\textsuperscript{\textregistered}} \cite{Mathematica}. The symbolic expression \(G^{\hat\alpha}{}_{\;\hat\beta}(r,\theta;\rho,\sigma,v)\) is evaluated at selected coordinates $(r_{0},\theta_{0})$ and warp bubble parameters, $(\rho_{0},\sigma_{0},v_{0})$, to obtain the numeric \(4\times4\) real matrix
\begin{equation}
   M=
   G^{\hat\alpha}{}_{\;\hat\beta}(r_0,\theta_0;\rho_0,\sigma_0,v_0) 
   \in\mathbb{R}^{4\times4},
\end{equation}
All numerical computations are performed using $p=50$-digit precision by default.
The raw eigensystem \(\{\lambda_I,{\bm v}_I\}_{I=1..4}\) of \(M\) is computed using \texttt{Eigensystem[]}.

\paragraph*{Phase fixing of eigenvectors}
Each raw eigenvector \({\bm v}_I\) is then multiplied by a unit-modulus factor  \(e^{\mathrm i\psi}\!\in\!U(1)\) (a \emph{global phase}). The phase \(\psi\) is chosen such that the time-like component \(v^{\hat 0}\) of the resulting vector becomes real and non-negative; if \(\lvert v^{\hat 0}\rvert<\epsilon_{\mathrm{chop}}=10^{-18}\), \(\psi\) is instead chosen to make the largest-magnitude component real and non-negative. Because this rotation by \(e^{\mathrm i\psi}\) leaves scalar products like \(\eta_{\hat\alpha\hat\beta}v^{\hat\alpha}\bar v^{\hat\beta}\) unchanged, this phase-fixing does not affect any physical result derived from such products (e.g., eigenvector norms) but ensures numerical reproducibility. The set of phase-fixed eigenvectors is denoted \(\{{\bm v}'_I\}_{I=1..4}\).

\paragraph*{Classification procedure}%
From the raw eigensystem \(\{\lambda_I,{\bm v}'_I\}\) we decide the
Hawking–Ellis type; \(N_C\) counts non-real eigenvalues.
\begin{enumerate}[label=(\Roman*)]
  \item \textbf{Type I}:  \(N_C=0\).
  \item \textbf{Type II}: at least one \emph{real} null eigenvector
        (Jordan vs.\ diagonal via \({\rm rank}(M-\lambda I)\)).
  \item \textbf{Type IV}: \(N_C=2\) \emph{and} the complex eigen-vectors
        satisfy \(g(z,\bar z)=0\) (null two-plane).
\end{enumerate}
Tests are attempted in the order I \(\rightarrow\) II \(\rightarrow\) IV.

\paragraph*{Numerical tolerances}
\begin{equation}
\begin{aligned}
|\Im\lambda_I| &>\epsilon_{\mathrm{imag}}=10^{-18} 
  &&\Rightarrow\; \lambda_I\ \text{is complex},\\ 
|\lambda_I-\lambda_J| &<\epsilon_{\mathrm{degen}}=10^{-10} 
  &&\Rightarrow\; \text{near-degeneracy},\\
|\eta_{\hat\alpha\hat\beta}v'^{\hat\alpha}\bar v'^{\hat\beta}| 
  &<\epsilon_{\mathrm{null}}=10^{-12}
  &&\Rightarrow\; {\bm v}'\ \text{is null}. 
\end{aligned}
\end{equation}

\paragraph*{Initial numerical ordering}%
After \texttt{Eigensystem[]} the raw list
\(\{\lambda_{I},\bm v'_{(I)}\}_{I=1}^{4}\) is sorted \emph{numerically}
according to the Hawking–Ellis type assigned by the classifier:
\begin{itemize}[nosep,leftmargin=1.8em]
  \item \textbf{Type I or II}: all eigen-values are real;
        order by descending \(\operatorname{Re}\lambda\).
  \item \textbf{Type IV}: complex conjugate pair first
        (\(\Im\lambda>0\) placed before \(\Im\lambda<0\)),
        then the two real roots in descending \(\operatorname{Re}\lambda\).
\end{itemize}
This step yields a deterministic but non-physical indexing.

\paragraph*{Physical reordering}%
The list is then mapped to the physically ordered eigensystem
\(\{\lambda_{G(\alpha)},\bm v_{G(\alpha)}\}_{\alpha=0}^{3}\):
\begin{itemize}[nosep,leftmargin=1.8em]
  \item \textbf{Type I or II}:  
        \(\bm v_{G(0)}\) is the unique causal eigenvector—%
        timelike for I \((g<-\epsilon_{\text{null}})\),  
        \emph{numerically} null for II \((|g|<\epsilon_{\text{null}})\).(The absolute‑value tolerance accounts for round‑off error.)  The remaining spacelike vectors occupy
        \(G(1), G(2), G(3)\) via a two-key sort:  
        (i) dominant spatial component, ascending \((r,\theta,\phi)\);  
        (ii) tie-breaker \(\operatorname{Re}\lambda\), descending.
\item \textbf{Type IV}:  
      the complex–conjugate pair is placed in \(G(0),G(1)\)
      (the root with $\operatorname{Im}\lambda>0$ goes to $G(0)$,
      the one with $\operatorname{Im}\lambda<0$ to $G(1)\bigr)$;  
      the two real spacelike eigenvectors then fill $G(2),G(3)$
      using the same two-key rule.
\end{itemize}
\smallskip
This two-stage procedure ensures that every slot label
\(G(0)\)–\(G(3)\) refers to the same physical quantity at every grid
point, regardless of the raw eigensolver order.
If no unique causal eigenvector is found, the original numerical order is retained.

\paragraph*{Convention for the Physically Ordered Spectrum}\label{par:spectrum_convention}
The final output of our numerical pipeline is the physically ordered set of eigenvalues \(\{\lambda_{G(0)},\lambda_{G(1)},\lambda_{G(2)},\lambda_{G(3)}\}\) and their eigenvectors. For Type I regions, \(\lambda_{G(0)}\) corresponds to \(-\kappa\varrho_p\), while \(\lambda_{G(1)},\lambda_{G(2)},\lambda_{G(3)}\) correspond to \(\kappa p_{(r)}, \kappa p_{(\theta)}, \kappa p_{(\phi)}\) respectively, based on the eigenvector projections detailed in the physical reordering. For Type IV regions, \(\lambda_{G(0)},\lambda_{G(1)}\) are the complex pair \(\xi \pm i\mu\), and \(\lambda_{G(2)},\lambda_{G(3)}\) are the sorted real eigenvalues corresponding to the transverse pressures.

\paragraph*{Reliability and precision} For every initially computed eigenpair $\{\lambda_I, \bm{v}_I\}_{I=1..4}$, we verify the residual 
$\|M{\bm v}_I-\lambda_I{\bm{v}}_I\|_2 = \mathcal{O}(10^{-p})$ at the working precision $p$. 
The numerical matrix $M=G^{\hat{\alpha}}{}_{\hat{\beta}}$, derived from real tetrad components, 
is itself real; any non-real eigenvalues occur in complex conjugate pairs. Reliability is further supported by verifying that the sum of the final, physically reordered eigenvalues, $\sum_{\alpha=0}^{3}\lambda_{G(\alpha)}$, equals the trace 
$\mathrm{Tr}\,M = G^{\hat{\beta}}{}_{\hat{\beta}} =  G^{(\beta)}{}_{(\beta)} =G = \kappa \, T =-R$ (the scalar curvature invariant) 
and is purely real to $\mathcal{O}(10^{-p})$ accuracy. Exploratory parameter scans are performed at machine precision ($p\approx16$ decimal digits), whereas all quoted numbers and the figures employ $p=50$.

\begin{figure}[htbp] 
    \centering 

    \begin{subfigure}[b]{0.32\textwidth} 
        \centering
        \includegraphics[width=\linewidth]{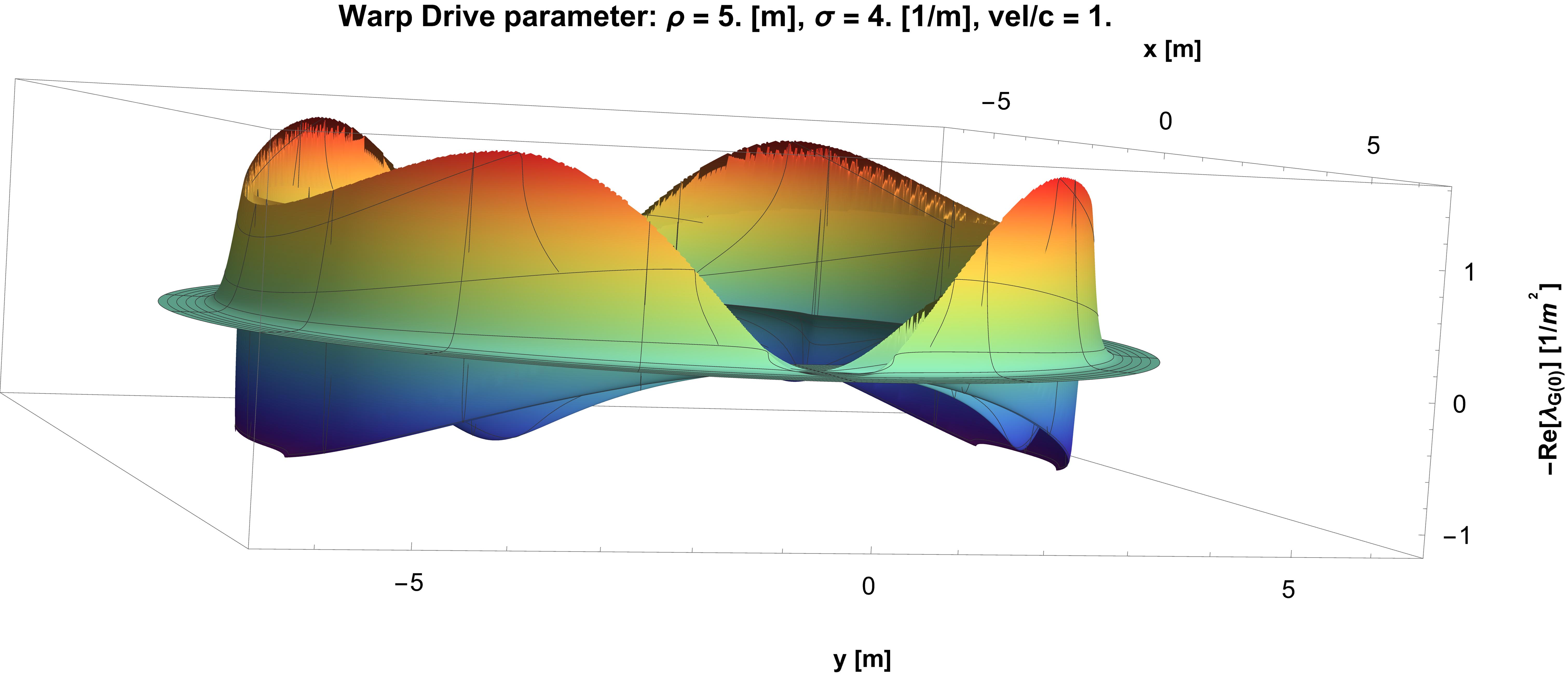}
        \caption{Alcubierre: $-\operatorname{Re}[\lambda_{G(0)}]$} 
        \label{fig:sub_r1_c1 A}
    \end{subfigure}
    \hfill 
    \begin{subfigure}[b]{0.32\textwidth} 
        \centering
        \includegraphics[width=\linewidth]{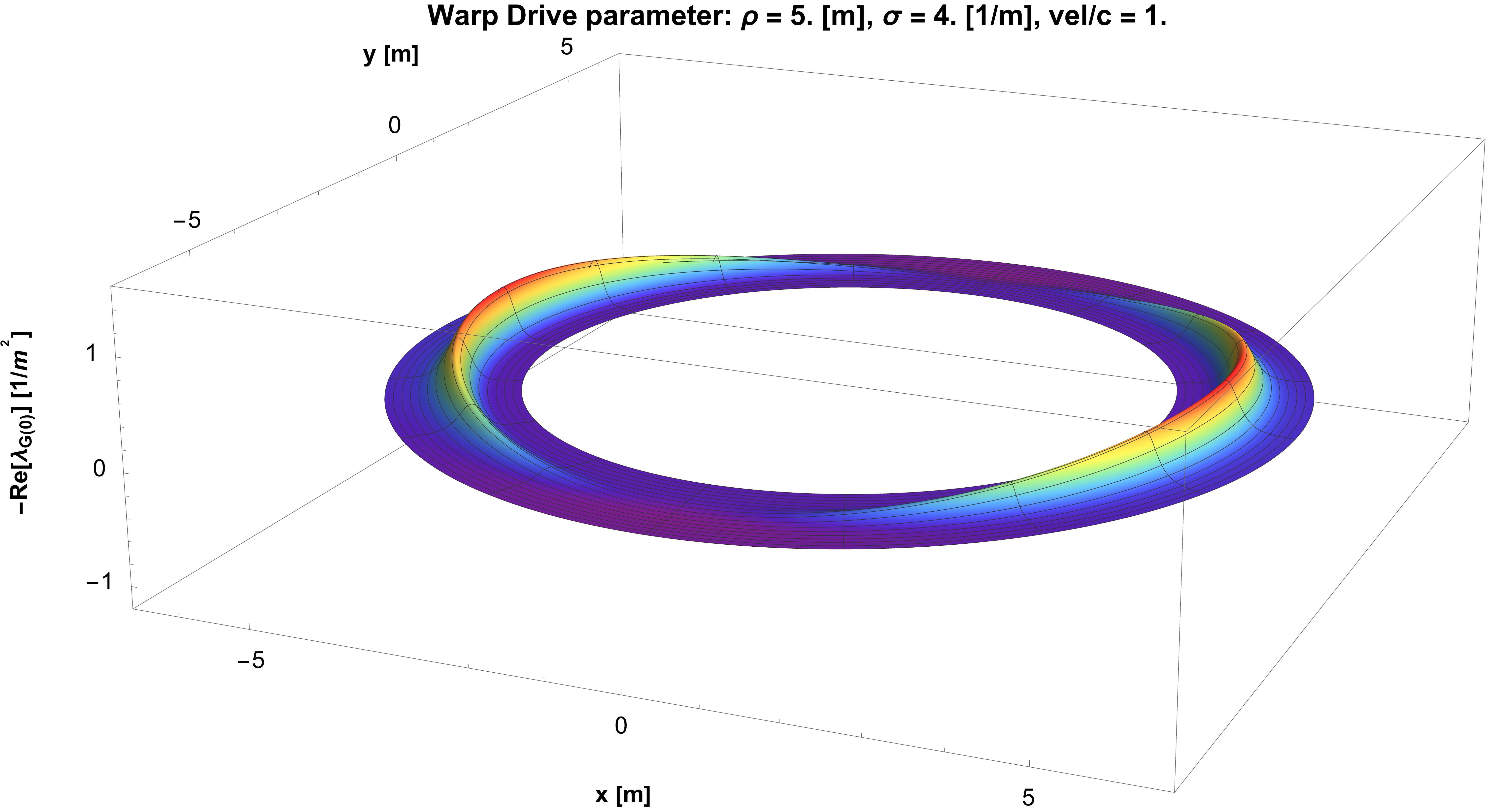} 
        \caption{Irrotational: $-\operatorname{Re}[\lambda_{G(0)}]$} 
        \label{fig:sub_r1_c2 A} 
    \end{subfigure}
    \hfill 
    \begin{subfigure}[b]{0.32\textwidth} 
        \centering
        \includegraphics[width=\linewidth]{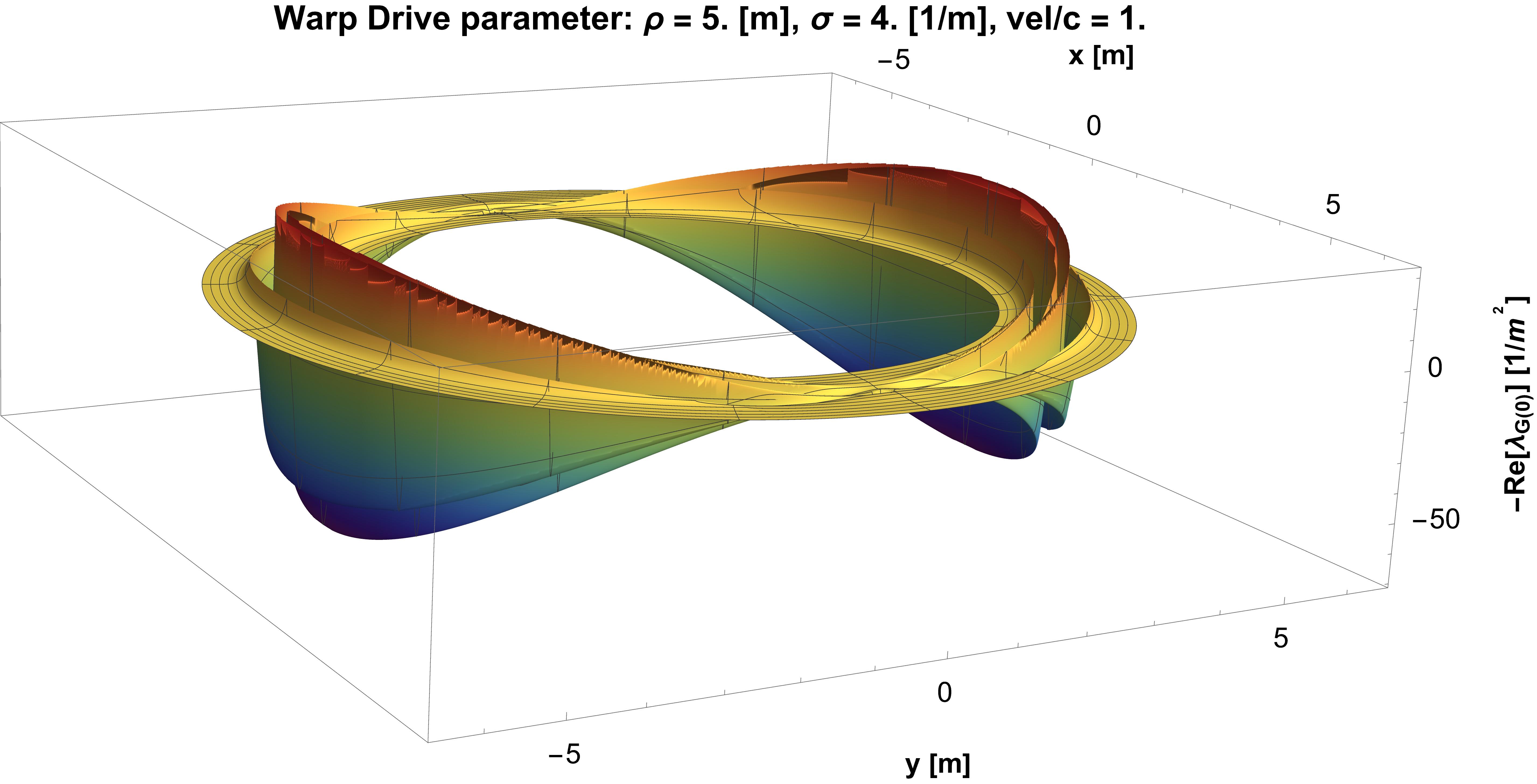}
        \caption{Natário: $-\operatorname{Re}[\lambda_{G(0)}]$}
        \label{fig:sub_r1_c3 A} 
    \end{subfigure}

    \vspace{0.5cm} 

    \begin{subfigure}[b]{0.32\textwidth} 
        \centering
        \includegraphics[width=\linewidth]{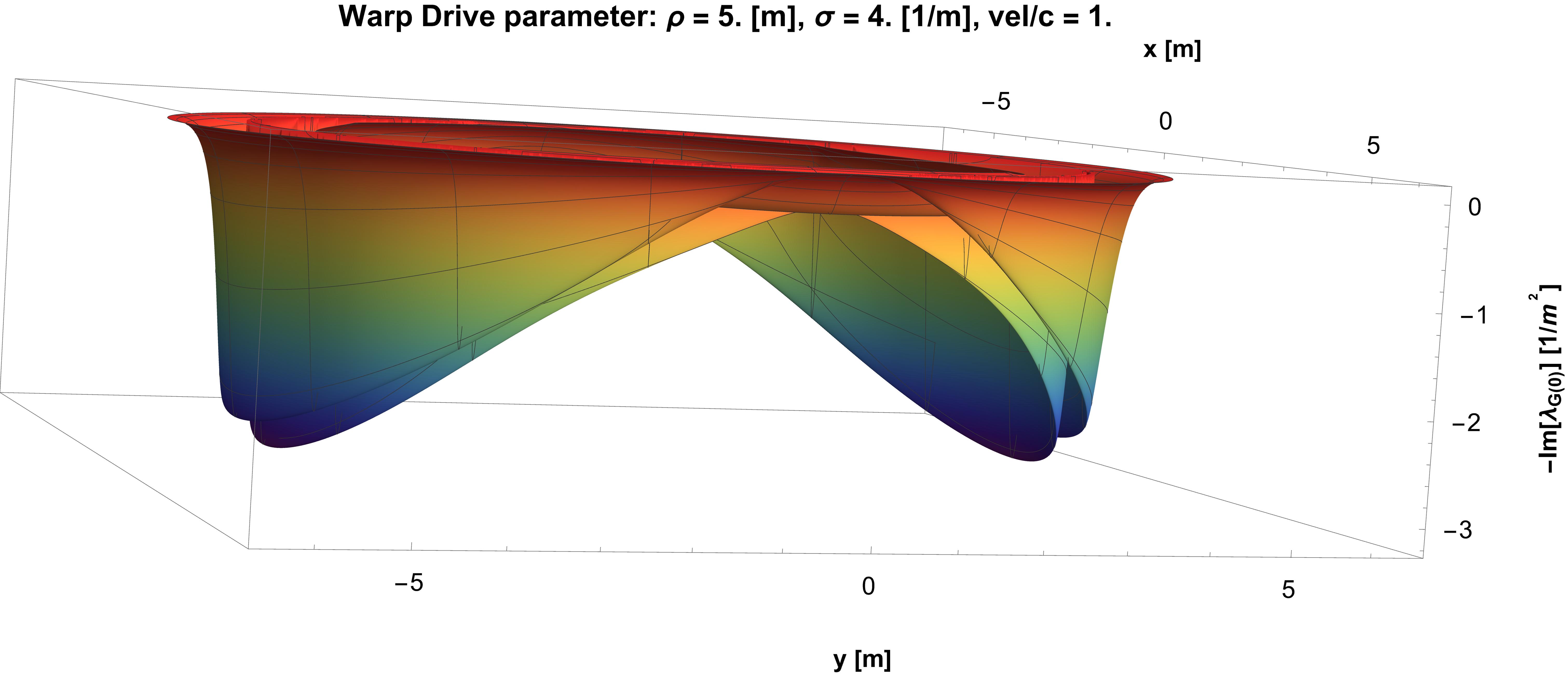} 
        \caption{Alcubierre: $-\operatorname{Im}[\lambda_{G(0)}]$} 
        \label{fig:sub_r2_c1 A} 
    \end{subfigure}
    \hfill 
    \begin{subfigure}[b]{0.32\textwidth} 
        \centering
        \includegraphics[width=\linewidth]{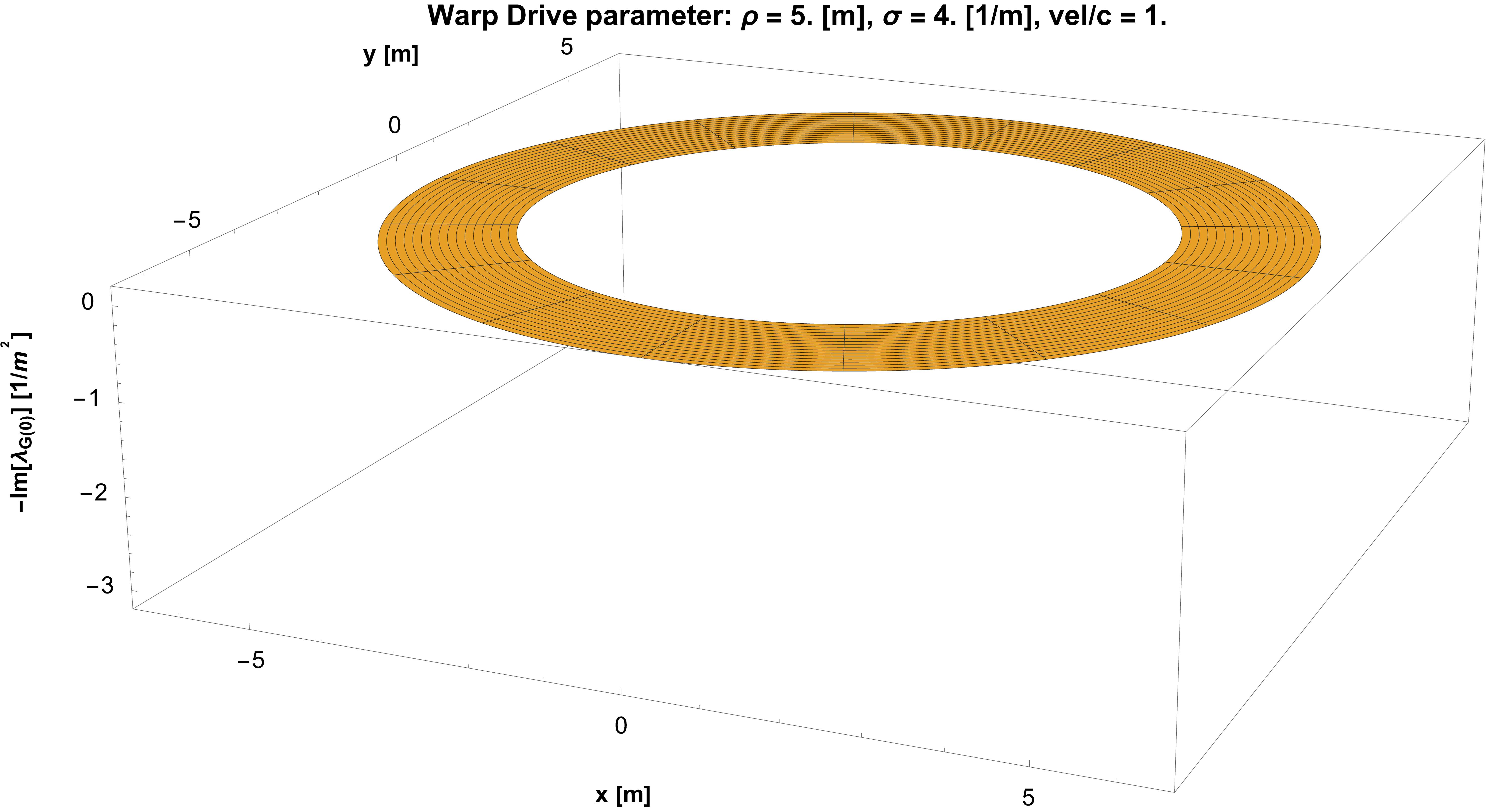}
        \caption{Irrotational: $-\operatorname{Im}[\lambda_{G(0)}]$} 
        \label{fig:sub_r2_c2 A}
    \end{subfigure}
    \hfill 
    \begin{subfigure}[b]{0.32\textwidth} 
        \centering
        \includegraphics[width=\linewidth]{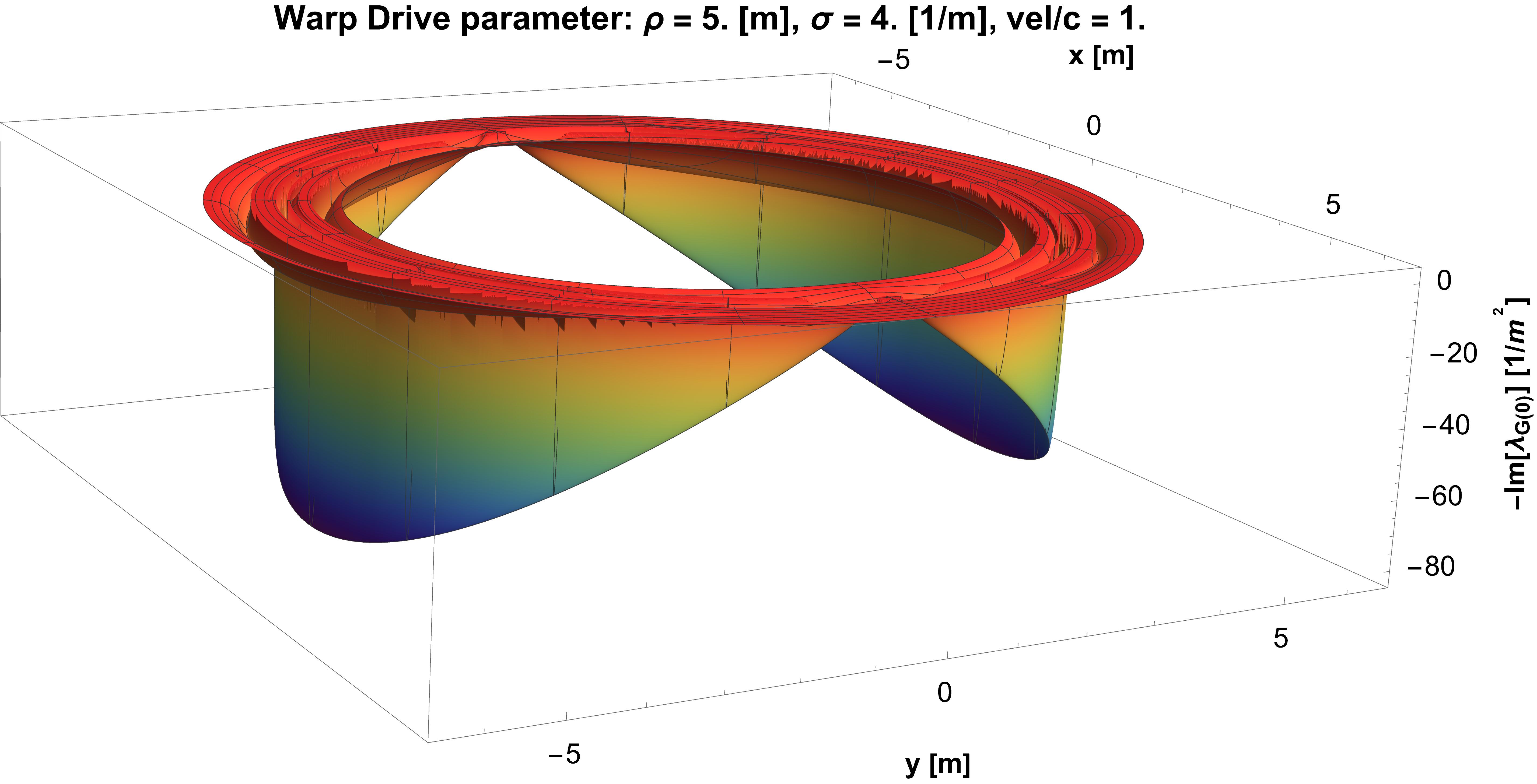}
        \caption{Natário: $-\operatorname{Im}[\lambda_{G(0)}]$}
        \label{fig:sub_r2_c3 A}
    \end{subfigure}

    \vspace{0.5cm} 

    \begin{subfigure}[b]{0.32\textwidth} 
        \centering
        \includegraphics[width=\linewidth]{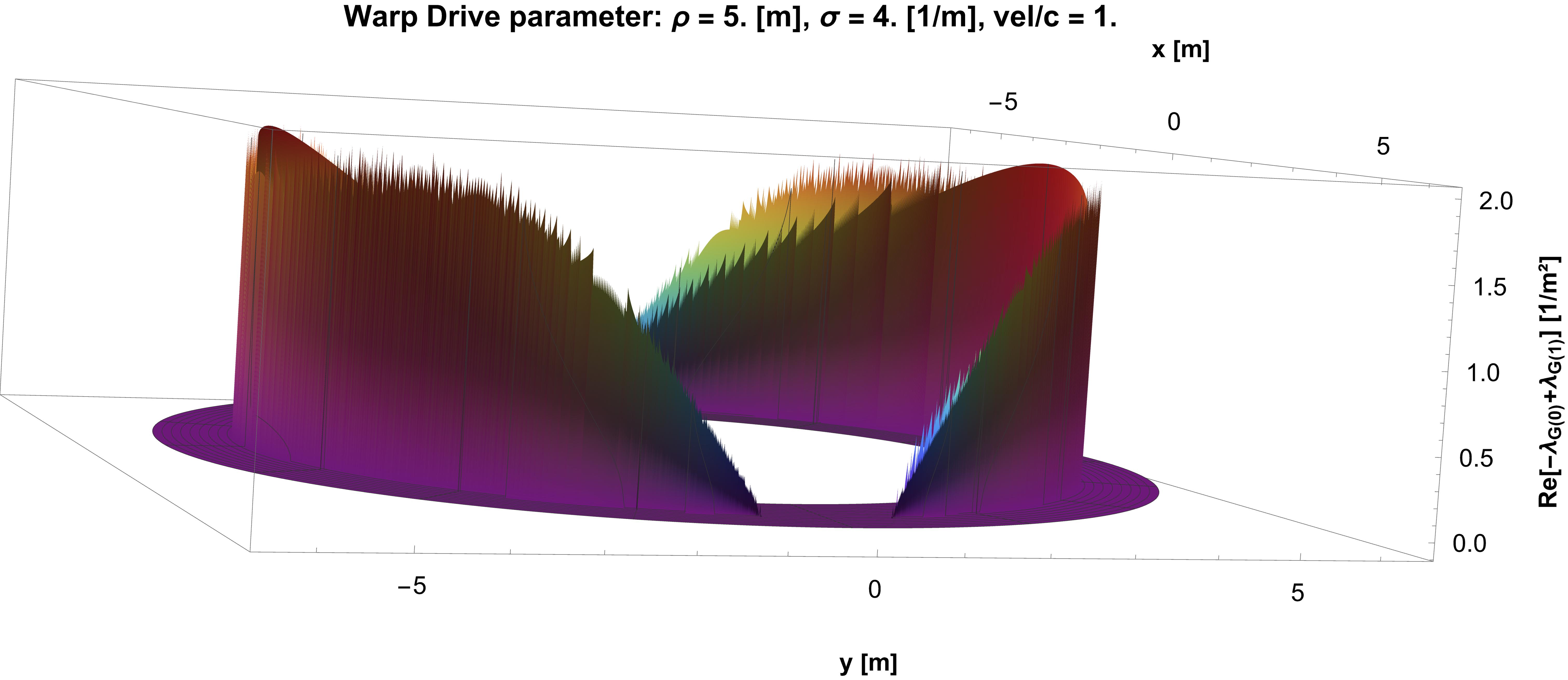}
        \caption{Alc.: $\operatorname{Re}[-\lambda_{G(0)}+\lambda_{G(1)}]$} 
        \label{fig:sub_r3_c1 A}
    \end{subfigure}
    \hfill 
    \begin{subfigure}[b]{0.32\textwidth} 
        \centering
        \includegraphics[width=\linewidth]{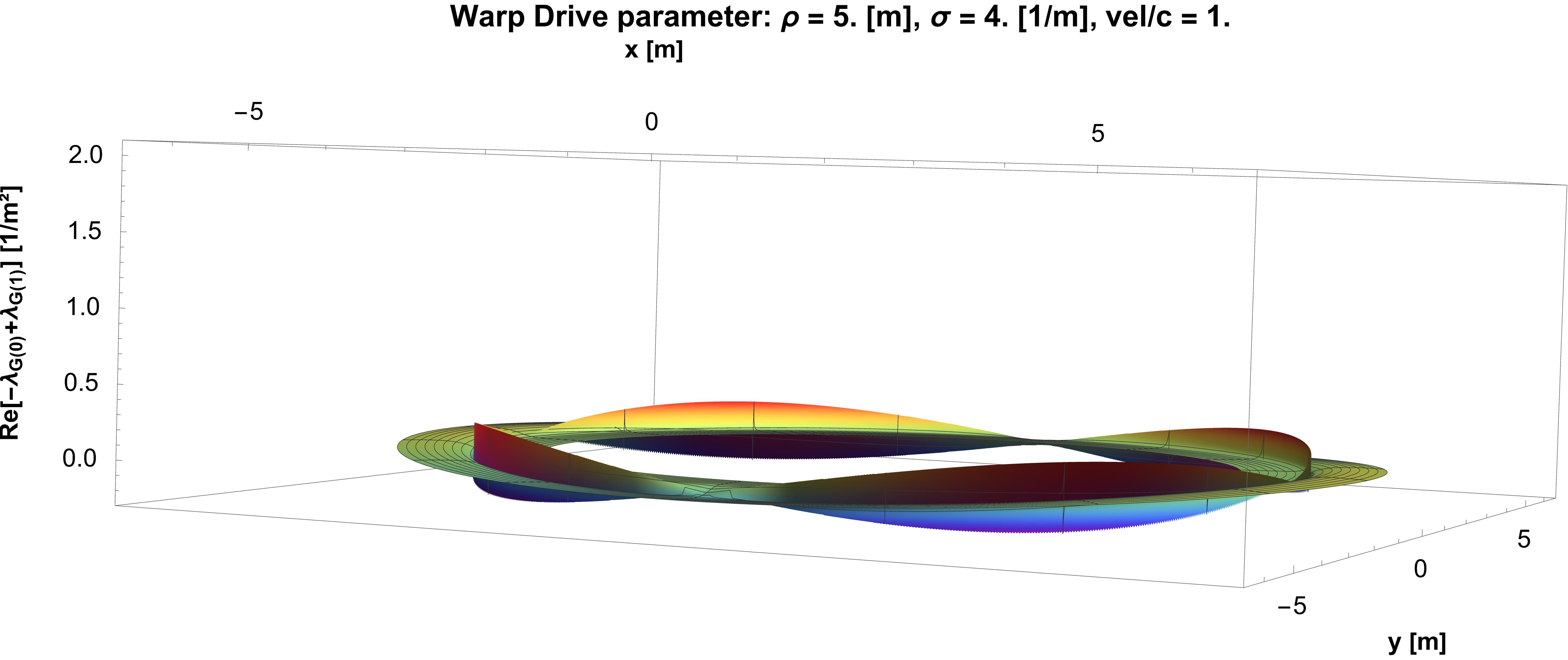}
        \caption{Irr.: $\operatorname{Re}[-\lambda_{G(0)}+\lambda_{G(1)}]$}
        \label{fig:sub_r3_c2 A}
    \end{subfigure}
    \hfill 
    \begin{subfigure}[b]{0.32\textwidth} 
        \centering
        \includegraphics[width=\linewidth]{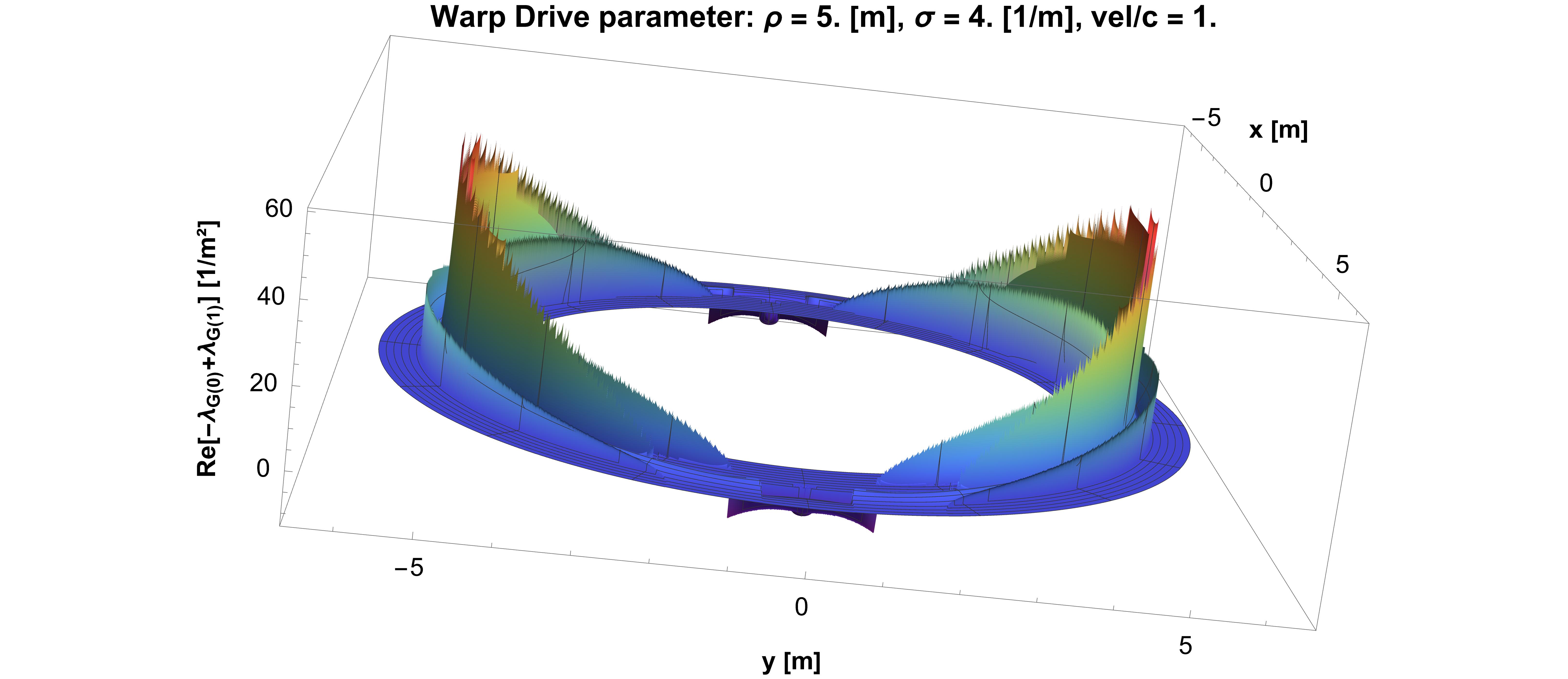}
        \caption{Nat.: $\operatorname{Re}[-\lambda_{G(0)}+\lambda_{G(1)}]$}
        \label{fig:sub_r3_c3 A}
    \end{subfigure}


\caption{Distribution of physically reordered eigenvalues $\lambda_{G(0)}$ and $\lambda_{G(1)}$ against the axis of travel $x$ and transverse $y$ (coordinates in Fig.~\ref{Fig1_SpherCoord}). Columns compare Alcubierre, irrotational, and Natário warp drives. Rows display: (1)~$-\Re[\lambda_{G(0)}]$; (2)~$-\Im[\lambda_{G(0)}]$; (3)~$\Re[-\lambda_{G(0)} + \lambda_{G(1)}]$. In Type~I regions, $\Im[\lambda_{G(0)}] = 0$, hence $-\lambda_{G(0)}$ (in row~(1)), equals $\kappa$ times the \emph{proper} energy density; conversely, in Type~IV regions, $-\Im[\lambda_{G(0)}]$(in row~(2)), is non-zero and proportional to the energy flux. See text for detailed physical interpretation and NEC analysis.}
\label{fig:main_3x3_grid A}

\end{figure}

\begin{figure}[htbp] 
    \centering 

    \begin{subfigure}[b]{0.32\textwidth} 
        \centering
        \includegraphics[width=\linewidth]{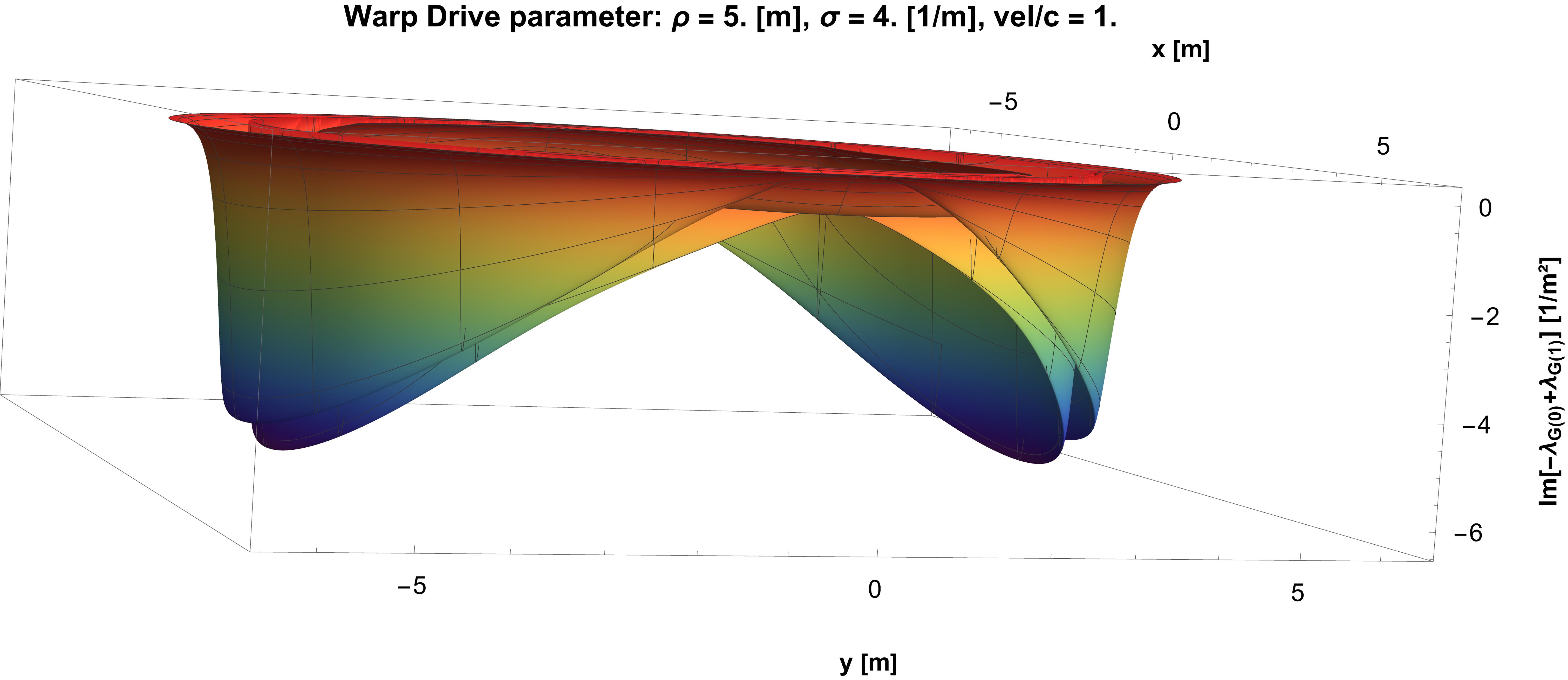}
        \caption{Alc.: $\operatorname{Im}[-\lambda_{G(0)}+\lambda_{G(1)}]$} 
        \label{fig:sub_r1_c1 B}
    \end{subfigure}
    \hfill 
    \begin{subfigure}[b]{0.32\textwidth} 
        \centering
        \includegraphics[width=\linewidth]{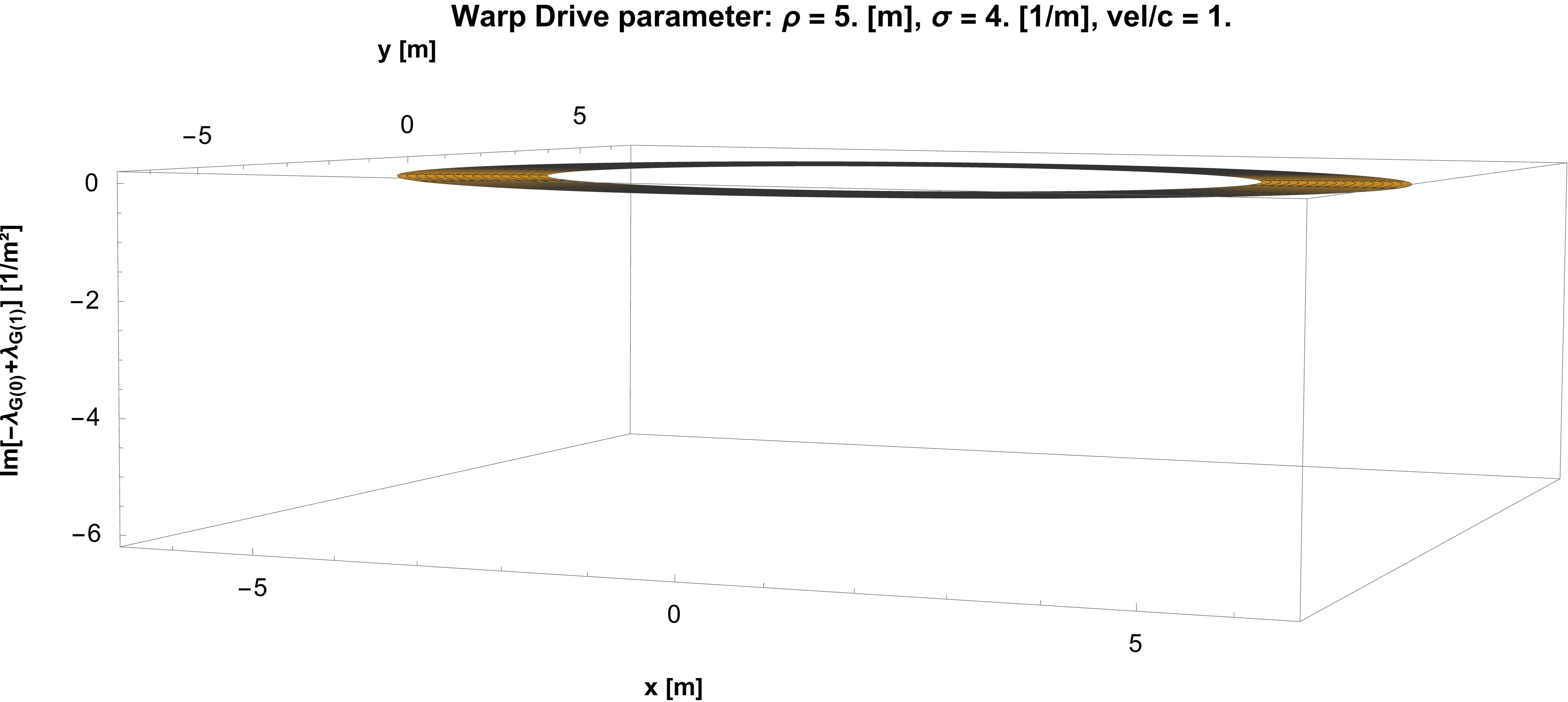}
        \caption{Irr.: $\operatorname{Im}[-\lambda_{G(0)}+\lambda_{G(1)}]$} 
        \label{fig:sub_r1_c2 B}
    \end{subfigure}
    \hfill 
    \begin{subfigure}[b]{0.32\textwidth} 
        \centering
        \includegraphics[width=\linewidth]{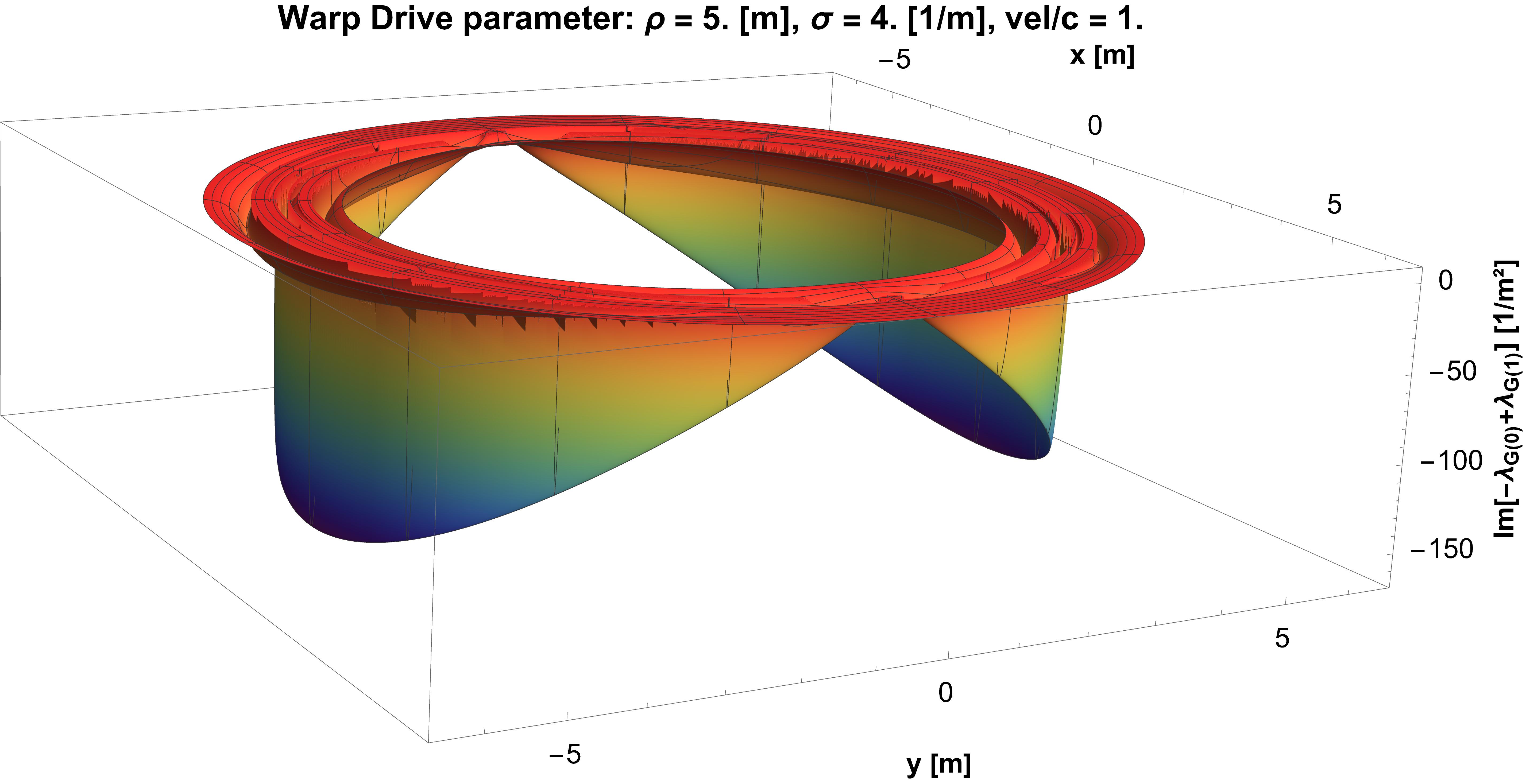}
        \caption{Nat.: $\operatorname{Im}[-\lambda_{G(0)}+\lambda_{G(1)}]$} 
        \label{fig:sub_r1_c3 B}
    \end{subfigure}

    \vspace{0.5cm} 

    \begin{subfigure}[b]{0.32\textwidth} 
        \centering
        \includegraphics[width=\linewidth]{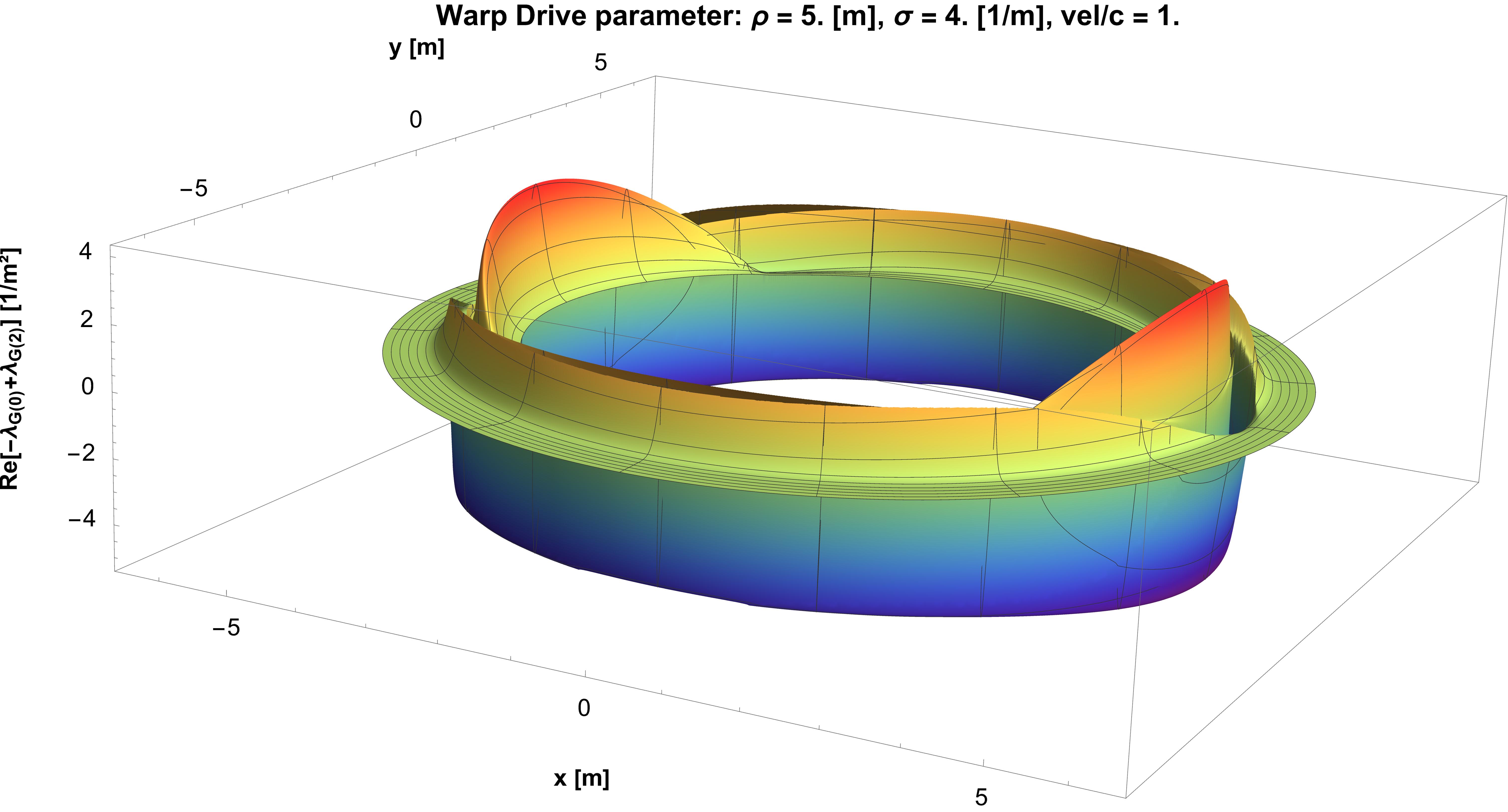}
        \caption{Alc.: $\operatorname{Re}[-\lambda_{G(0)}+\lambda_{G(2)}]$} 
        \label{fig:sub_r2_c1 B}
    \end{subfigure}
    \hfill 
    \begin{subfigure}[b]{0.32\textwidth}
        \centering
        \includegraphics[width=\linewidth]{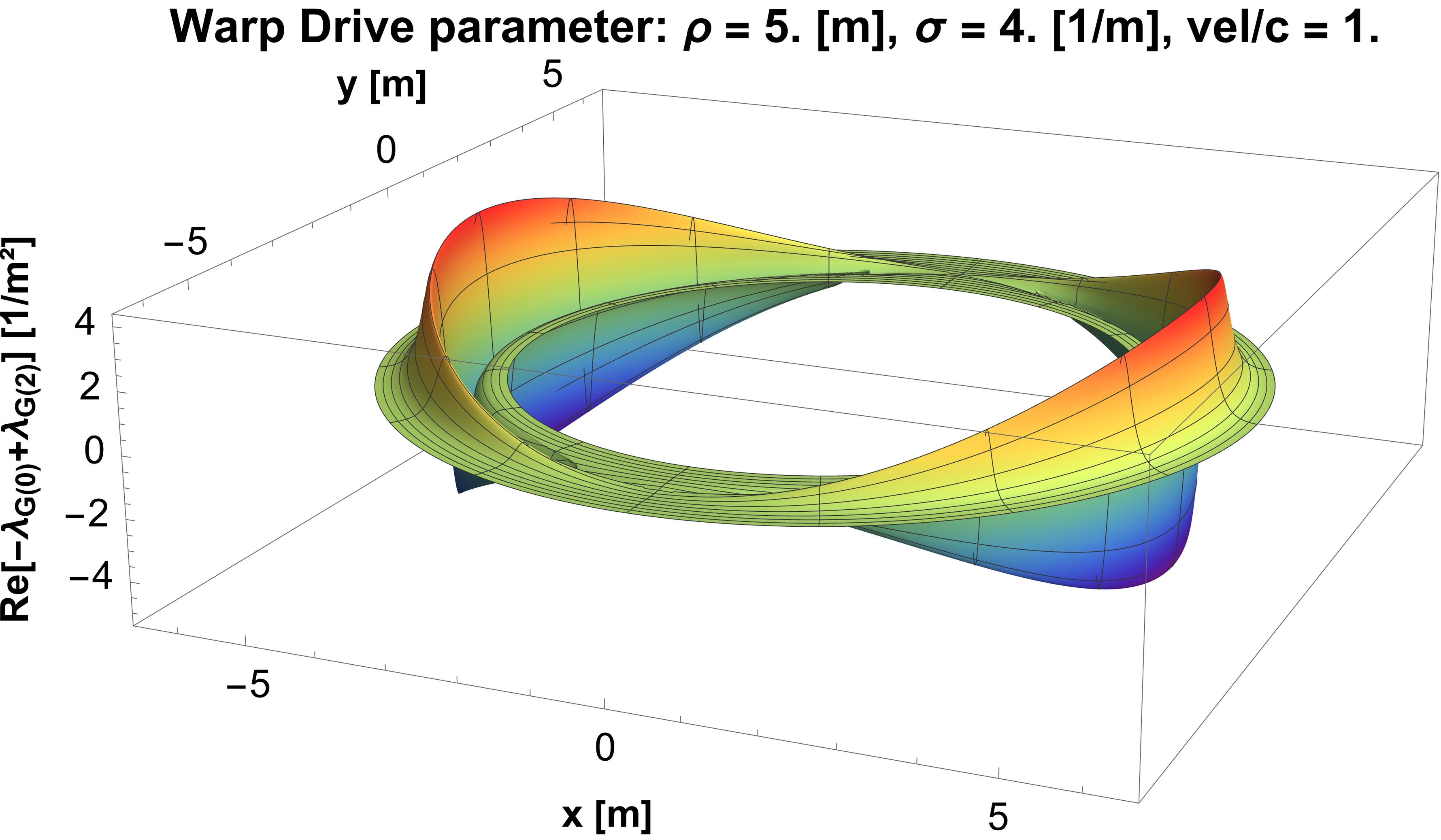}
        \caption{Irr.: $\operatorname{Re}[-\lambda_{G(0)}+\lambda_{G(2)}]$} 
        \label{fig:sub_r2_c2 B}
    \end{subfigure}
    \hfill 
    \begin{subfigure}[b]{0.32\textwidth}
        \centering
        \includegraphics[width=\linewidth]{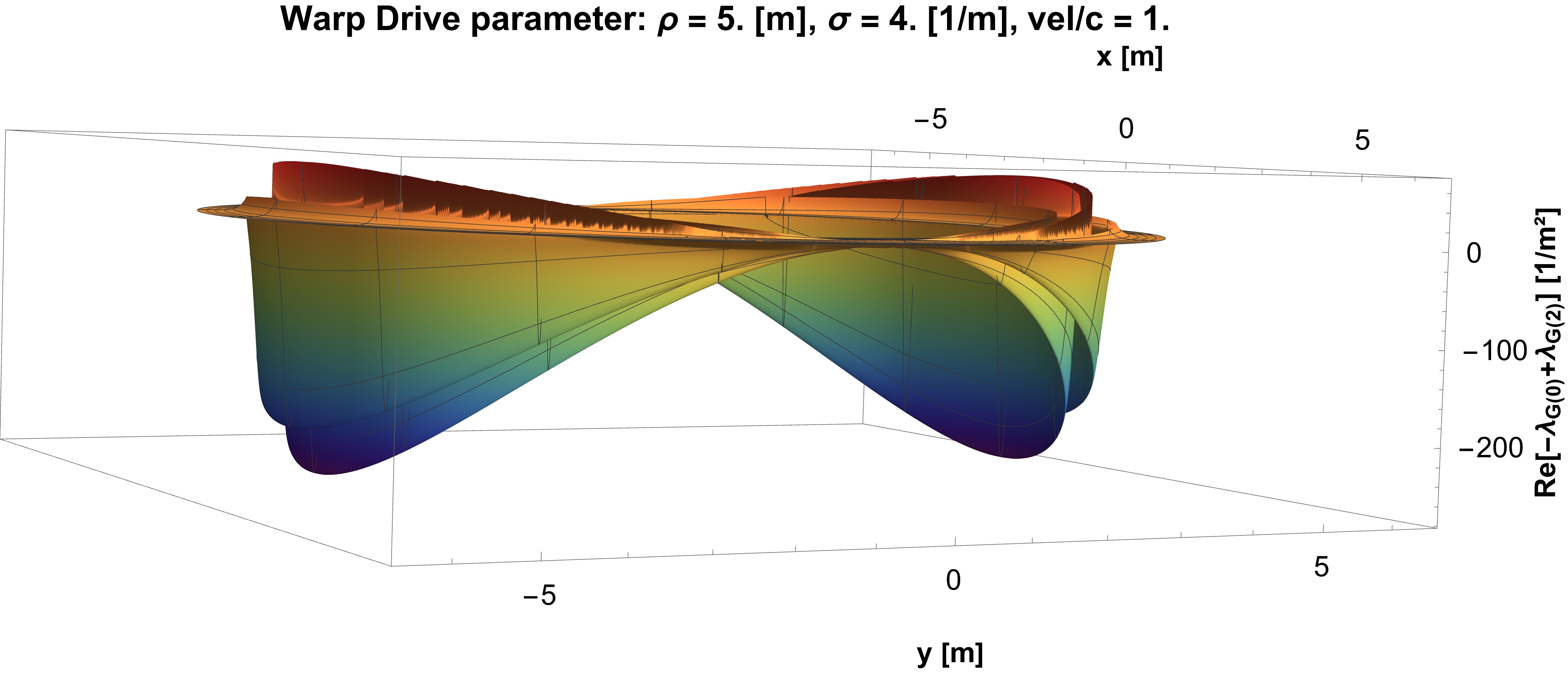}
        \caption{Nat.: $\operatorname{Re}[-\lambda_{G(0)}+\lambda_{G(2)}]$}
        \label{fig:sub_r2_c3 B}
    \end{subfigure}

    \vspace{0.5cm} 

    \begin{subfigure}[b]{0.32\textwidth}
        \centering
        \includegraphics[width=\linewidth]{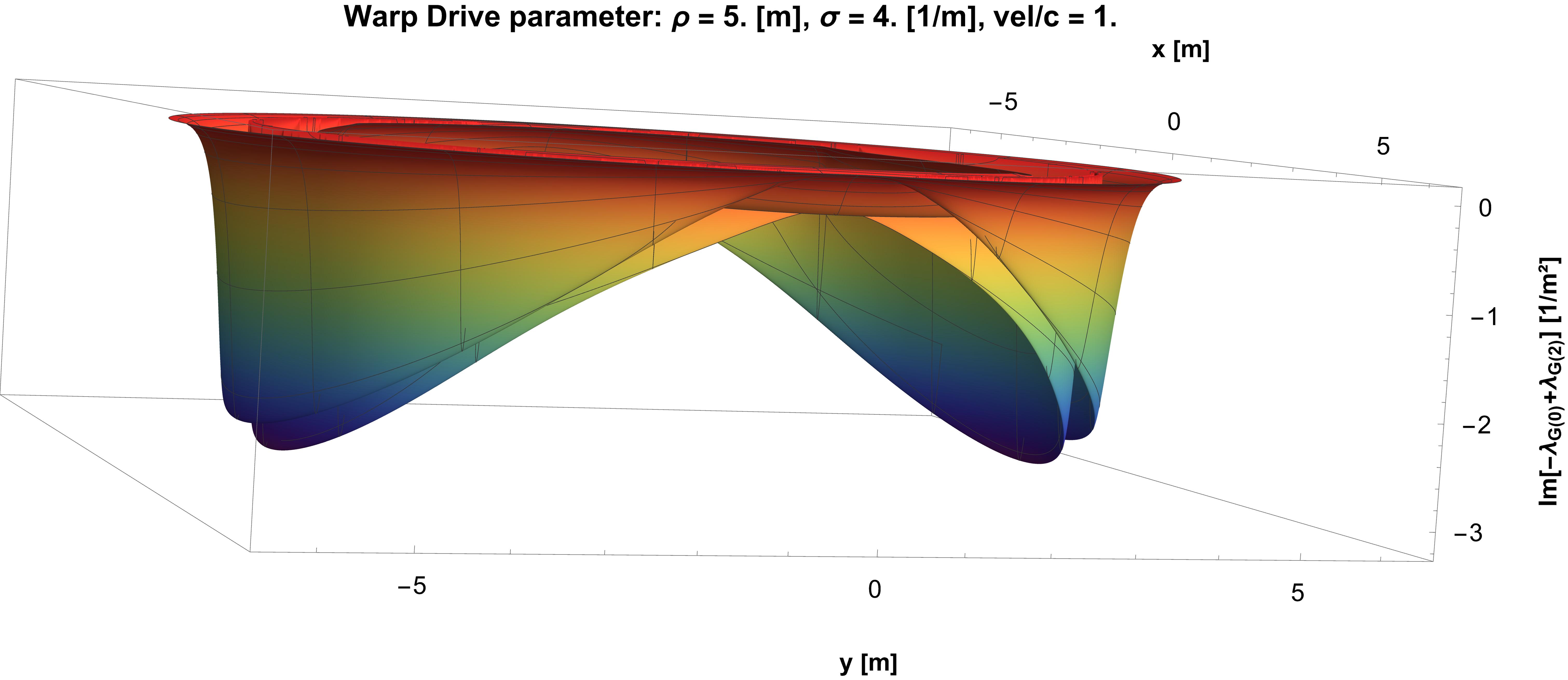}
        \caption{Alc.: $\operatorname{Im}[-\lambda_{G(0)}+\lambda_{G(2)}]$} 
        \label{fig:sub_r3_c1 B}
    \end{subfigure}
    \hfill 
    \begin{subfigure}[b]{0.32\textwidth}
        \centering
        \includegraphics[width=\linewidth]{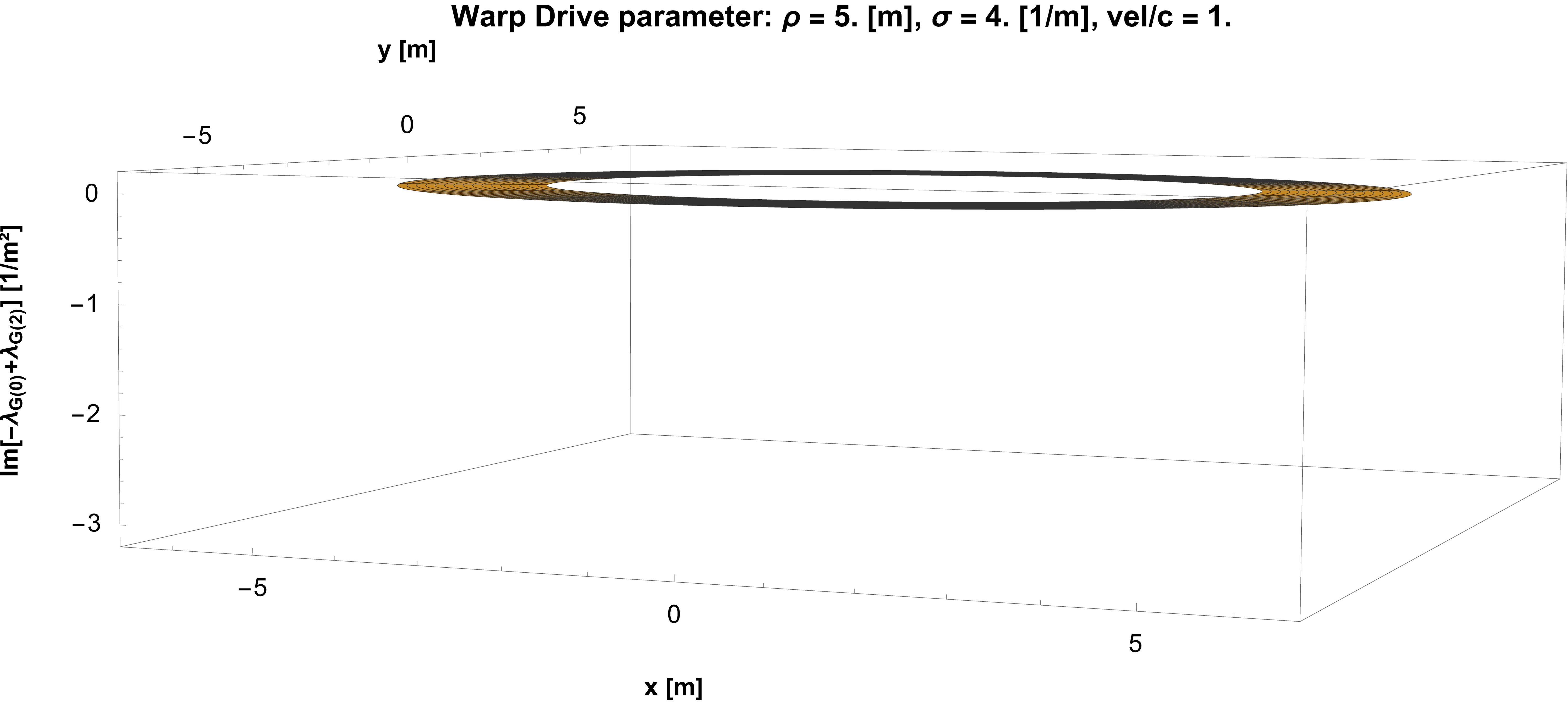}
        \caption{Irr.: $\operatorname{Im}[-\lambda_{G(0)}+\lambda_{G(2)}]$} 
        \label{fig:sub_r3_c2 B}
    \end{subfigure}
    \hfill 
    \begin{subfigure}[b]{0.32\textwidth}
        \centering
        \includegraphics[width=\linewidth]{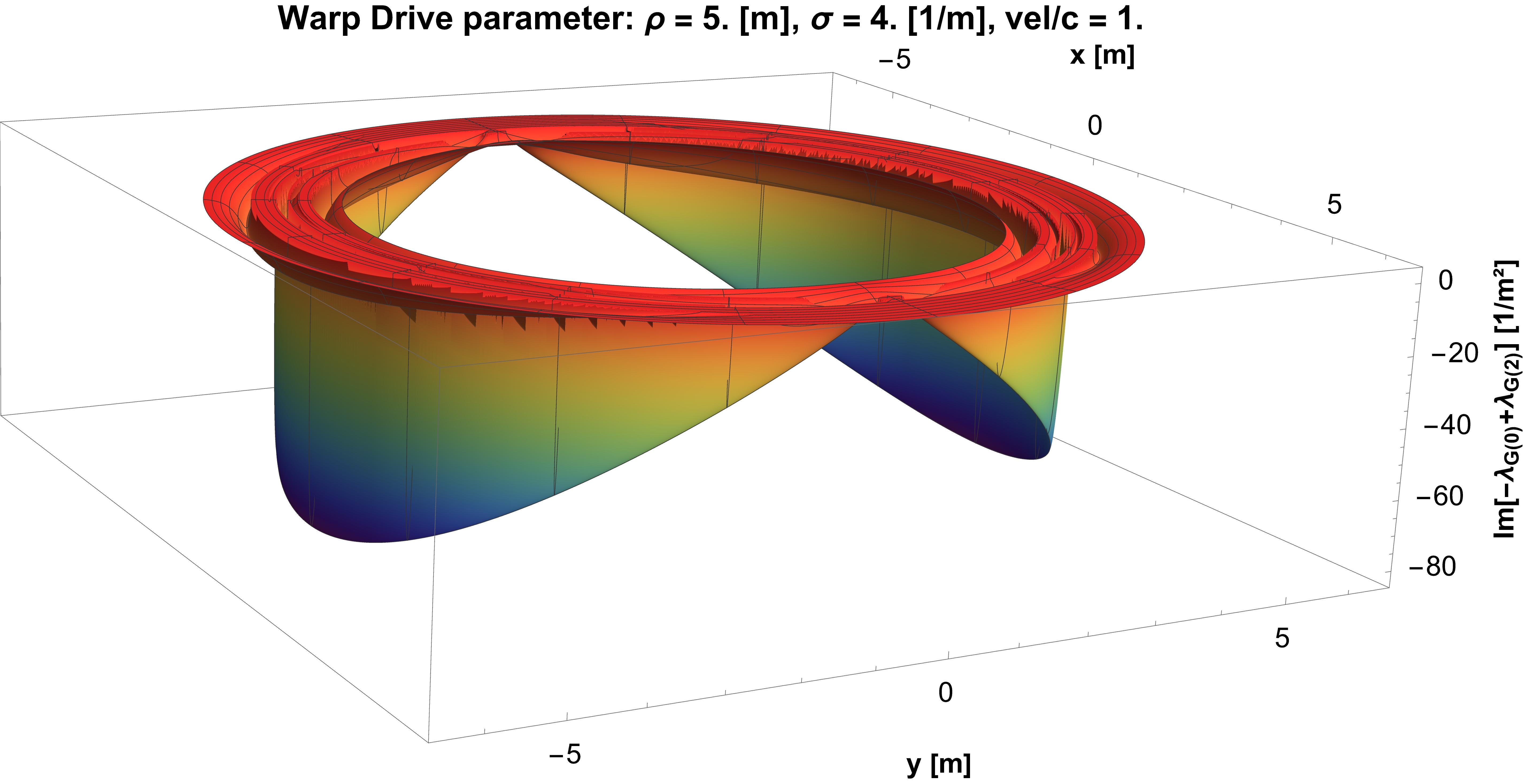}
        \caption{Nat.: $\operatorname{Im}[-\lambda_{G(0)}+\lambda_{G(2)}]$} 
        \label{fig:sub_r3_c3 B}
    \end{subfigure}

    \caption{Distribution of eigenvalue combinations relevant to energy condition analysis. Plotted against the axis of travel $x$ and transverse coordinate $y$ (coordinates in Fig.~\ref{Fig1_SpherCoord}). Columns compare Alcubierre, irrotational, and Natário warp drives. Rows display: (1)~$\Im[-\lambda_{G(0)}+\lambda_{G(1)}]$; (2)~$\Re[-\lambda_{G(0)}+\lambda_{G(2)}]$; (3)~$\Im[-\lambda_{G(0)}+\lambda_{G(2)}]$. See text for null energy condition (NEC) interpretations.}
    \label{fig:main_3x3_grid B}

\end{figure}

\begin{figure}[htbp] 
    \centering 

    \begin{subfigure}[b]{0.32\textwidth} 
        \centering
        \includegraphics[width=\linewidth]{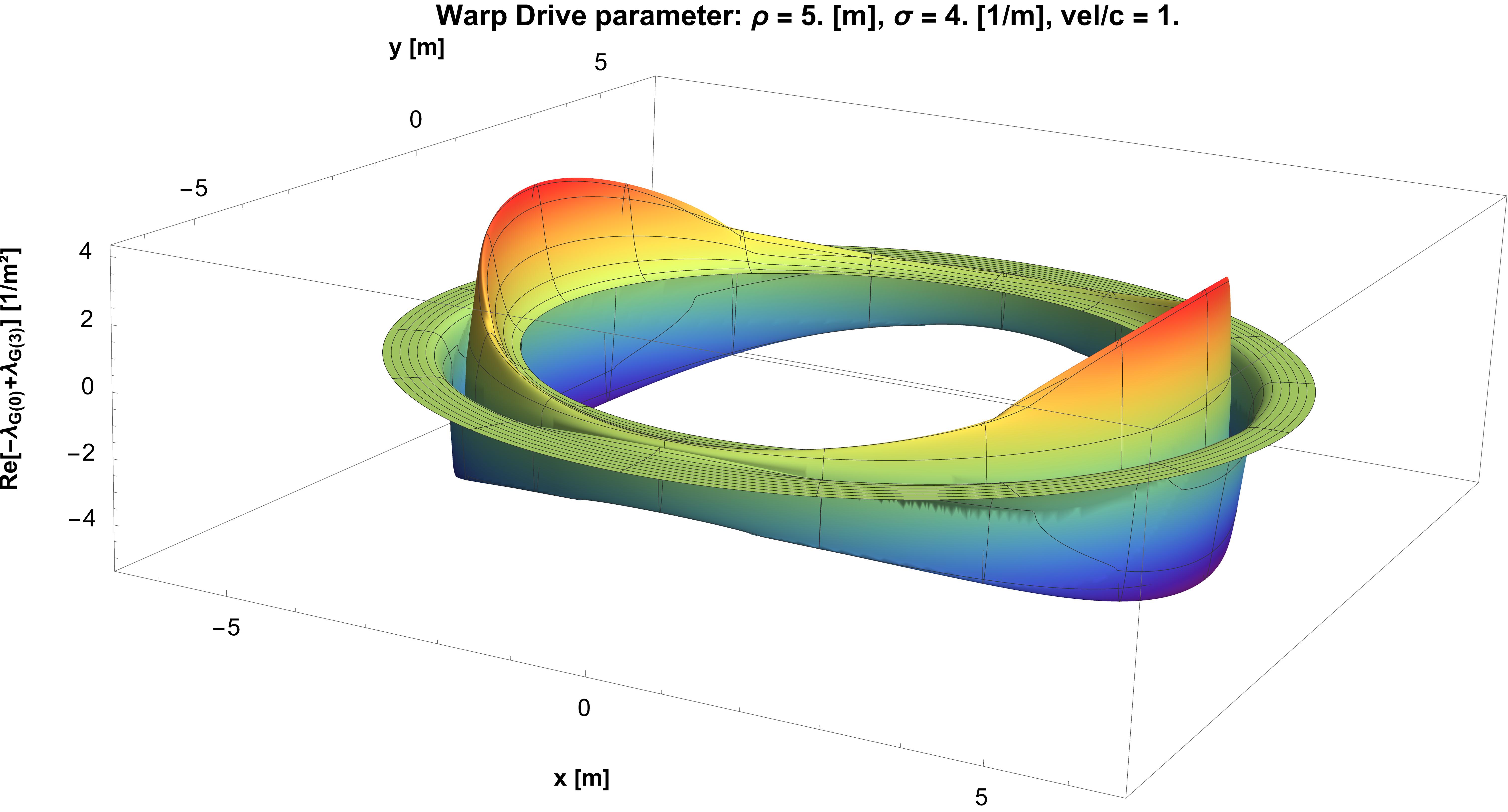}
        \caption{Alc.: $\operatorname{Re}[-\lambda_{G(0)}+\lambda_{G(3)}]$} 
        \label{fig:sub_r1_c1 C}
    \end{subfigure}
    \hfill 
    \begin{subfigure}[b]{0.32\textwidth}
        \centering
        \includegraphics[width=\linewidth]{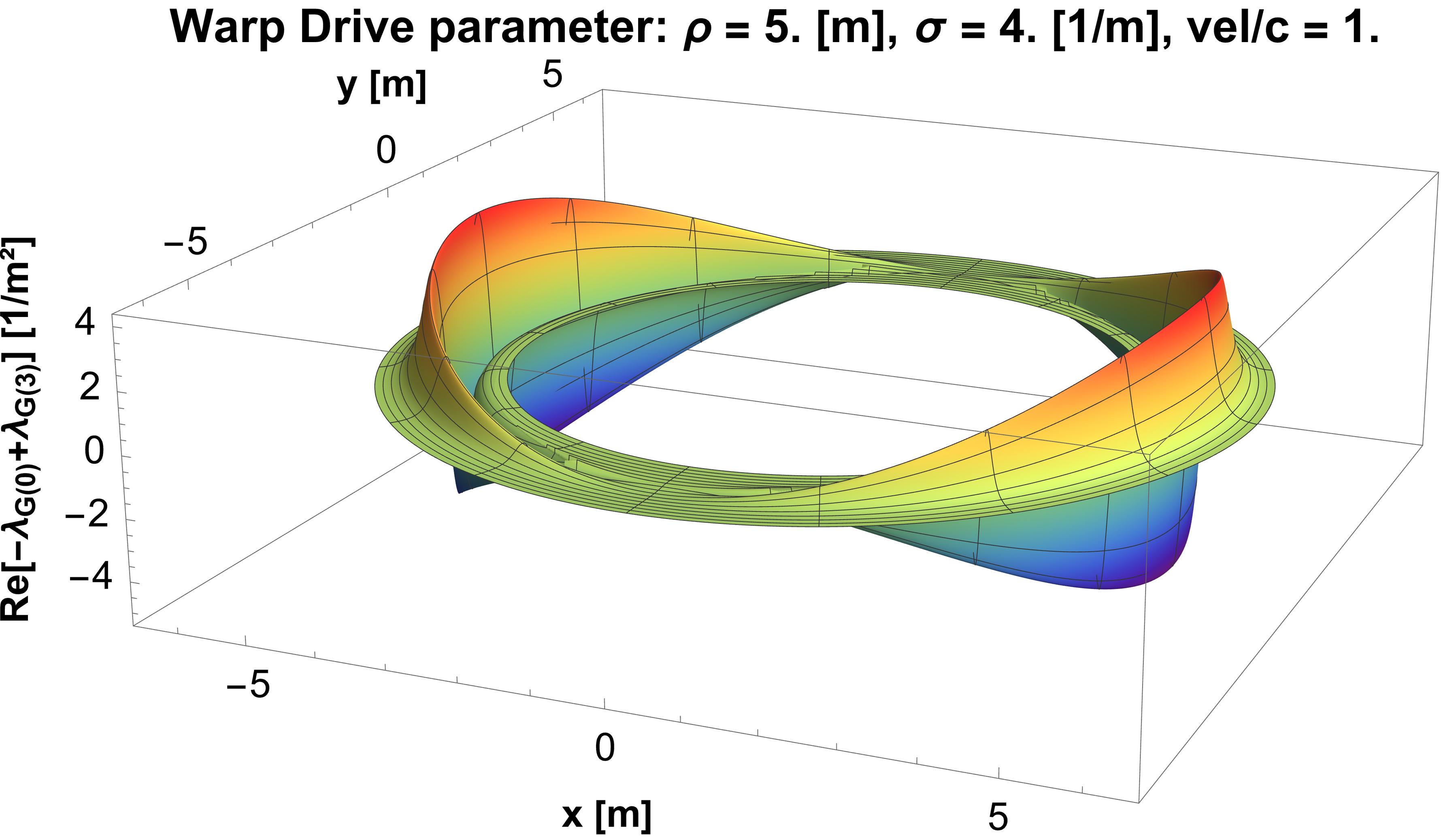}
        \caption{Irr.: $\operatorname{Re}[-\lambda_{G(0)}+\lambda_{G(3)}]$} 
        \label{fig:sub_r1_c2 C}
    \end{subfigure}
    \hfill 
    \begin{subfigure}[b]{0.32\textwidth}
        \centering
        \includegraphics[width=\linewidth]{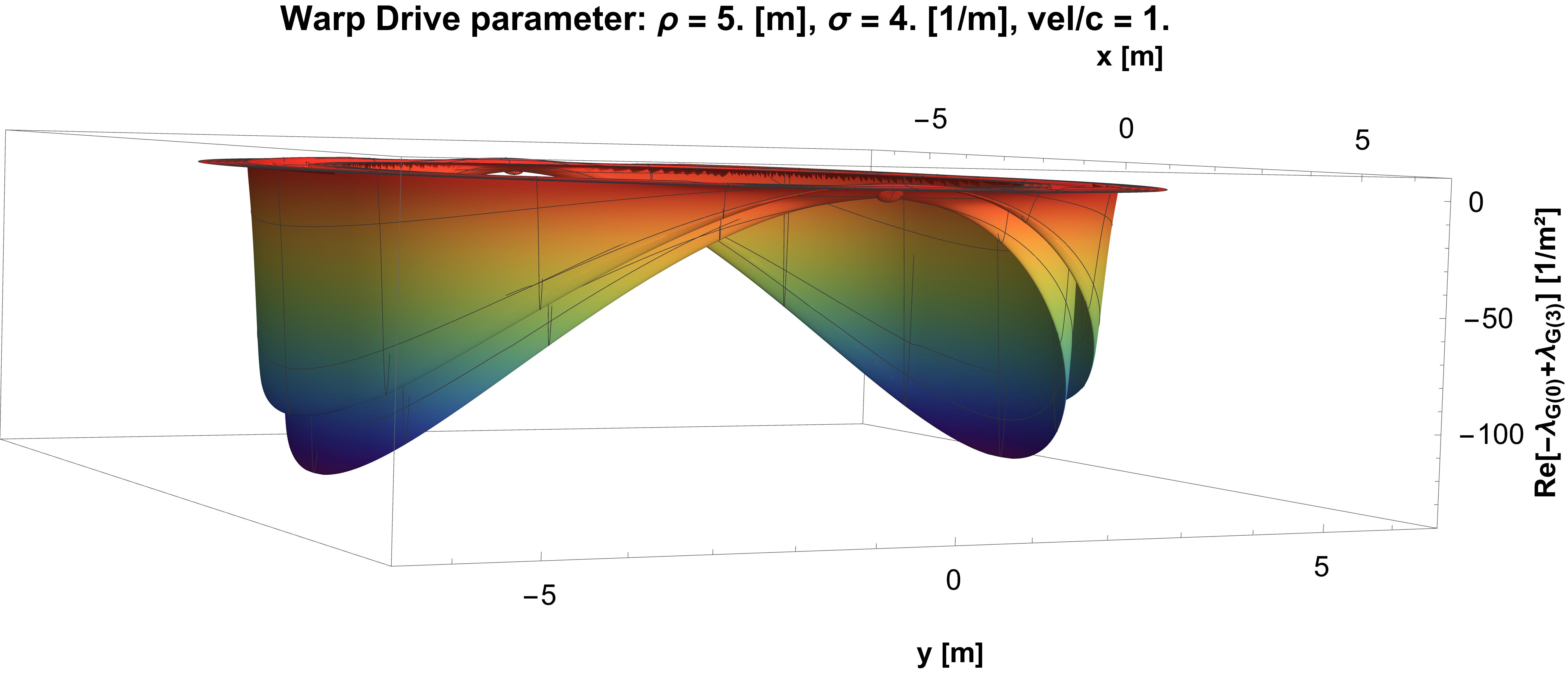}
        \caption{Nat.: $\operatorname{Re}[-\lambda_{G(0)}+\lambda_{G(3)}]$} 
        \label{fig:sub_r1_c3 C}
    \end{subfigure}

    \vspace{0.5cm} 

    \begin{subfigure}[b]{0.32\textwidth} 
        \centering
        \includegraphics[width=\linewidth]{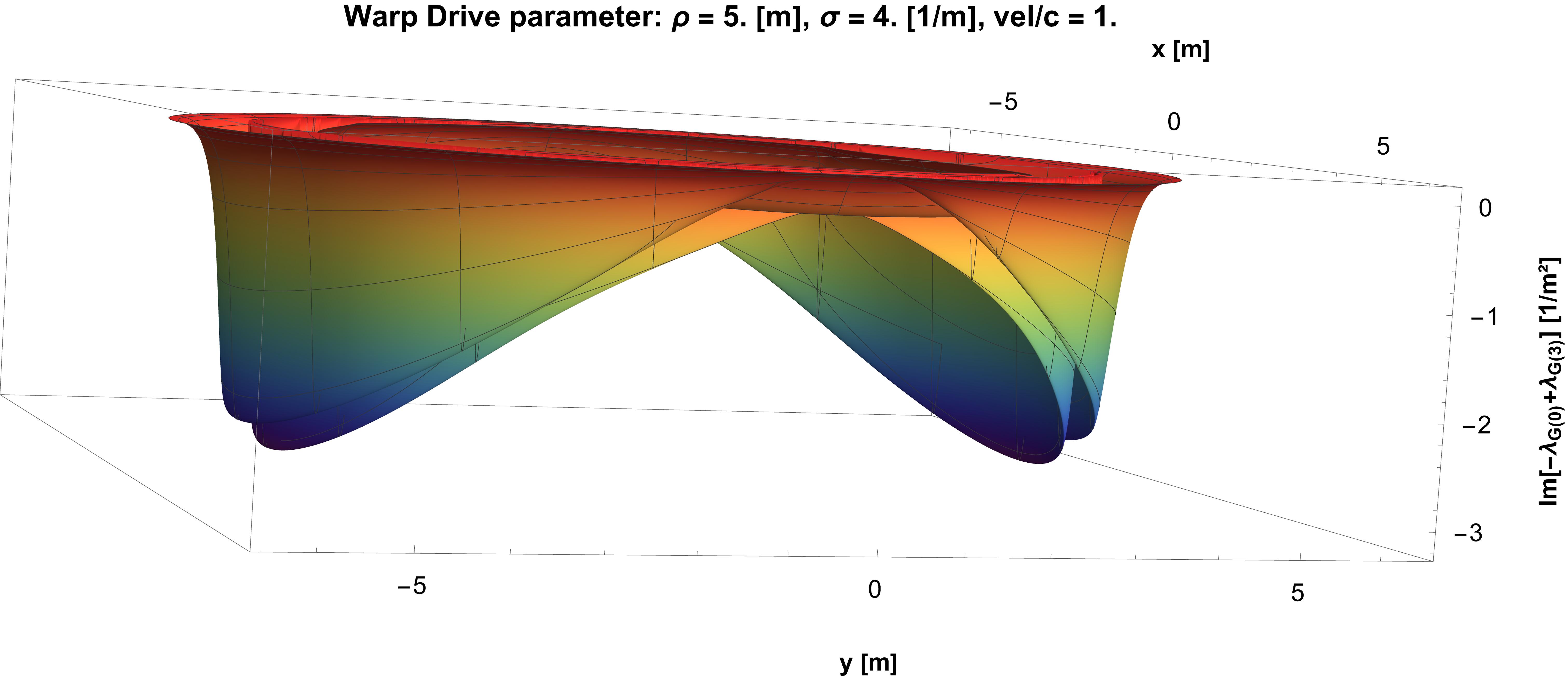}
        \caption{Alc.: $\operatorname{Im}[-\lambda_{G(0)}+\lambda_{G(3)}]$} 
        \label{fig:sub_r2_c1 C}
    \end{subfigure}
    \hfill 
    \begin{subfigure}[b]{0.32\textwidth}
        \centering
        \includegraphics[width=\linewidth]{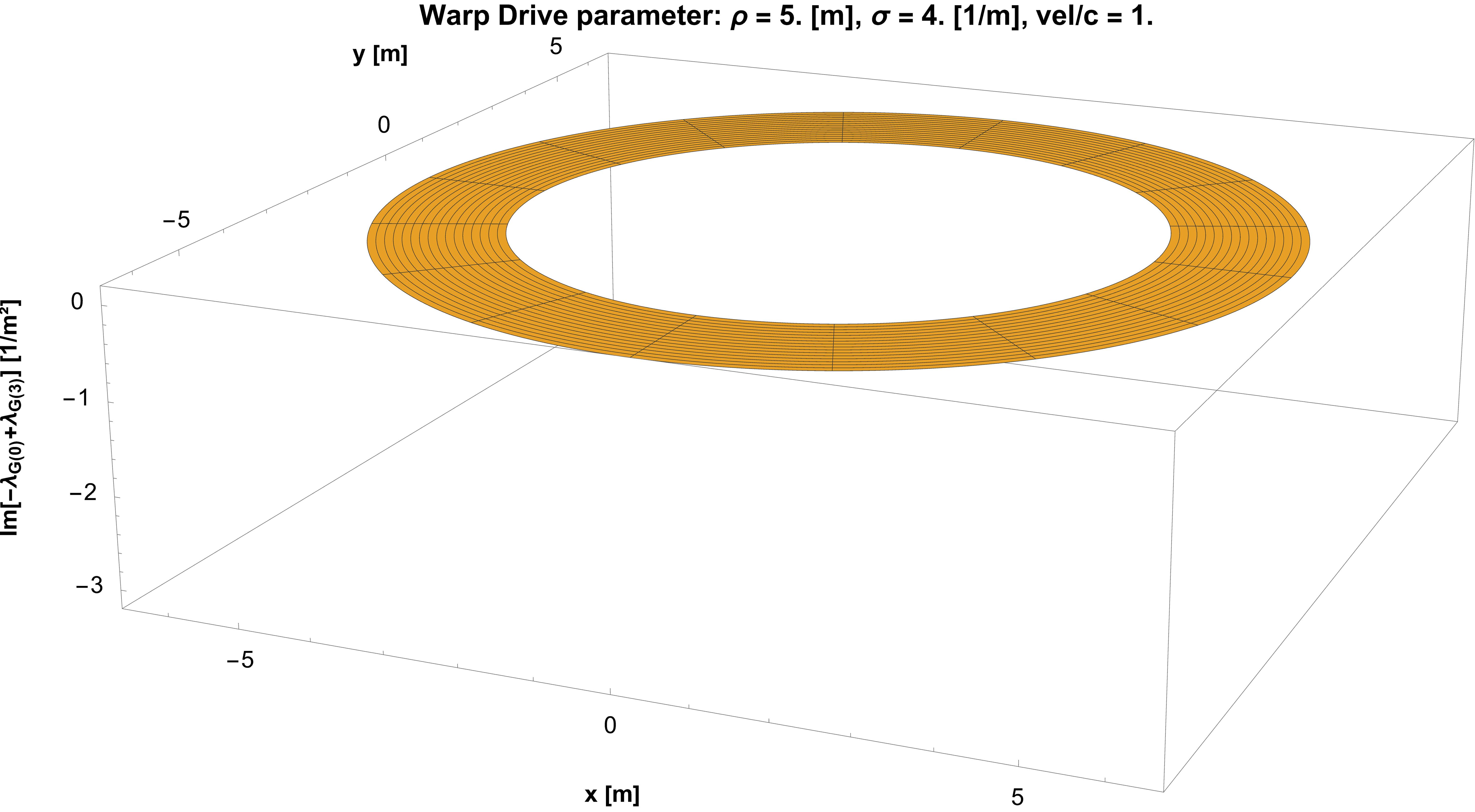}
        \caption{Irr.: $\operatorname{Im}[-\lambda_{G(0)}+\lambda_{G(3)}]$} 
        \label{fig:sub_r2_c2 C}
    \end{subfigure}
    \hfill 
    \begin{subfigure}[b]{0.32\textwidth}
        \centering
        \includegraphics[width=\linewidth]{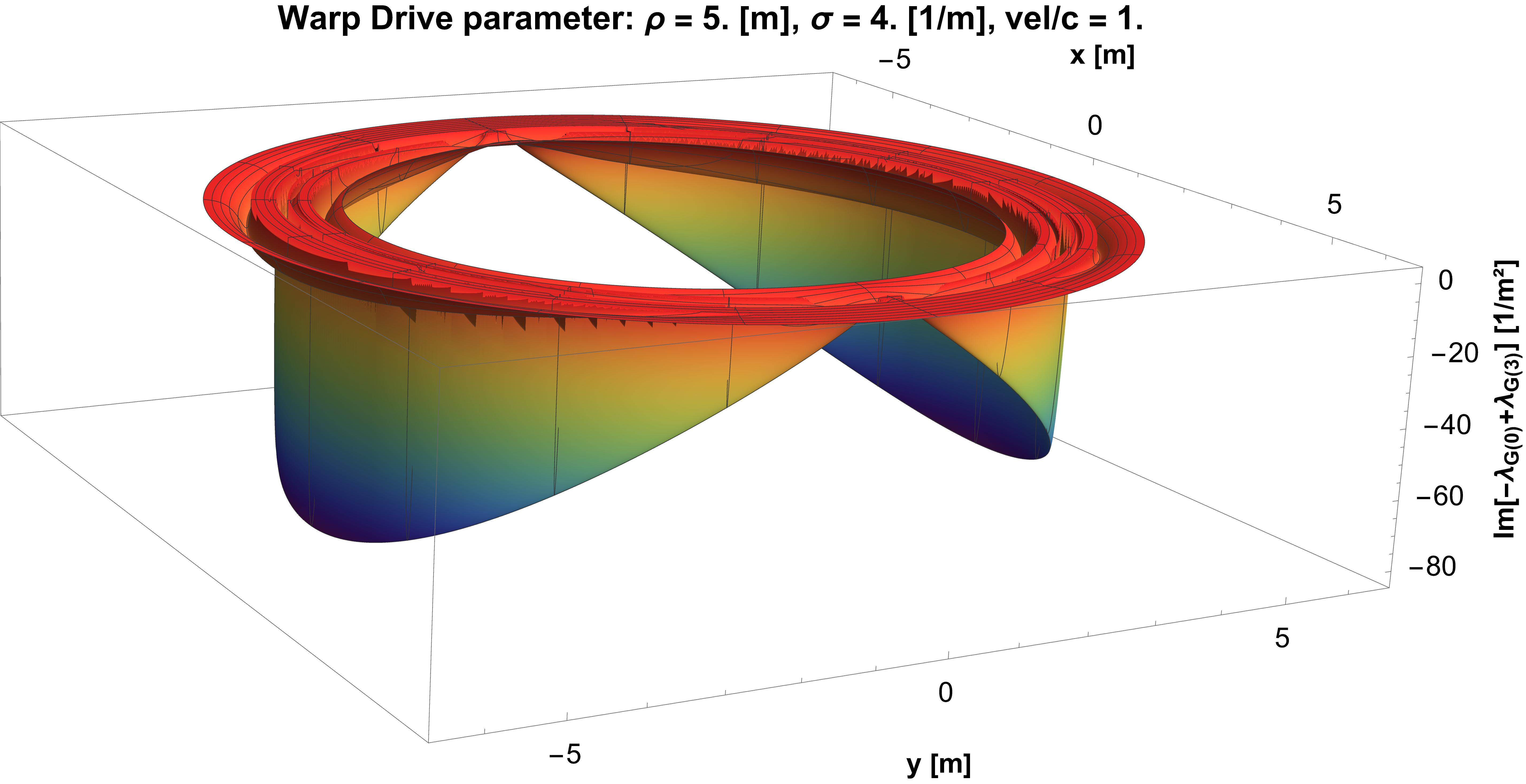}
        \caption{Nat.: $\operatorname{Im}[-\lambda_{G(0)}+\lambda_{G(3)}]$} 
        \label{fig:sub_r2_c3 C}
    \end{subfigure}

    \vspace{0.5cm} 

    \begin{subfigure}[b]{0.32\textwidth} 
        \centering
        \includegraphics[width=\linewidth]{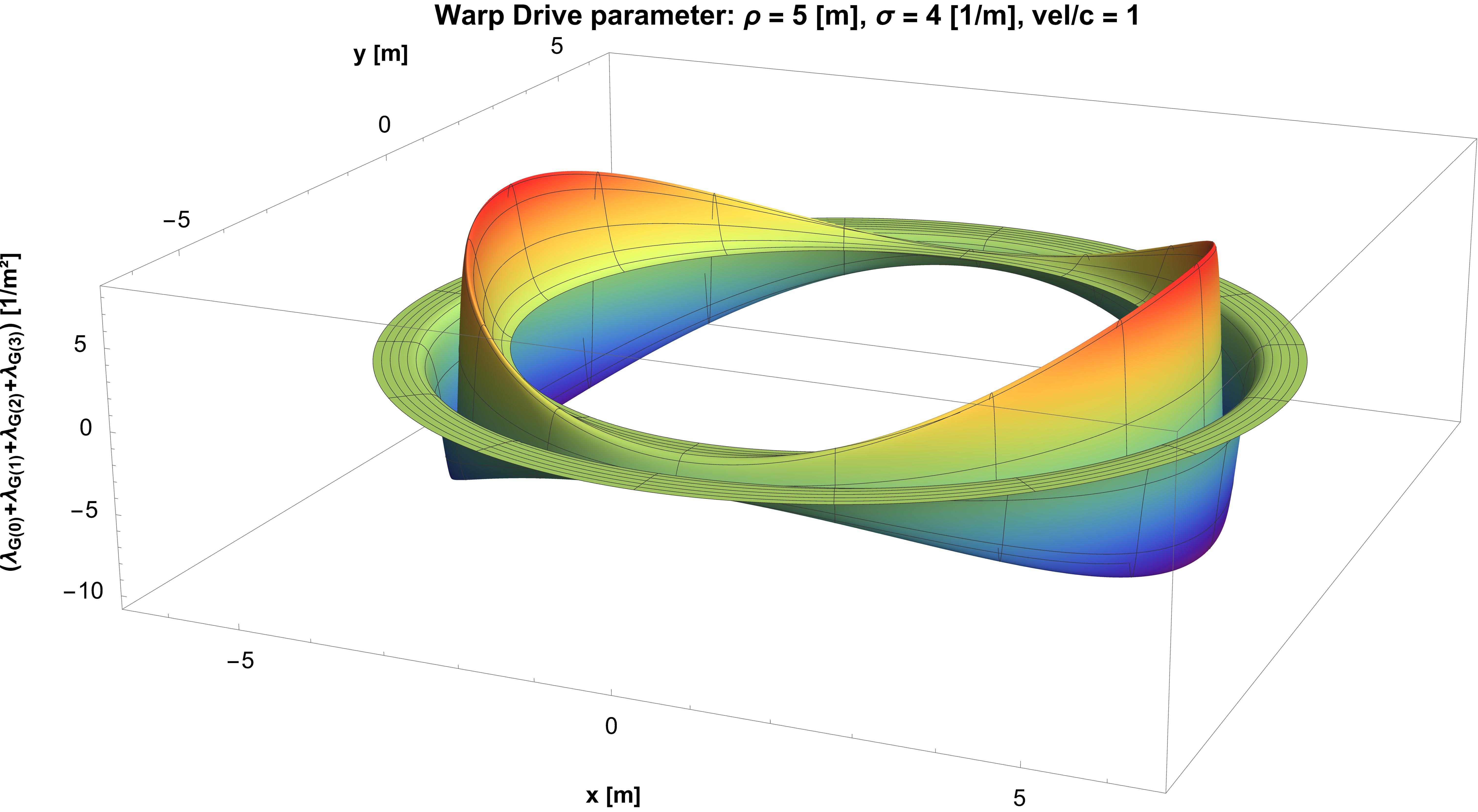}
        \caption{Alc. $G = \sum_{\sigma=0}^{3} \lambda_{G(\sigma)}$}
        \label{fig:sub_r3_c1 C}
    \end{subfigure}
    \hfill 
    \begin{subfigure}[b]{0.32\textwidth}
        \centering
        \includegraphics[width=\linewidth]{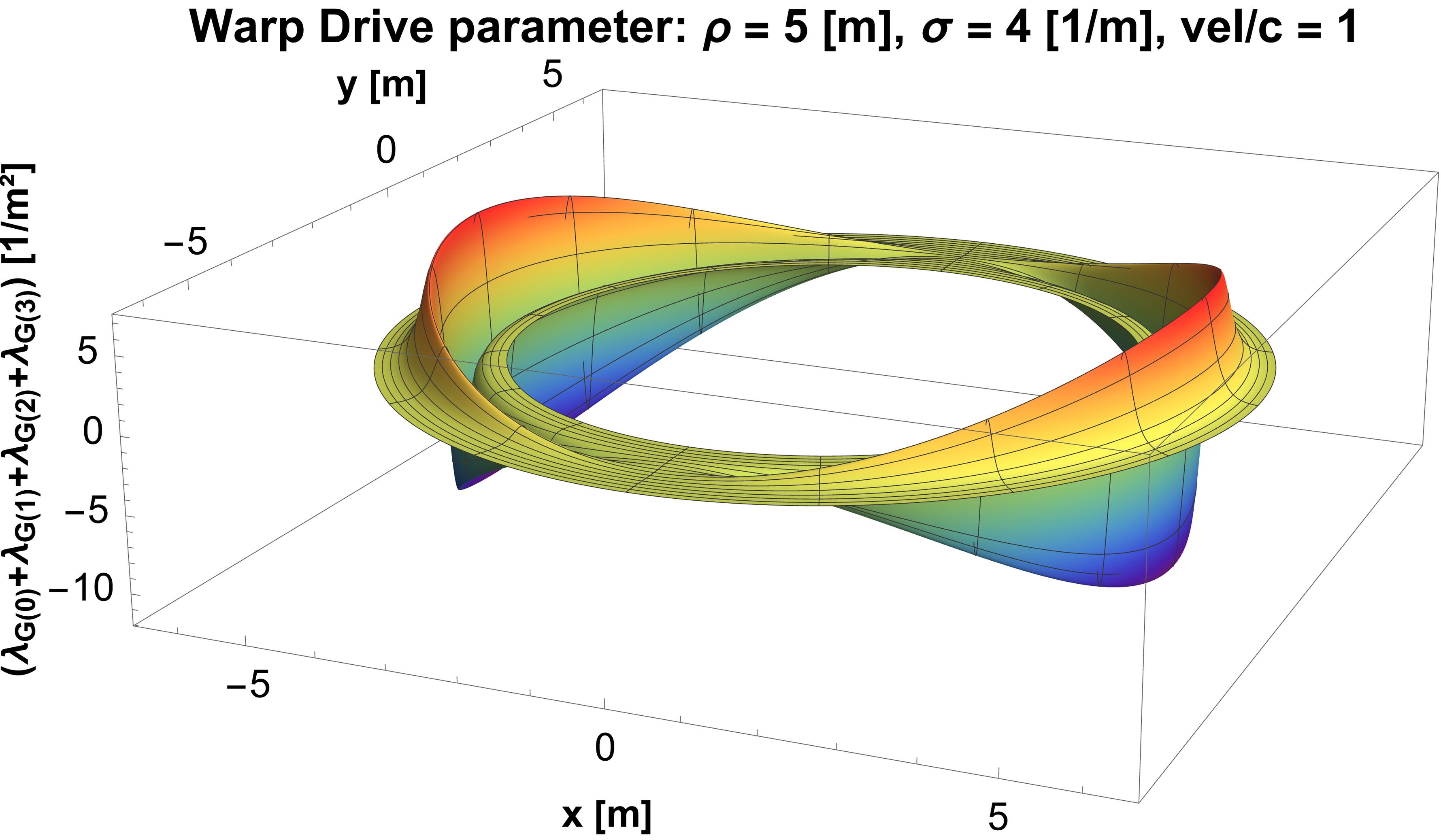}
        \caption{Irr. $G = \sum_{\sigma=0}^{3} \lambda_{G(\sigma)}$} 
        \label{fig:sub_r3_c2 C}
    \end{subfigure}
    \hfill 
    \begin{subfigure}[b]{0.32\textwidth}
        \centering
        \includegraphics[width=\linewidth]{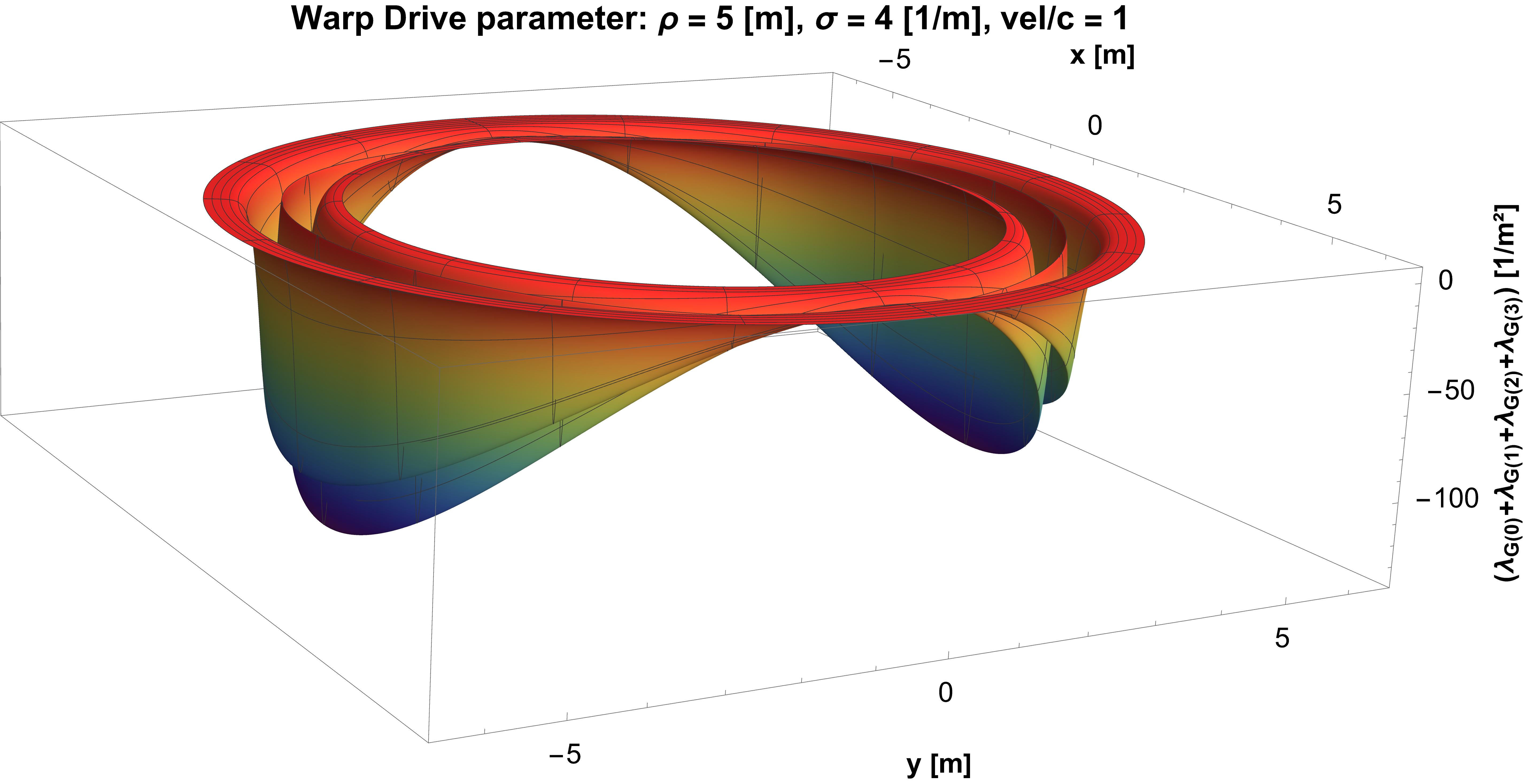}
        \caption{Nat. $G = \sum_{\sigma=0}^{3} \lambda_{G(\sigma)}$}
        \label{fig:sub_r3_c3 C}
    \end{subfigure}

    \caption{Distribution of further eigenvalue combinations. Plotted against the axis of travel $x$ and transverse coordinate $y$ (coordinates in Fig.~\ref{Fig1_SpherCoord}). Columns compare Alcubierre, irrotational, and Natário warp drives. Rows display: (1)~$\Re[-\lambda_{G(0)}+\lambda_{G(3)}]$; (2)~$\Im[-\lambda_{G(0)}+\lambda_{G(3)}]$; (3)~Trace invariant $G = \sum_{\sigma=0}^{3} \lambda_{G(\sigma)}$. See text for physical interpretation and energy condition analysis.}
    \label{fig:main_3x3_grid C} 

\end{figure}

\begin{figure}[htbp] 
    \centering 

    \begin{subfigure}[b]{0.32\textwidth} 
        \centering
        \includegraphics[width=\linewidth]{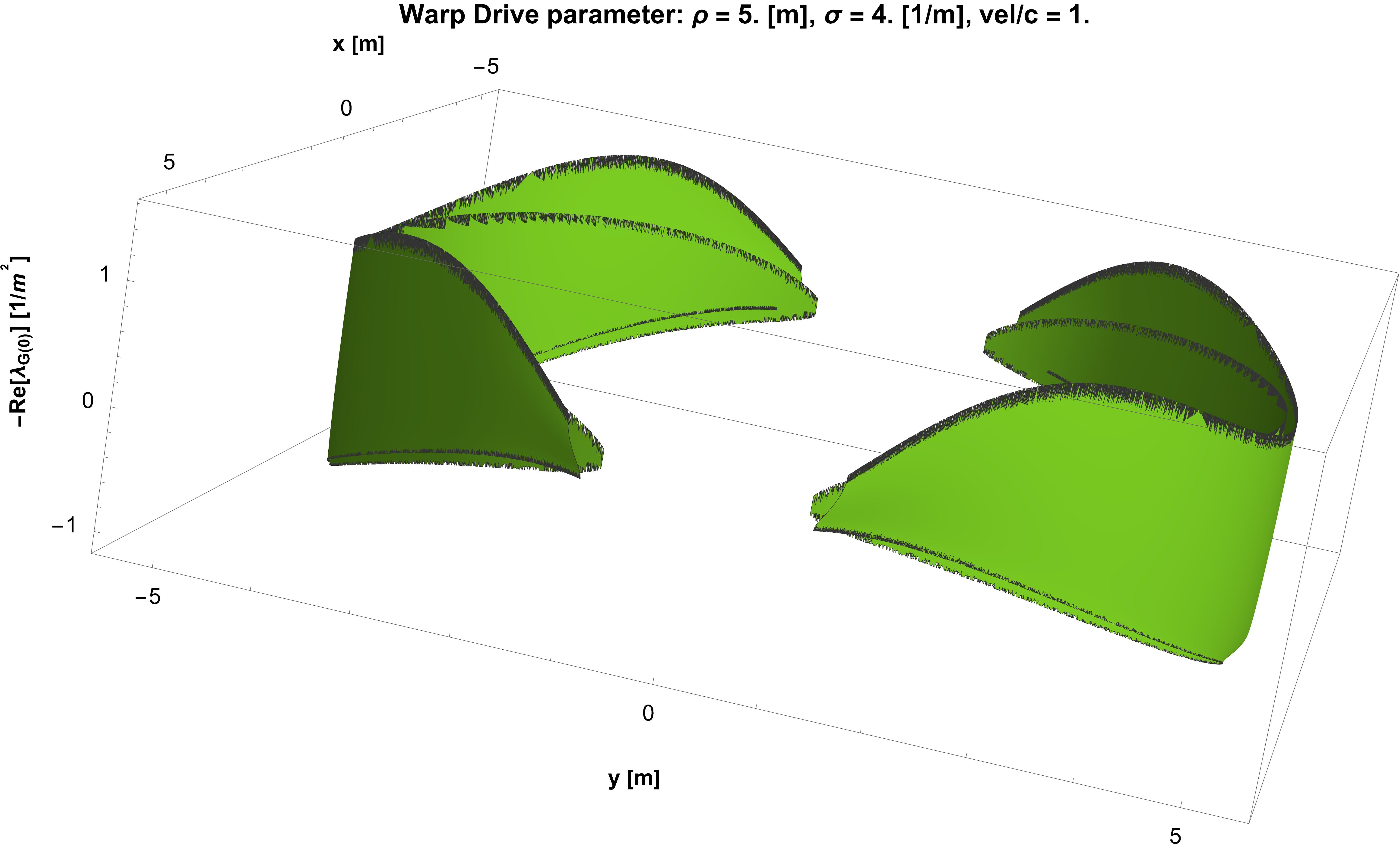}
        \caption{\centering Alc.: $-\operatorname{Re}[\lambda_{G(0)}]$ \\Type I} 
        \label{fig:HEtyp_r1_c1}
    \end{subfigure}
    \hfill 
    \begin{subfigure}[b]{0.32\textwidth}
        \centering
        \includegraphics[width=\linewidth]{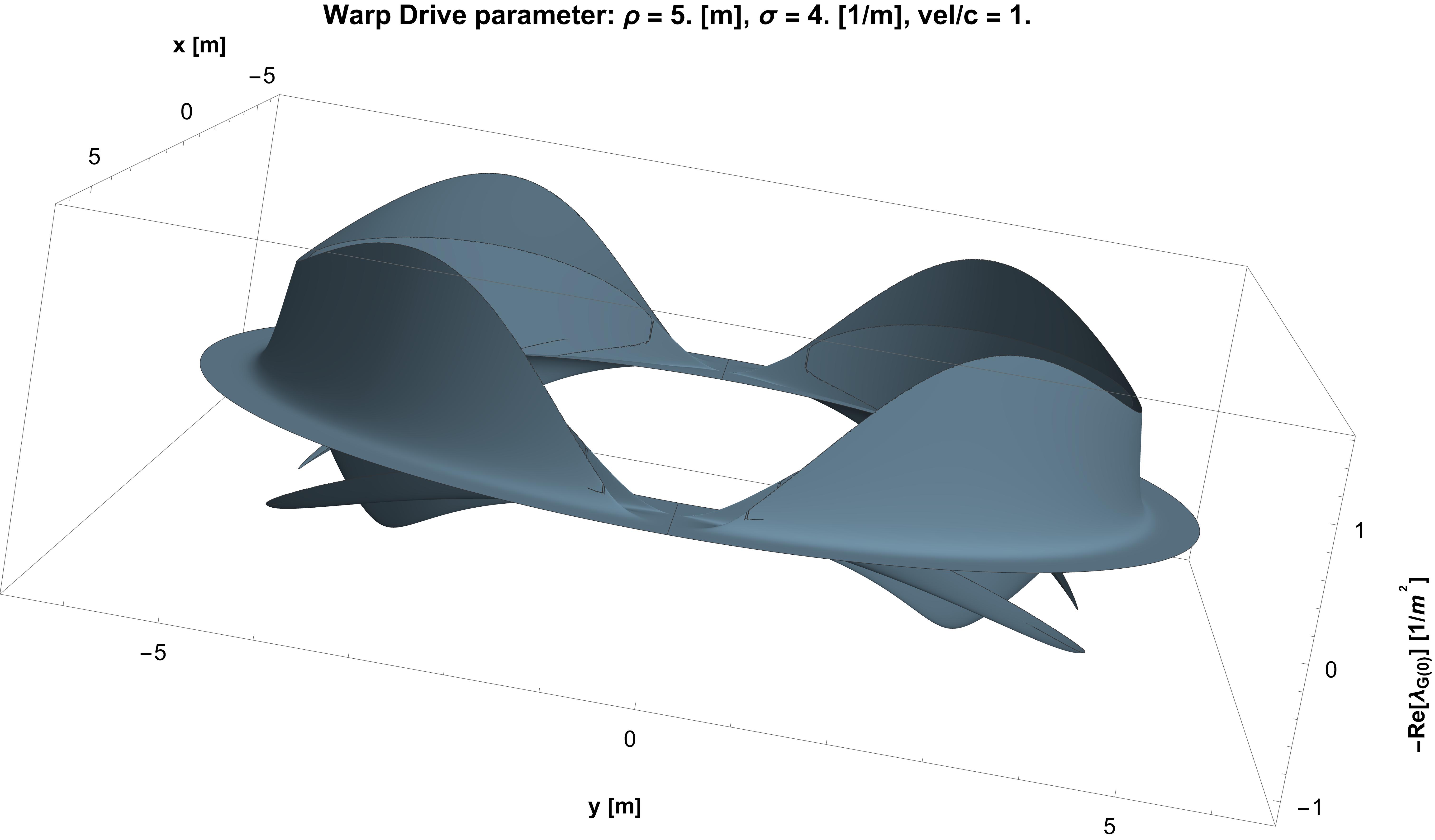}
        \caption{\centering Alc.: $-\operatorname{Re}[\lambda_{G(0)}]$ \\Type IV} 
        \label{fig:HEtyp_r1_c2}
    \end{subfigure}
    \hfill 
    \begin{subfigure}[b]{0.32\textwidth}
        \centering
        \includegraphics[width=\linewidth]{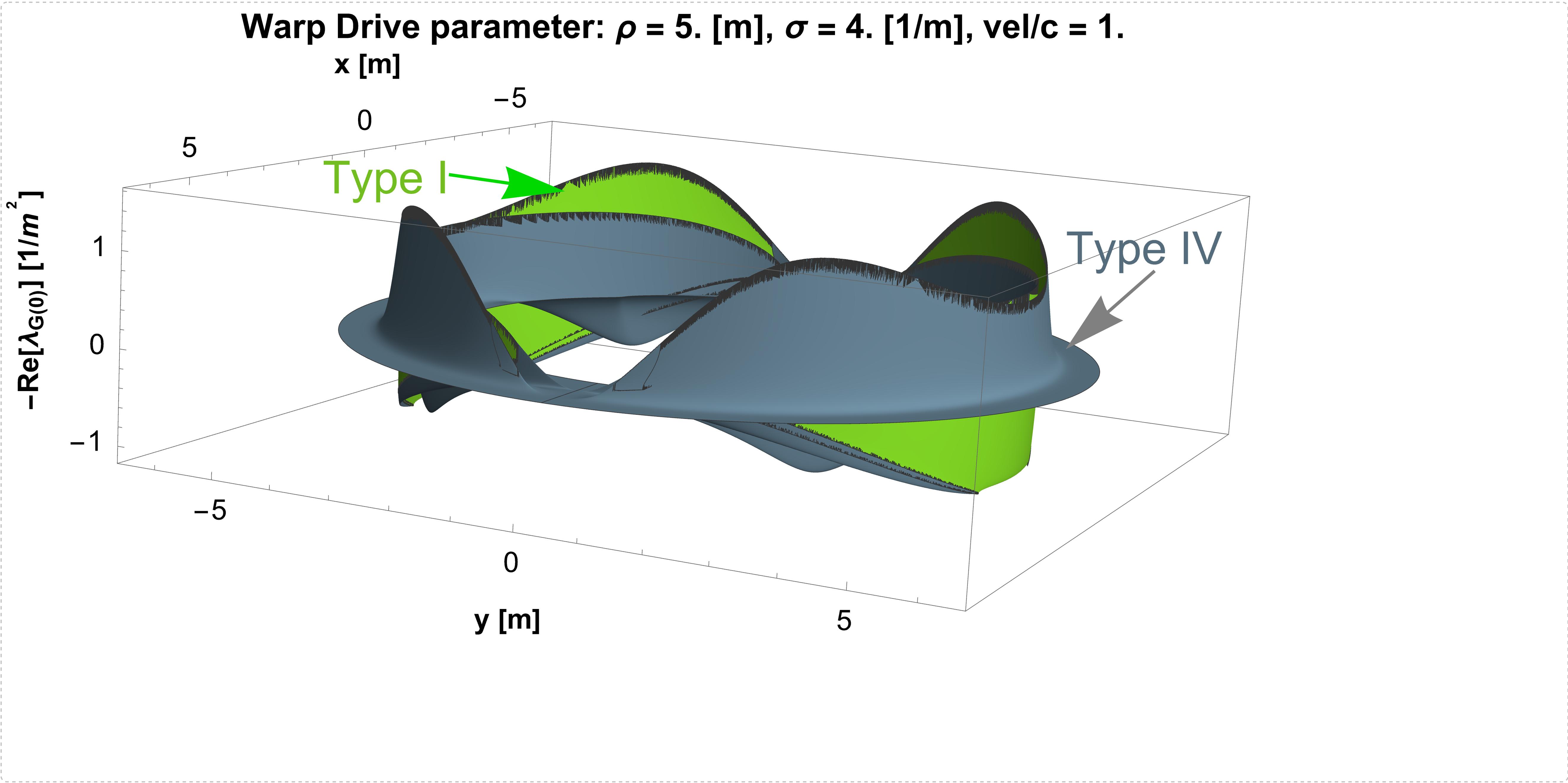}
        \caption{\centering Alc.: $-\operatorname{Re}[\lambda_{G(0)}]$ \\Type I \& IV} 
        \label{fig:HEtyp_r1_c3}
    \end{subfigure}

    \vspace{0.5cm} 

    \begin{subfigure}[b]{0.32\textwidth} 
        \centering
        \includegraphics[width=\linewidth]{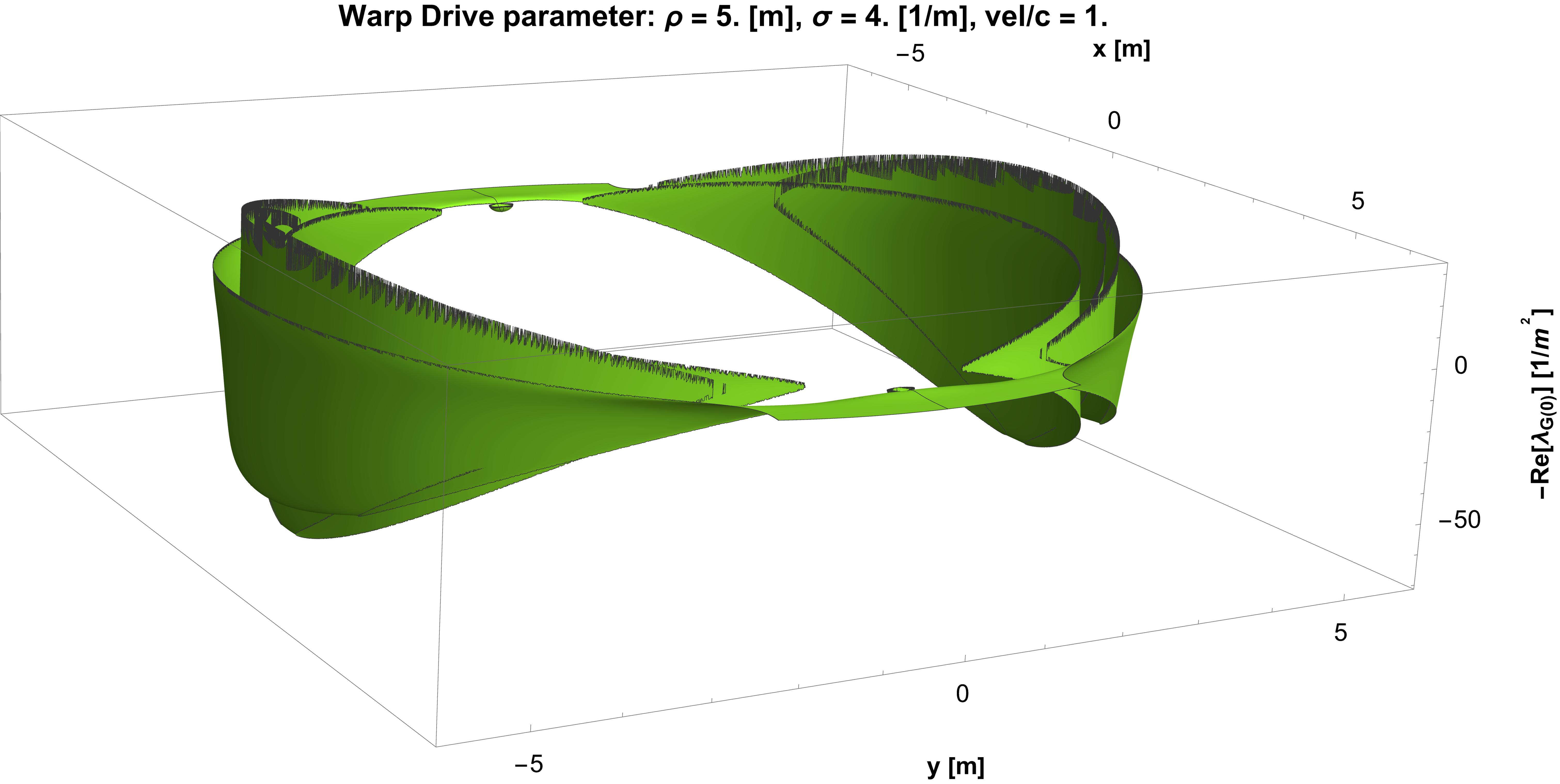}
        \caption{\centering Nat.: $-\operatorname{Re}[\lambda_{G(0)}]$ \\Type I} 
        \label{fig:HEtyp_r2_c1}
    \end{subfigure}
    \hfill 
    \begin{subfigure}[b]{0.32\textwidth}
        \centering
        \includegraphics[width=\linewidth]{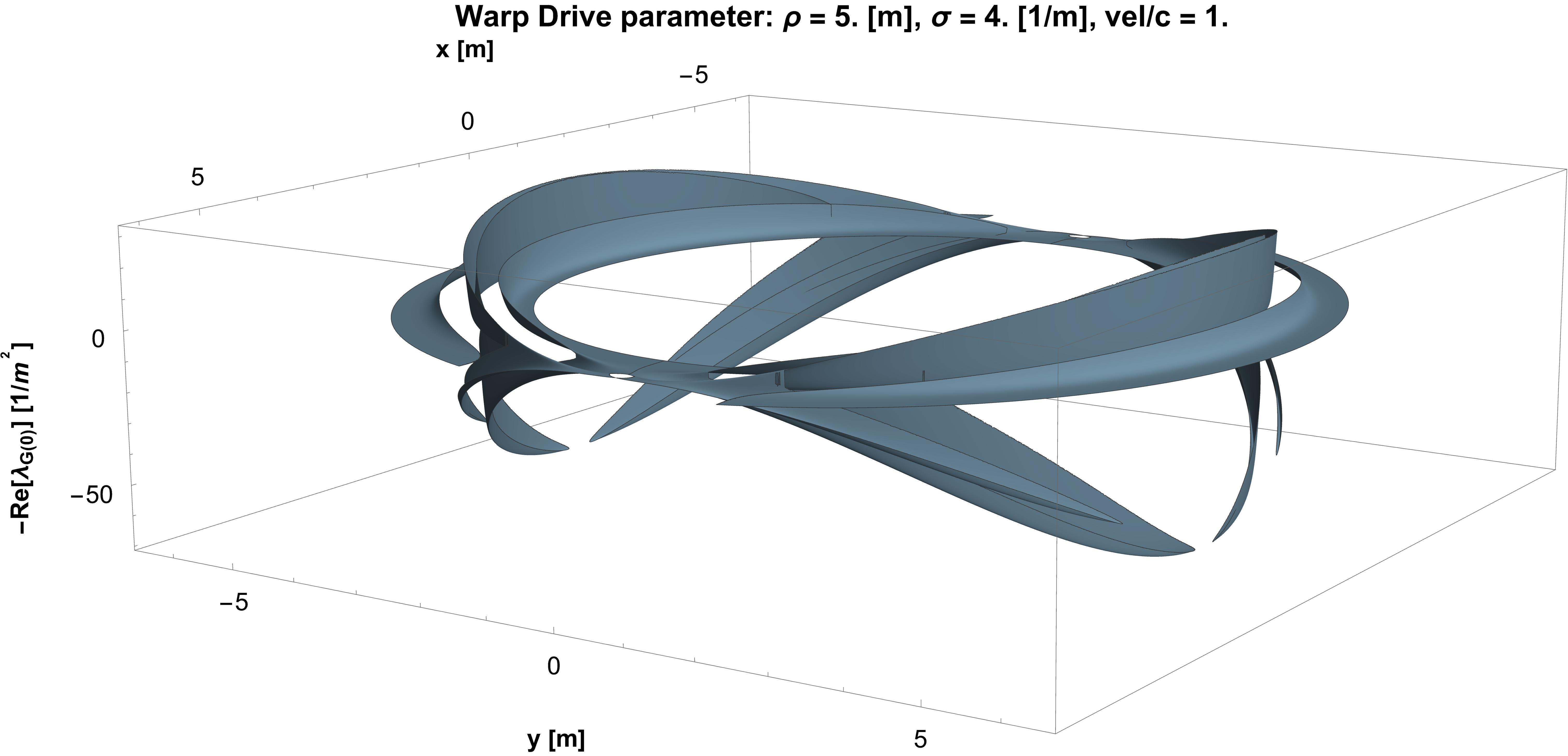}
        \caption{\centering Nat.: $-\operatorname{Re}[\lambda_{G(0)}]$ \\Type IV} 
        \label{fig:HEtyp_r2_c2}
    \end{subfigure}
    \hfill 
    \begin{subfigure}[b]{0.32\textwidth}
        \centering
        \includegraphics[width=\linewidth]{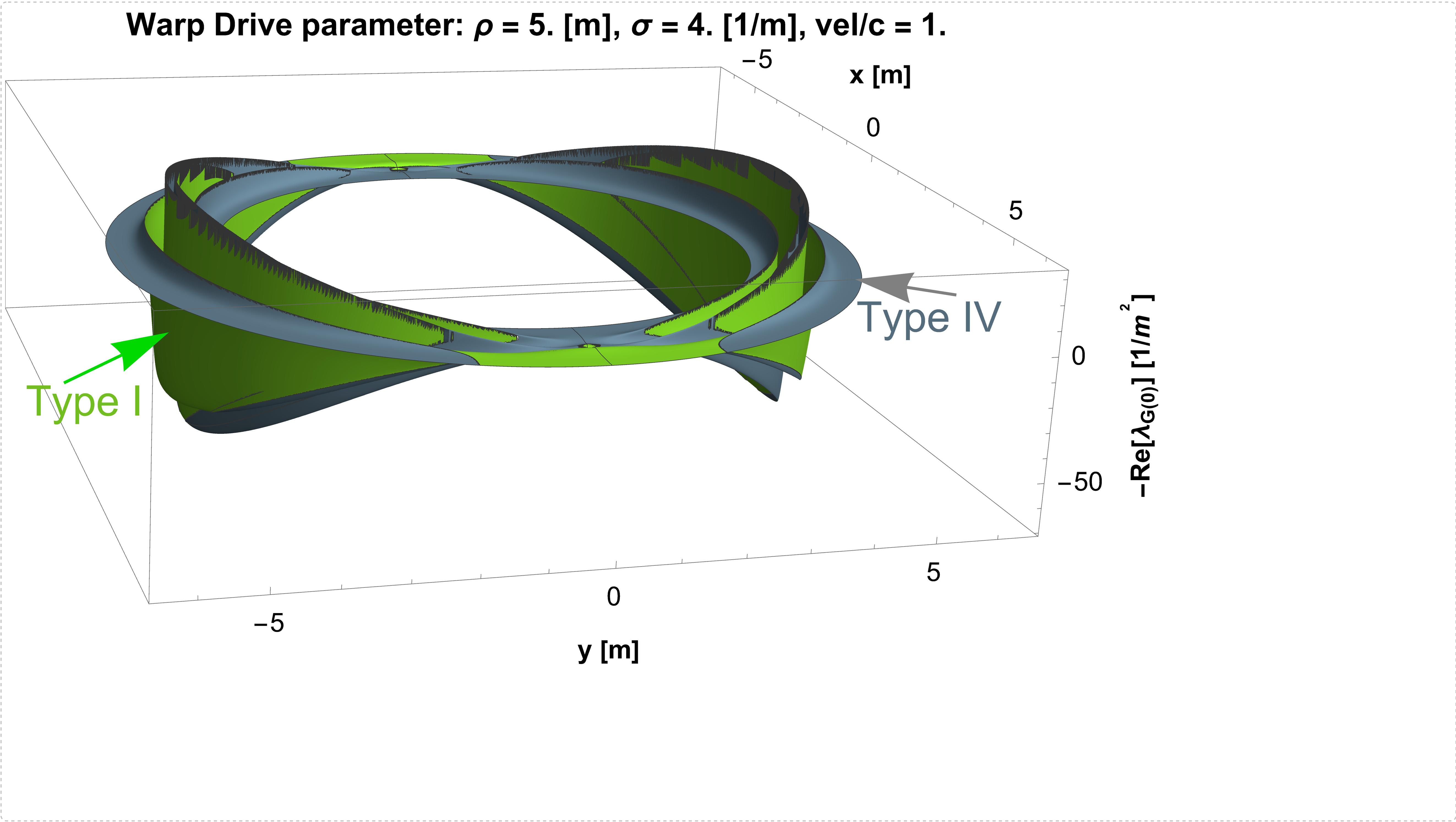}
        \caption{\centering Nat.: $-\operatorname{Re}[\lambda_{G(0)}]$ \\Type I \& IV} 
        \label{fig:HEtyp_r2_c3}
    \end{subfigure}

    \vspace{0.5cm} 

    \caption{Hawking–Ellis classification of the timelike-associated eigenvalue \(-\operatorname{Re}[\lambda_{G(0)}]\) (where \(\varrho_p = -\operatorname{Re}[\lambda_{G(0)}] /\kappa\) in Type I regions), plotted in the prime meridional $x$–$y$ plane (with $x$ the direction of motion and $y$ the transverse axis) (convention in Fig.~\ref{Fig1_SpherCoord}). Columns depict regions classified as: (Left) Type I; (Center) Type IV; (Right) overlays both types. Rows correspond to: (1) Alcubierre's warp drive; (2) Natário's warp drive.}
\label{fig:main_2x3_grid}

\end{figure}

\begin{figure}[htbp] 
    \centering 

    \begin{subfigure}[b]{0.48\textwidth} 
        \centering
        \includegraphics[width=\linewidth]{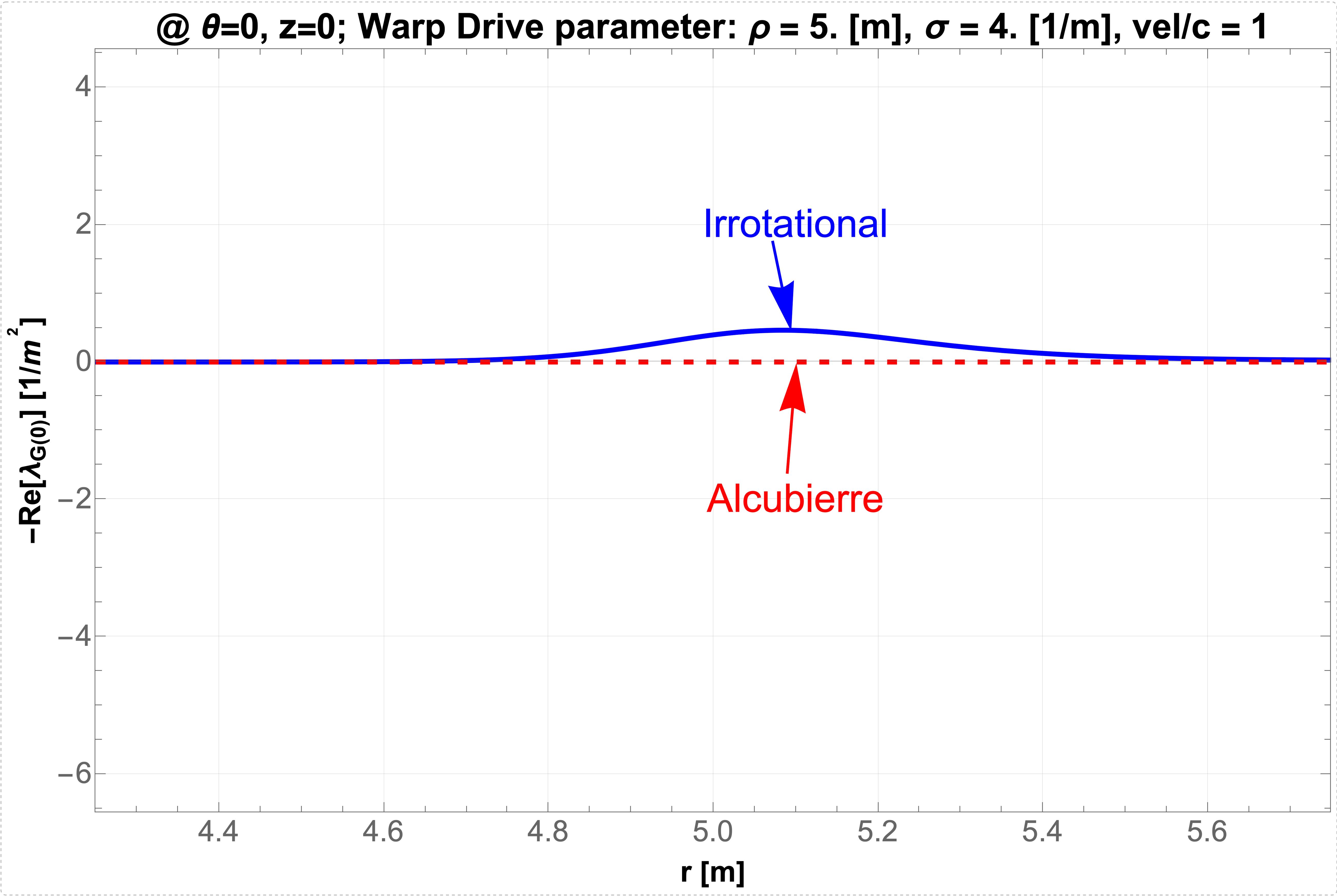}
        \caption{$- \operatorname{Re}[\lambda_{G(0)}] \ (\theta=0) \ [\text{1/m}^2] \text{ vs.\ } r \ [\text{m}]$} 
        \label{fig:sub1 lambda0}
    \end{subfigure}
    \hfill 
    \begin{subfigure}[b]{0.48\textwidth} 
        \centering
        \includegraphics[width=\linewidth]{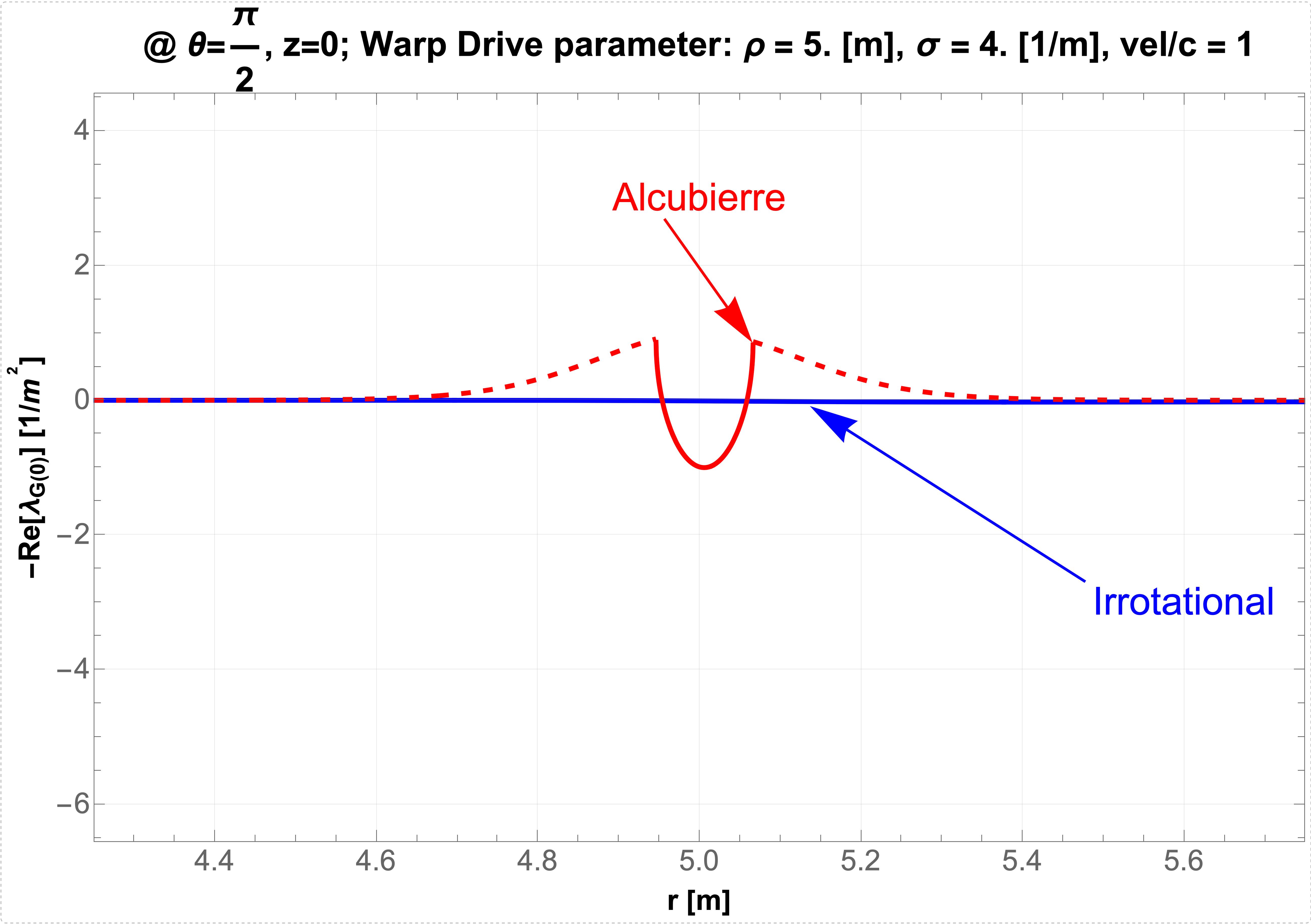}
        \caption{$- \operatorname{Re}[\lambda_{G(0)}] \ (\theta=\pi/2) \ [\text{1/m}^2] \text{ vs.\ } r \ [\text{m}]$} 
        \label{fig:sub2 lambda0}
    \end{subfigure}

    \vspace{0.5cm} 

    \begin{subfigure}[b]{0.48\textwidth} 
        \centering
        \includegraphics[width=\linewidth]{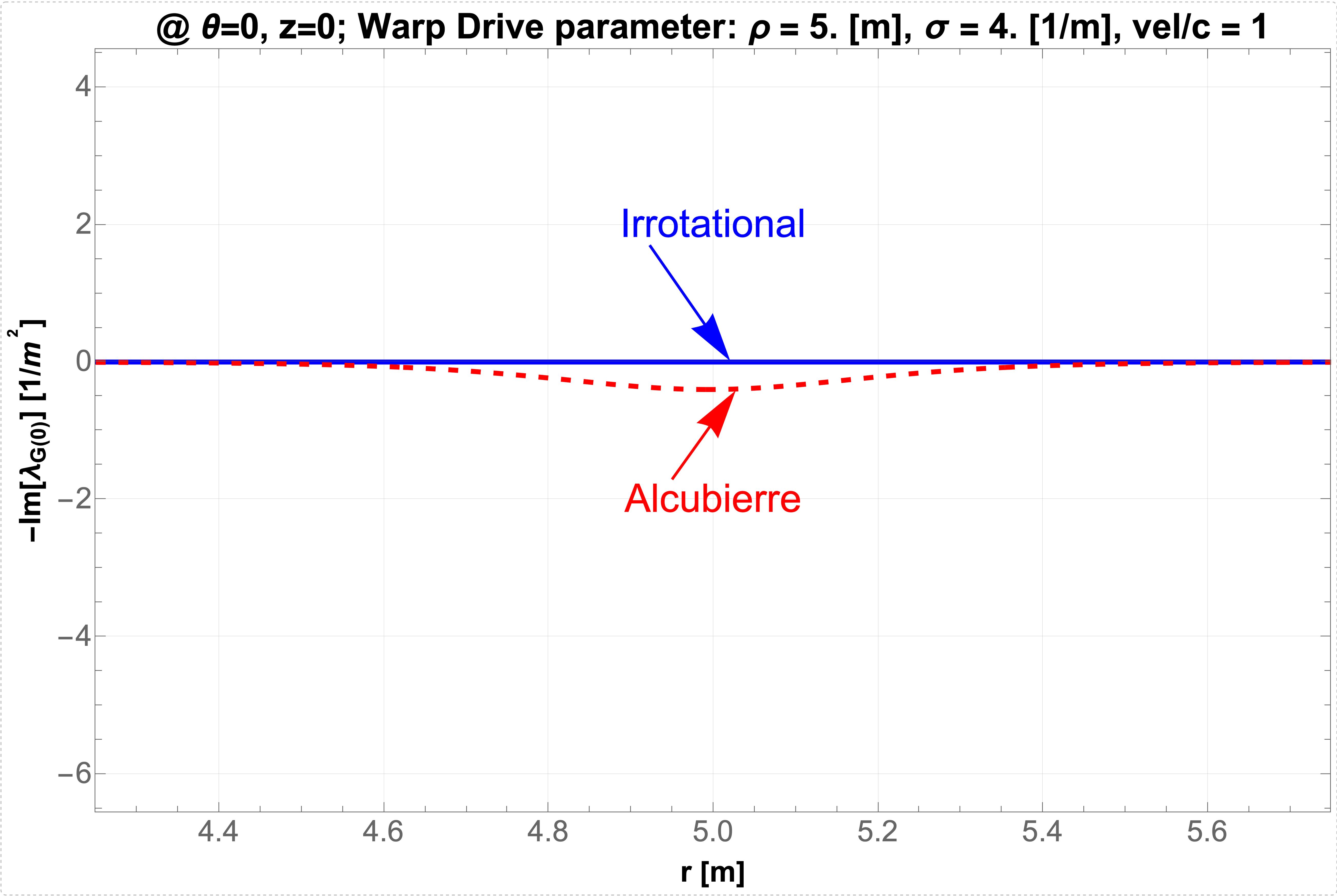} 
        \caption{$- \operatorname{Im}[\lambda_{G(0)}] \ (\theta=0) \ [\text{1/m}^2] \text{ vs.\ } r \ [\text{m}]$} 
        \label{fig:sub3 lambda0}
    \end{subfigure}
    \hfill 
    \begin{subfigure}[b]{0.48\textwidth} 
        \centering
        \includegraphics[width=\linewidth]{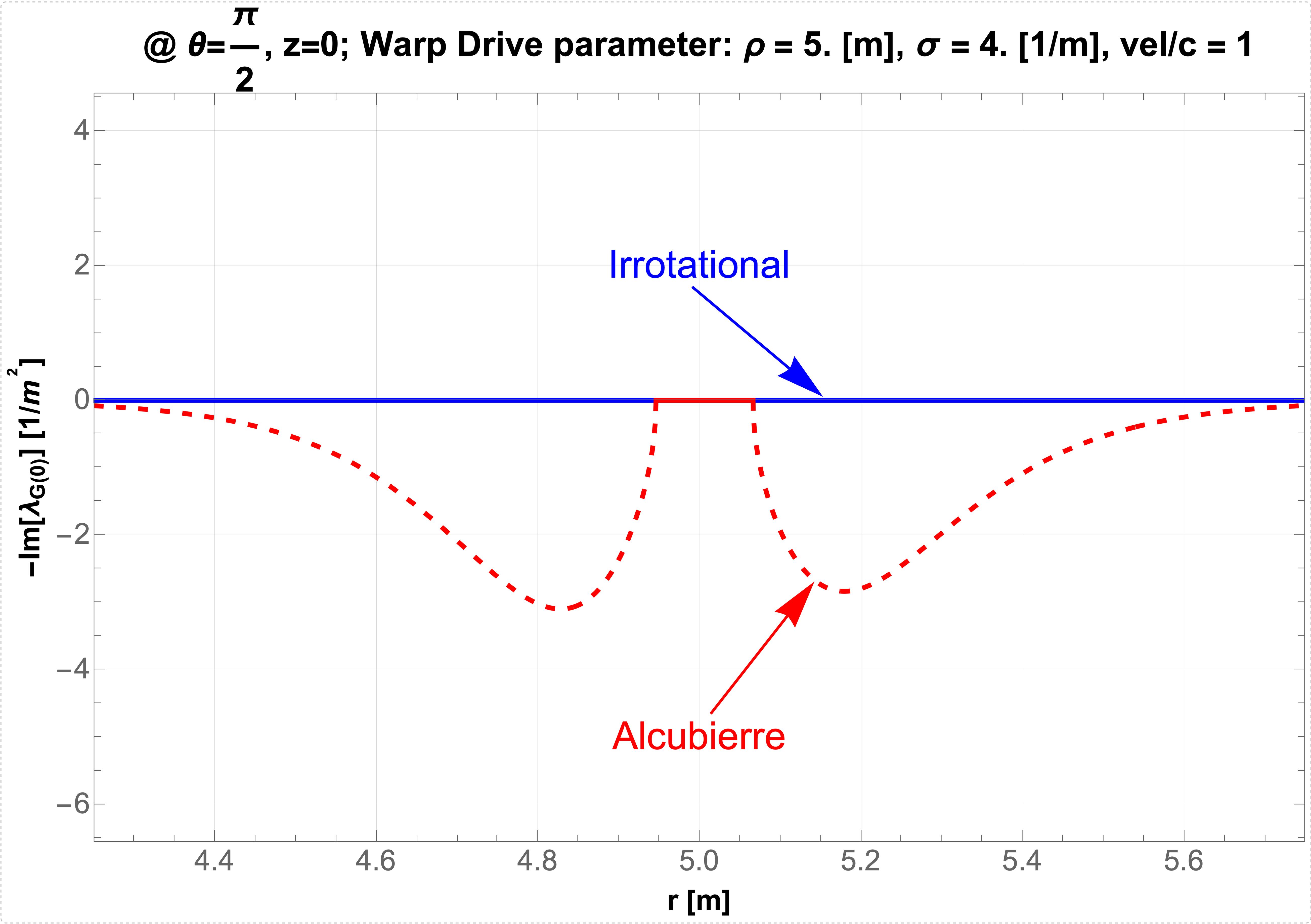}
        \caption{$- \operatorname{Im}[\lambda_{G(0)}] \ (\theta=\pi/2) \ [\text{1/m}^2] \text{ vs.\ } r \ [\text{m}]$} 
        \label{fig:sub4 lambda0}
    \end{subfigure}

\caption{Real and imaginary parts of the timelike-dominant eigenvalue $-\lambda_{G(0)}$ ($\theta=0, \pi/2$) \text{ vs.\ radial distance } $r$ for Alcubierre and irrotational warp drives. Line style indicates local Hawking–Ellis classification: \textbf{solid} for Type I, \textbf{dashed} for Type IV (null complex pair).}
\label{fig:main_figure lambda0}

\end{figure}

\begin{figure}[htbp] 
    \centering 

    \begin{subfigure}[b]{0.48\textwidth} 
        \centering
        \includegraphics[width=\linewidth]{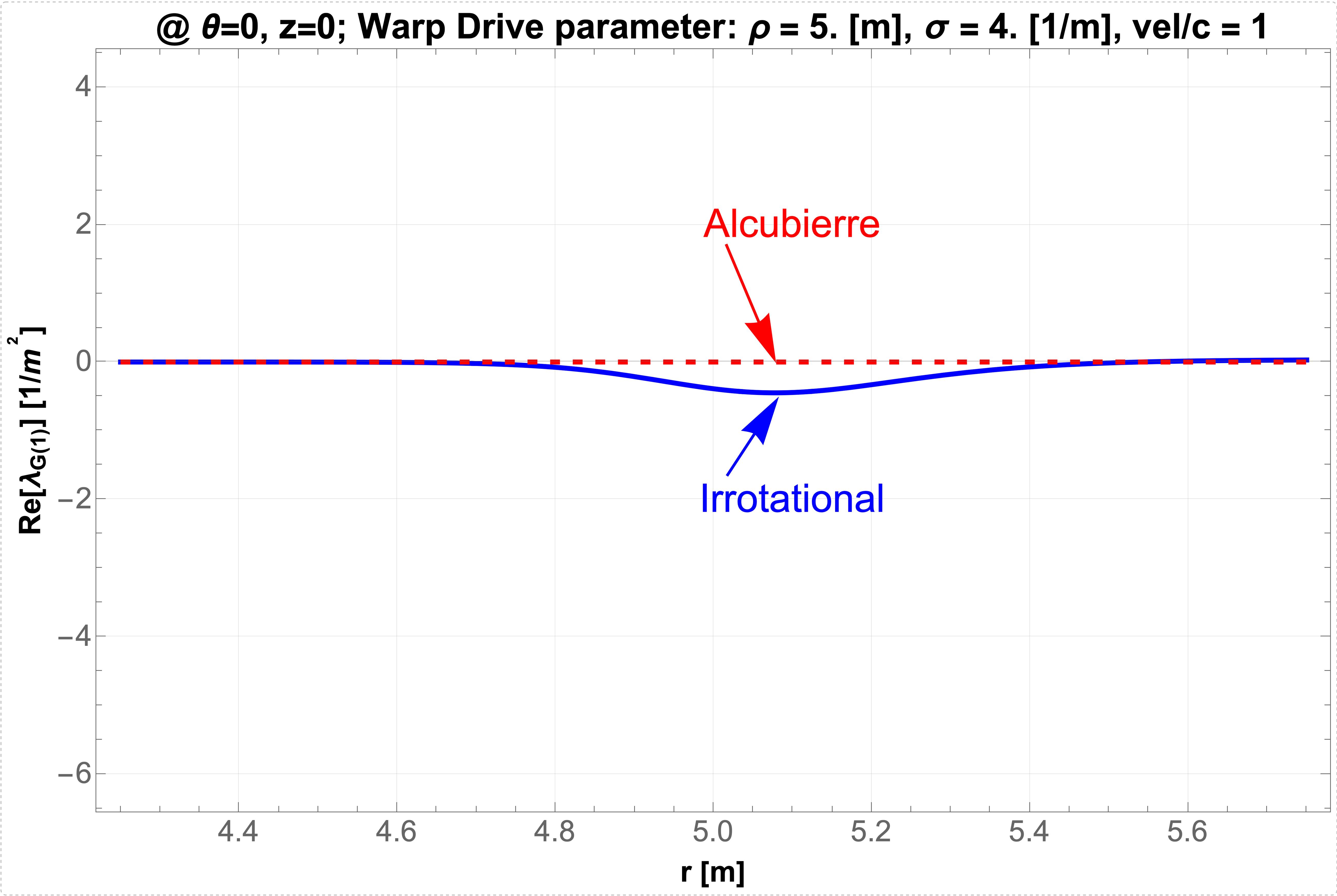}
        \caption{$\operatorname{Re}[\lambda_{G(1)}] \ (\theta=0) \ [\text{1/m}^2] \text{ vs.\ } r \ [\text{m}]$} 
        \label{fig:sub1 lambda10} 
    \end{subfigure}
    \hfill 
    \begin{subfigure}[b]{0.48\textwidth} 
        \centering
        \includegraphics[width=\linewidth]{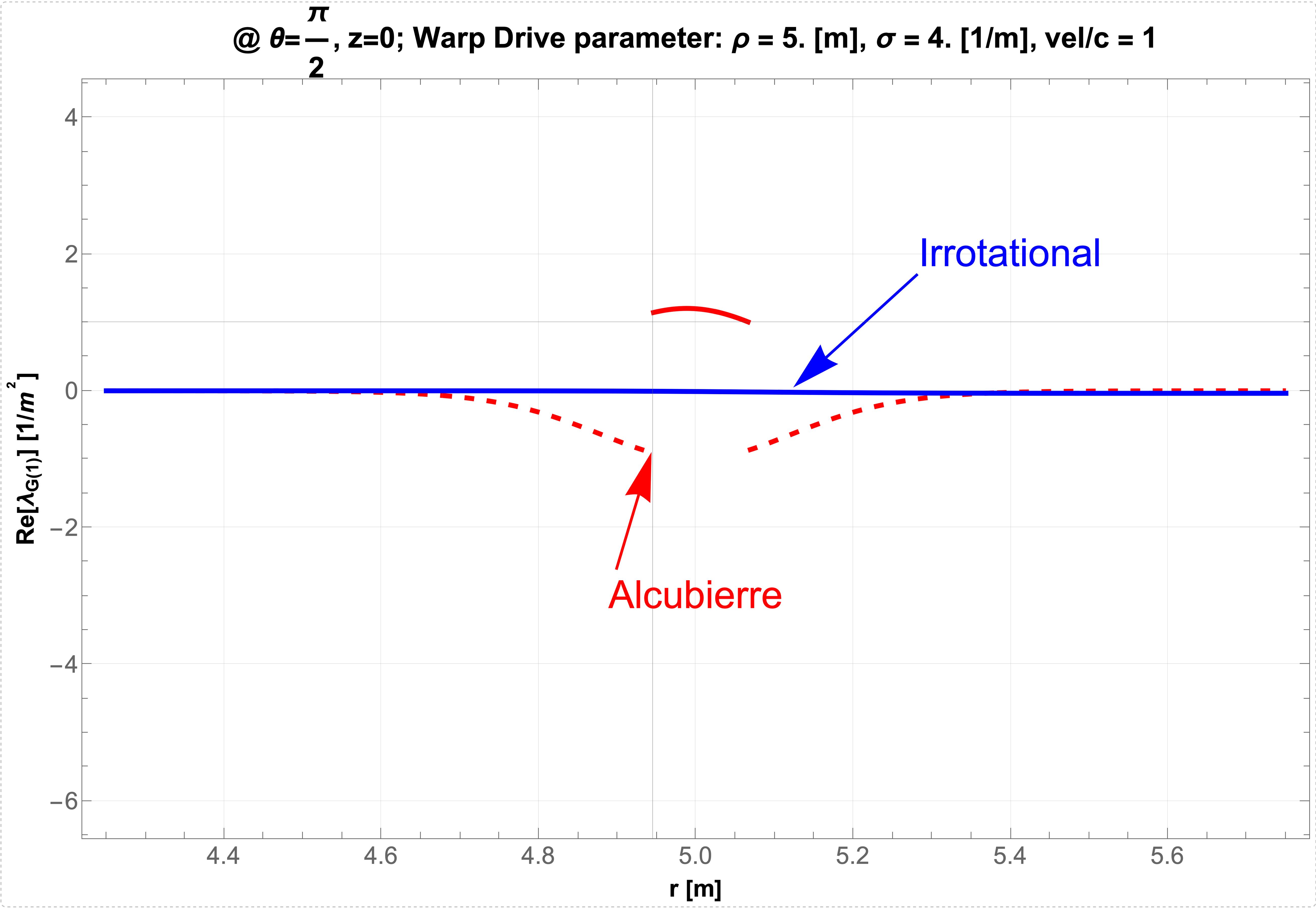}
        \caption{$\operatorname{Re}[\lambda_{G(1)}] \ (\theta=\pi/2) \ [\text{1/m}^2] \text{ vs.\ } r \ [\text{m}]$} 
        \label{fig:sub2 lambda2Pi2}
    \end{subfigure}

    \vspace{0.5cm} 

    \begin{subfigure}[b]{0.48\textwidth} 
        \centering
        \includegraphics[width=\linewidth]{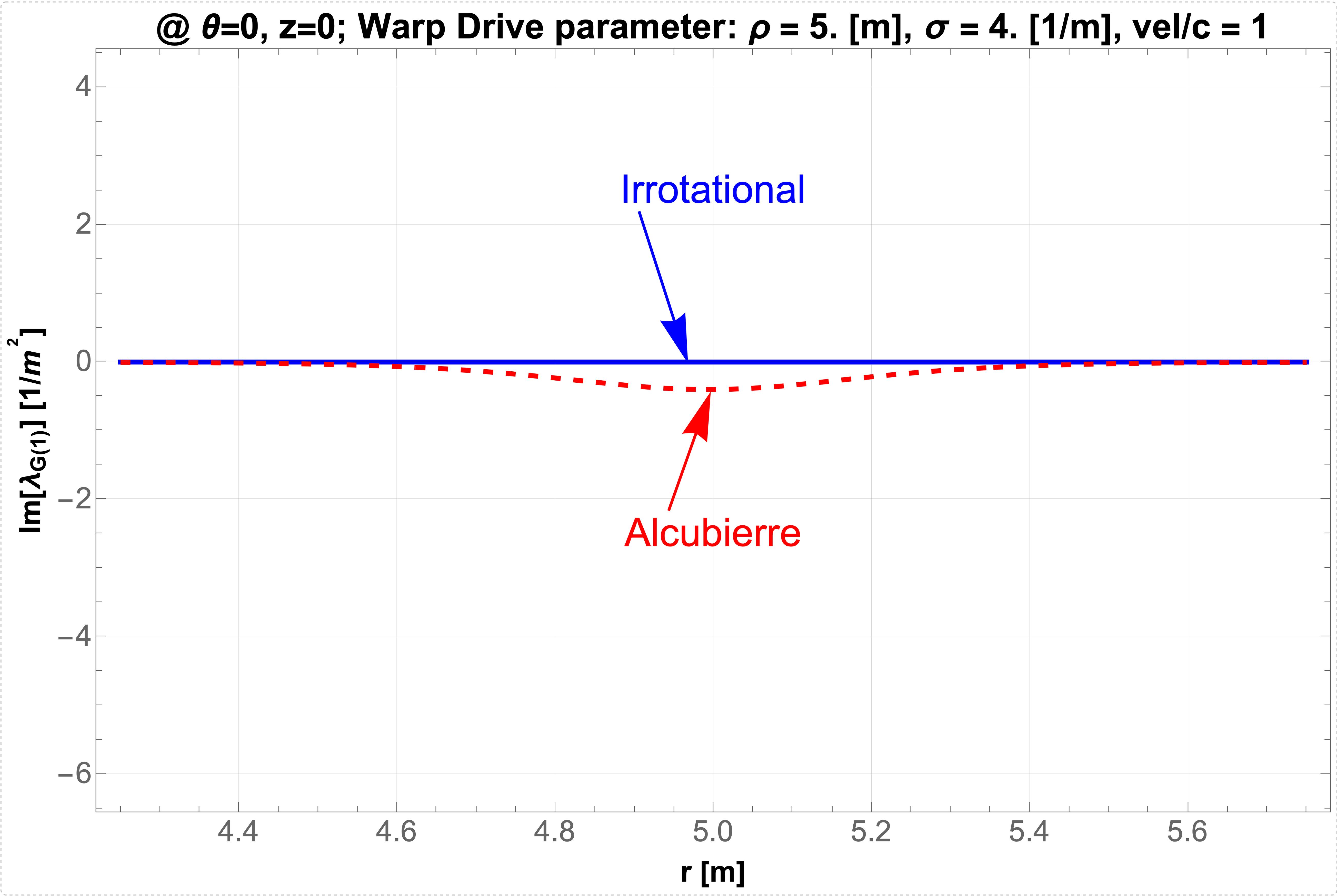} 
        \caption{$\operatorname{Im}[\lambda_{G(1)}] \ (\theta=0) \ [\text{1/m}^2] \text{ vs.\ } r \ [\text{m}]$} 
        \label{fig:sub3 lambda10} 
    \end{subfigure}
    \hfill 
    \begin{subfigure}[b]{0.48\textwidth} 
        \centering
        \includegraphics[width=\linewidth]{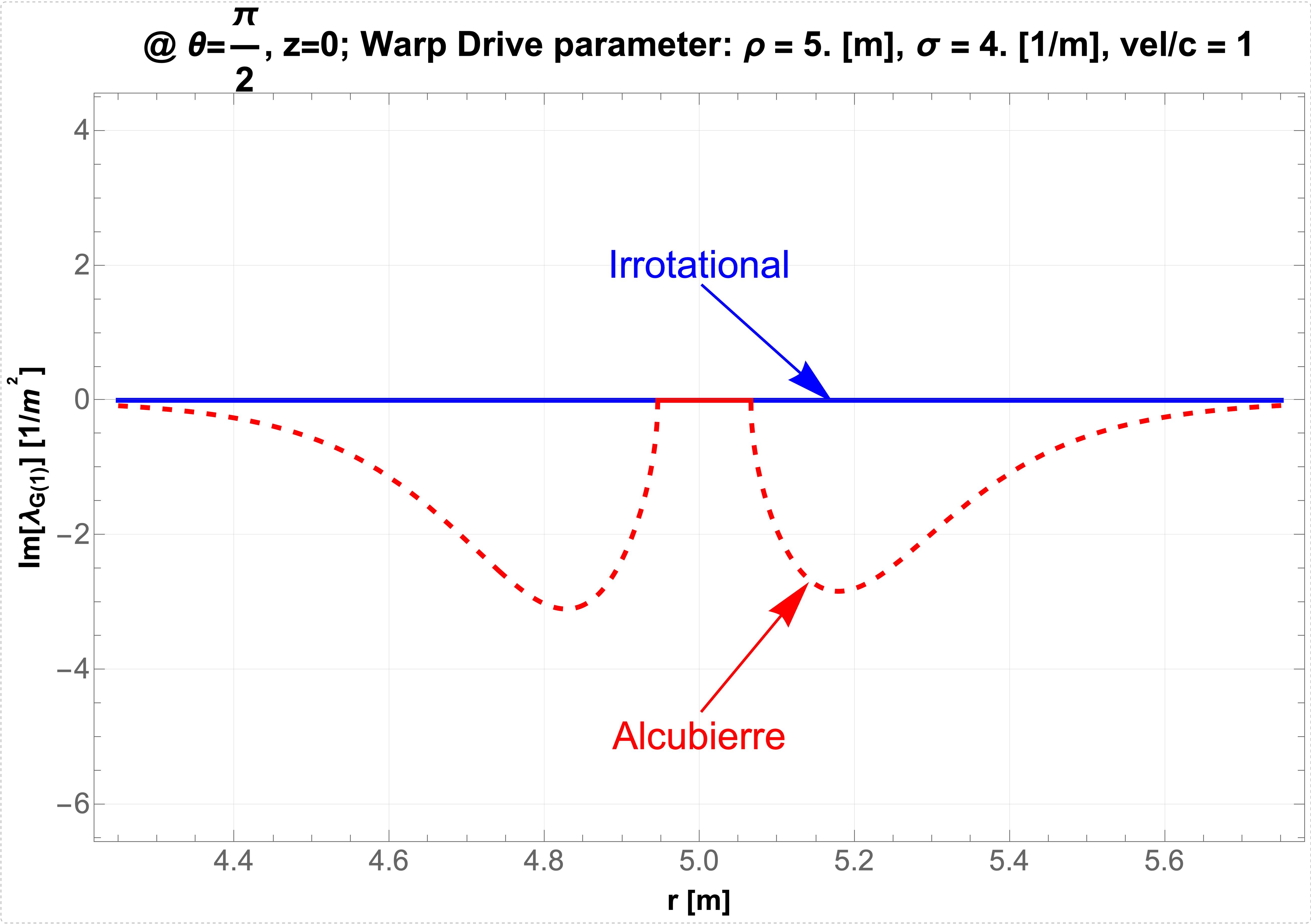} 
        \caption{$\operatorname{Im}[\lambda_{G(1)}] \ (\theta=\pi/2) \ [\text{1/m}^2] \text{ vs.\ } r \ [\text{m}]$} 
        \label{fig:sub4 lambda2Pi2}
    \end{subfigure}

\caption{Real and imaginary parts of the spacelike-dominant eigenvalues in the travel direction, $\lambda_{G(1)} \ (\theta=0,\pi/2)$, \text{ vs.\ radial distance } $r$ for Alcubierre and irrotational warp drives. Line style indicates local Hawking–Ellis classification: \textbf{solid} for Type I, \textbf{dashed} for Type IV (null complex pair).}
\label{fig:main_figure complex spacelike}

\end{figure}

\begin{figure}[htbp] 
    \centering 

    \begin{subfigure}[b]{0.48\textwidth} 
        \centering
        \includegraphics[width=\linewidth]{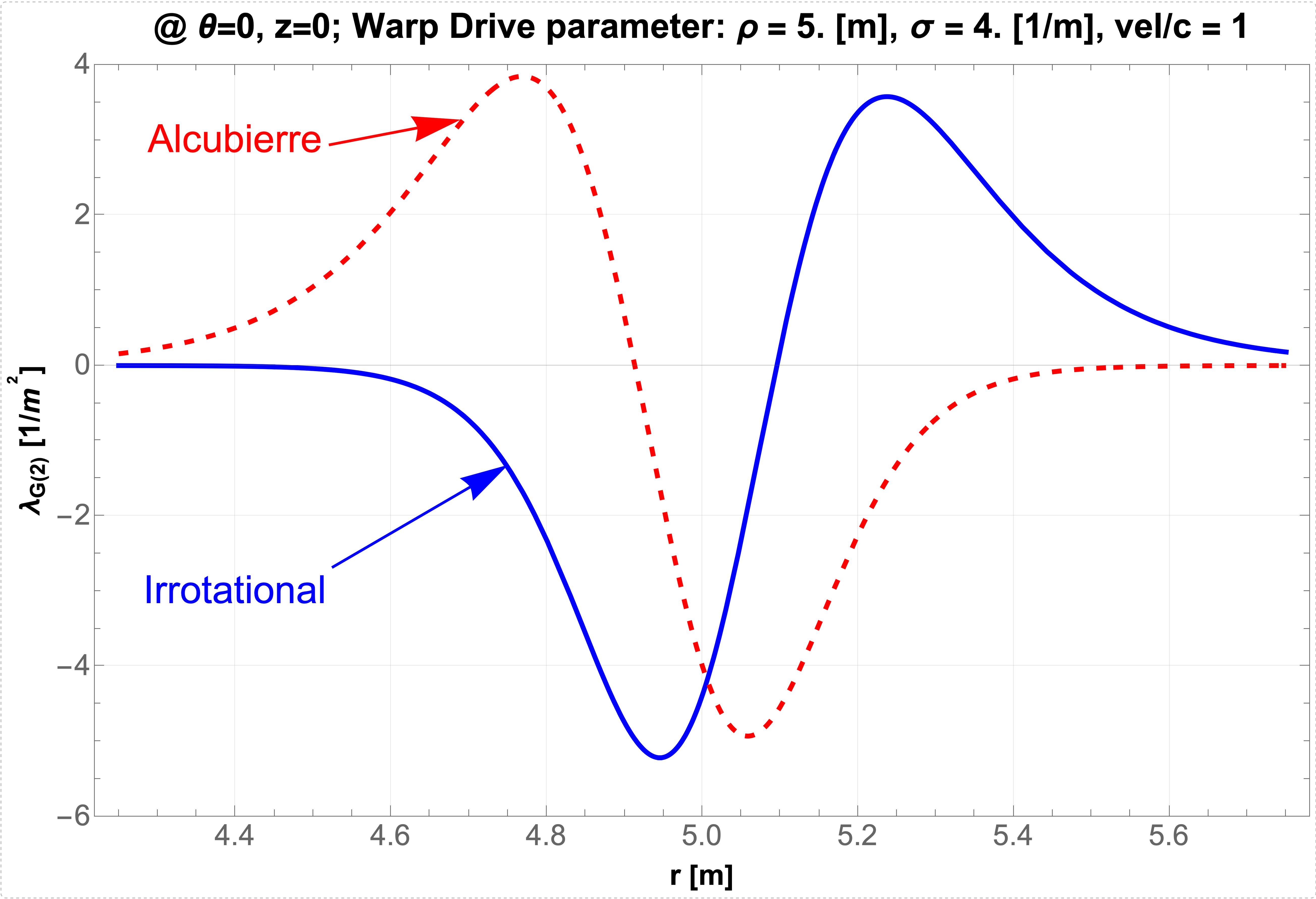} 
        \caption{$\lambda_{G(2)} \ (\theta=0) \ [\text{1/m}^2] \text{ vs.\ } r \ [\text{m}]$} 
        \label{fig:sub31 lambda20} 
    \end{subfigure}
    \hfill 
    \begin{subfigure}[b]{0.48\textwidth} 
        \centering
        \includegraphics[width=\linewidth]{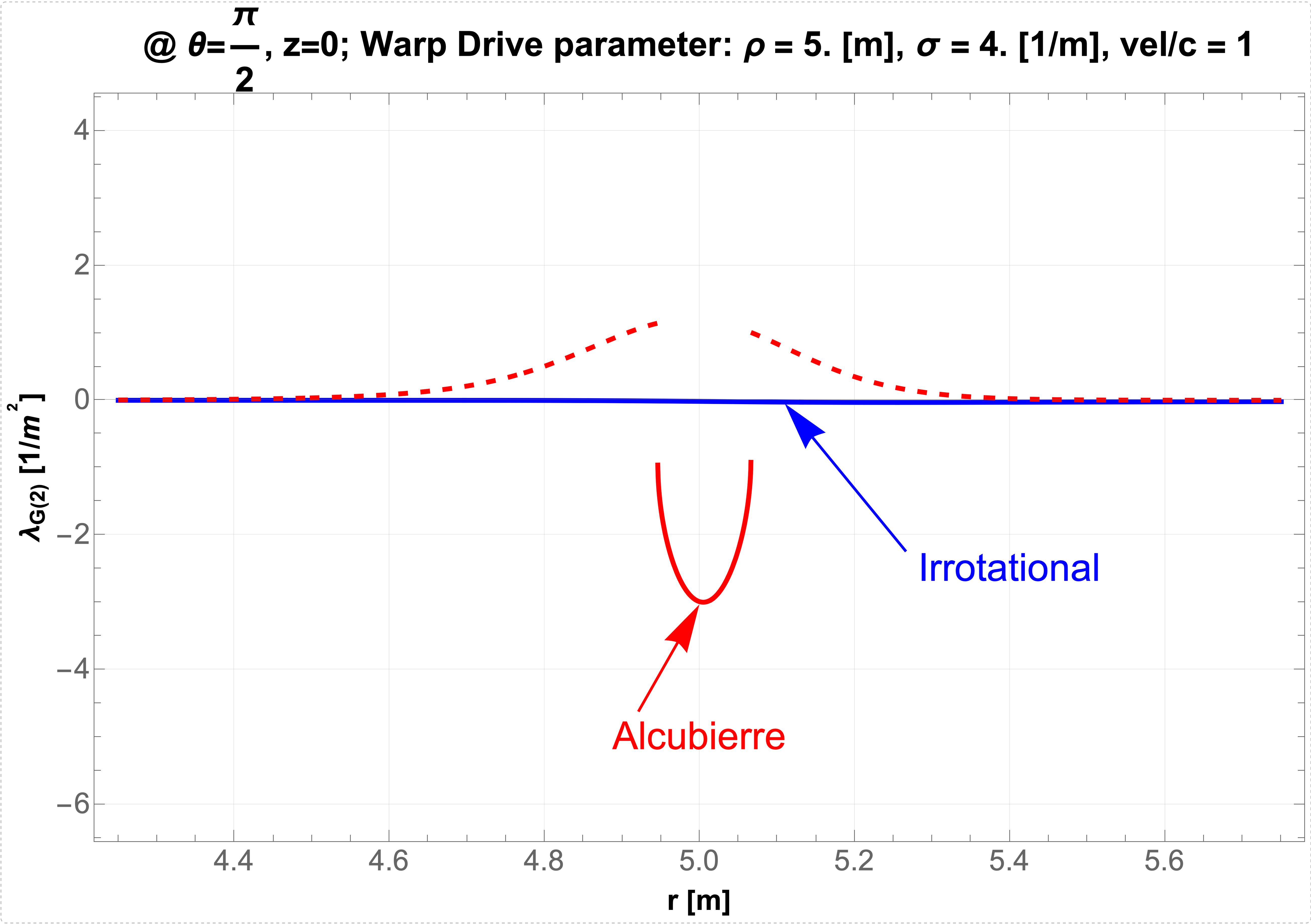}
        \caption{$\lambda_{G(2)} \ (\theta=\pi/2) \ [\text{1/m}^2] \text{ vs.\ } r \ [\text{m}]$} 
        \label{fig:sub32 lambda1pi2} 

    \end{subfigure}

    \vspace{0.5cm} 

    \begin{subfigure}[b]{0.48\textwidth} 
        \centering
        \includegraphics[width=\linewidth]{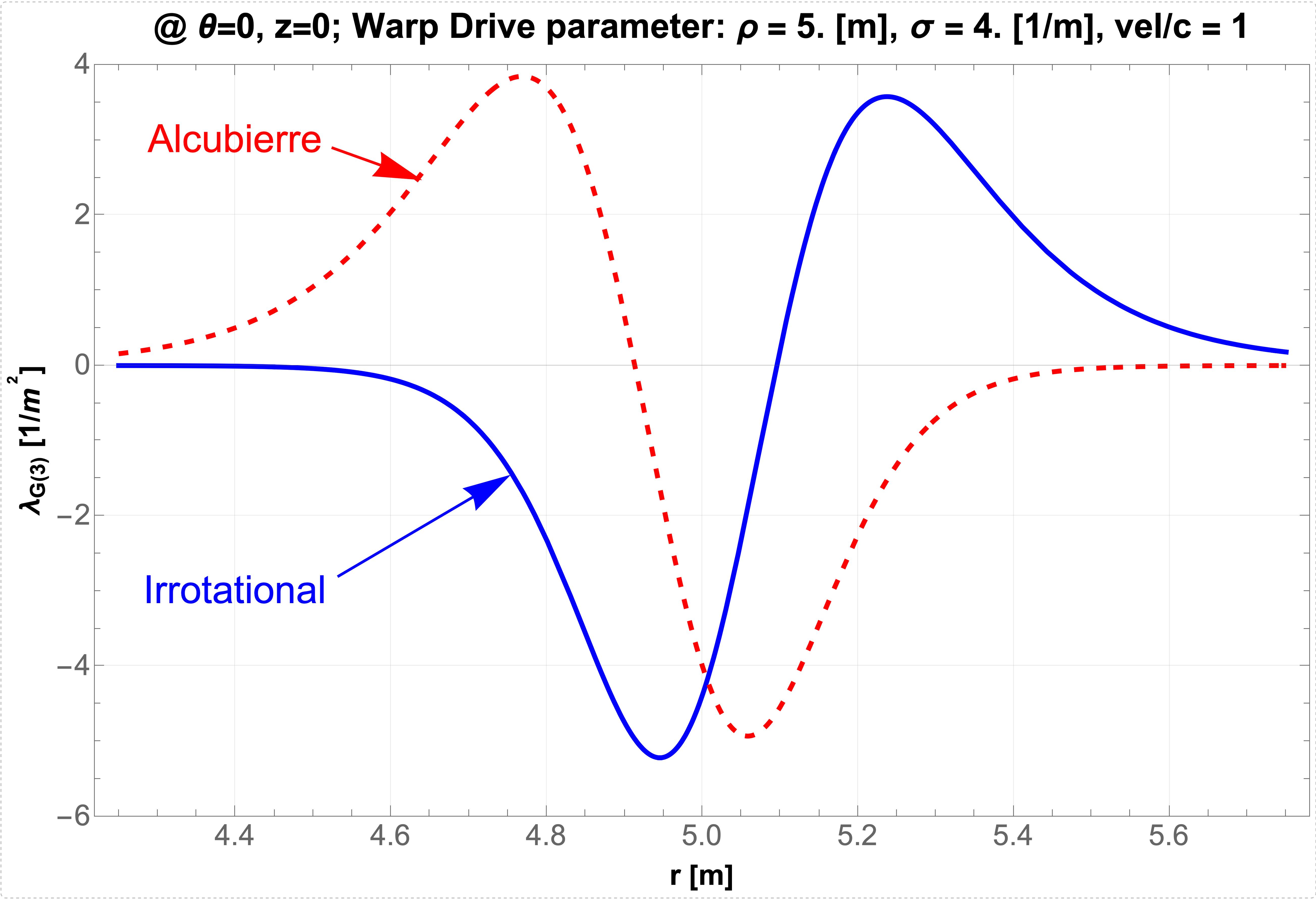}
        \caption{$\lambda_{G(3)} \ (\theta=0) \ [\text{1/m}^2] \text{ vs.\ } r \ [\text{m}]$} 
        \label{fig:sub33 lambda30} 

    \end{subfigure}
    \hfill 
    \begin{subfigure}[b]{0.48\textwidth} 
        \centering
        \includegraphics[width=\linewidth]{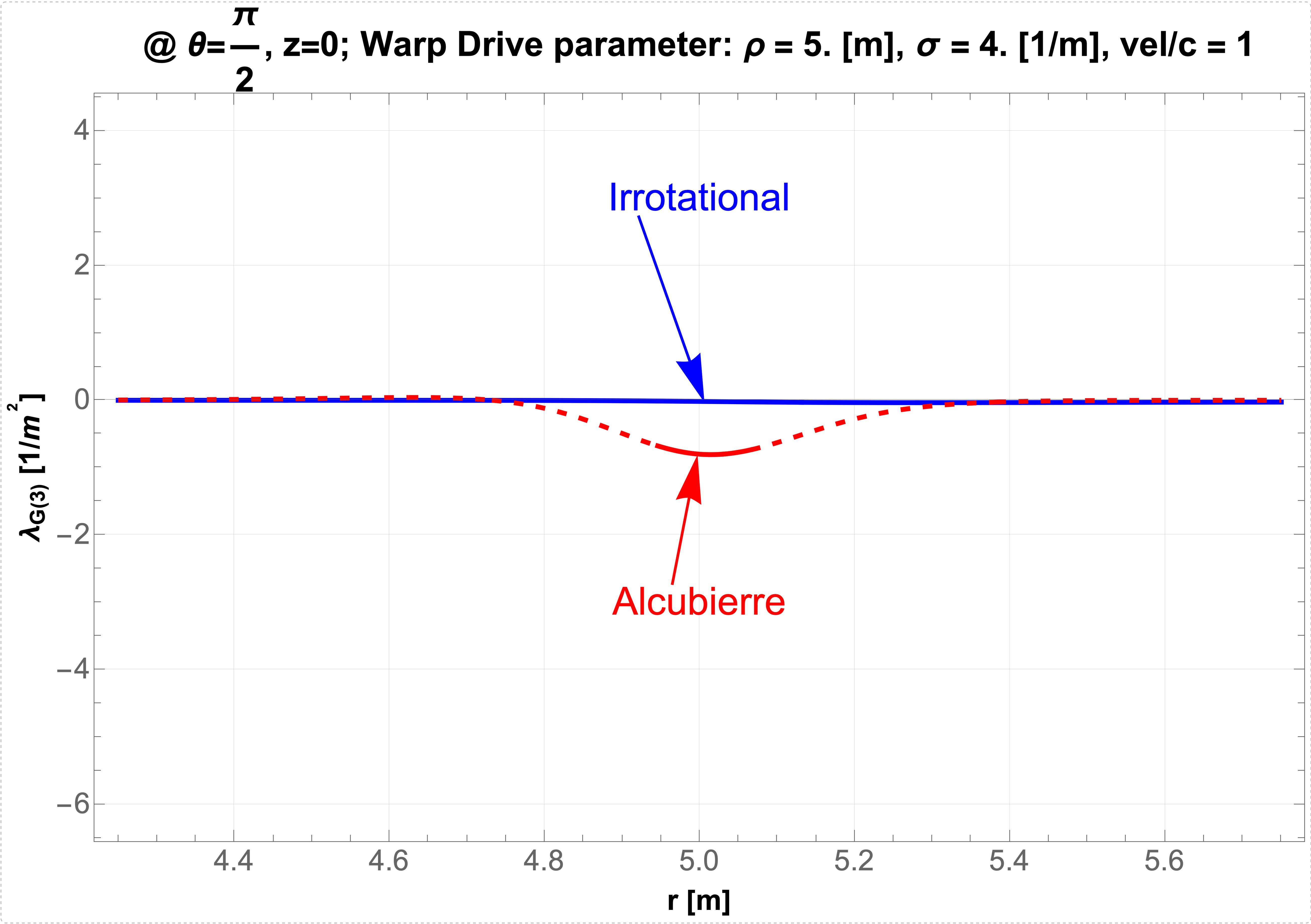}
        \caption{$\lambda_{G(3)} \ (\theta=\pi/2) \ [\text{1/m}^2] \text{ vs.\ } r \ [\text{m}]$} 
        \label{fig:sub34 lambda3pi2} 

    \end{subfigure}

\caption{Spacelike-dominant purely real eigenvalues perpendicular to the travel direction: $\lambda_{G(2)} \ (\theta=0, \pi/2)$ and $\lambda_{G(3)} \ (\theta=0, \pi/2$) \text{ vs.\ radial distance } $r$ for Alcubierre and irrotational warp drives. Line style indicates local Hawking–Ellis classification: \textbf{solid} for Type I, \textbf{dashed} for Type IV (null complex pair).}
\label{fig:main_figure real spacelike}

\end{figure}

\begin{figure}[htbp] 
    \centering 

    \begin{subfigure}[b]{0.48\textwidth} 
        \centering
        \includegraphics[width=\linewidth]{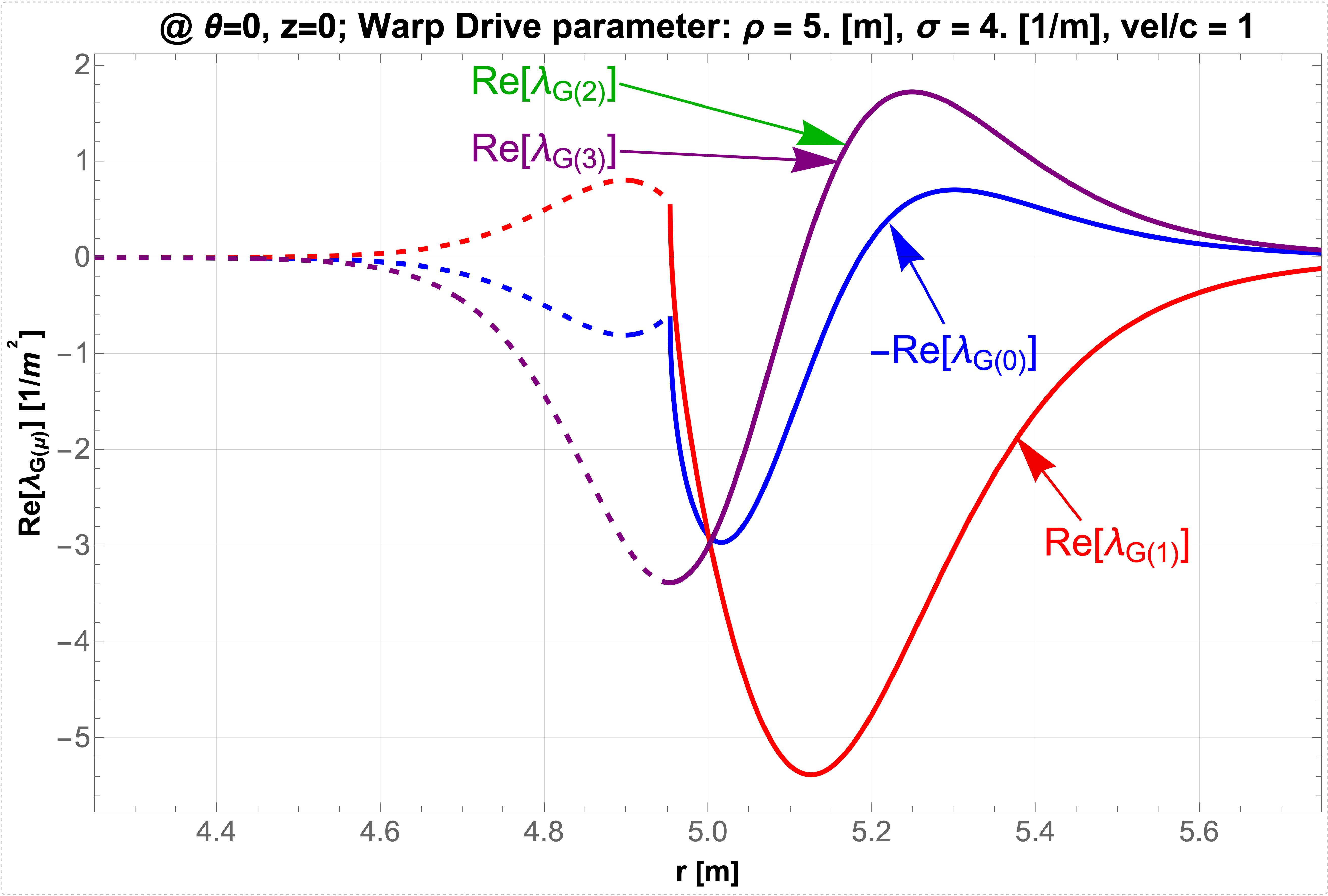}
        \caption{Real parts $\mathrm{Re}(\lambda_{G(\mu)})$ vs. $r$ at $\theta=0$.}
        \label{fig:natario_re_theta0} 
    \end{subfigure}
    \hfill 
    \begin{subfigure}[b]{0.48\textwidth} 
        \centering
        \includegraphics[width=\linewidth]{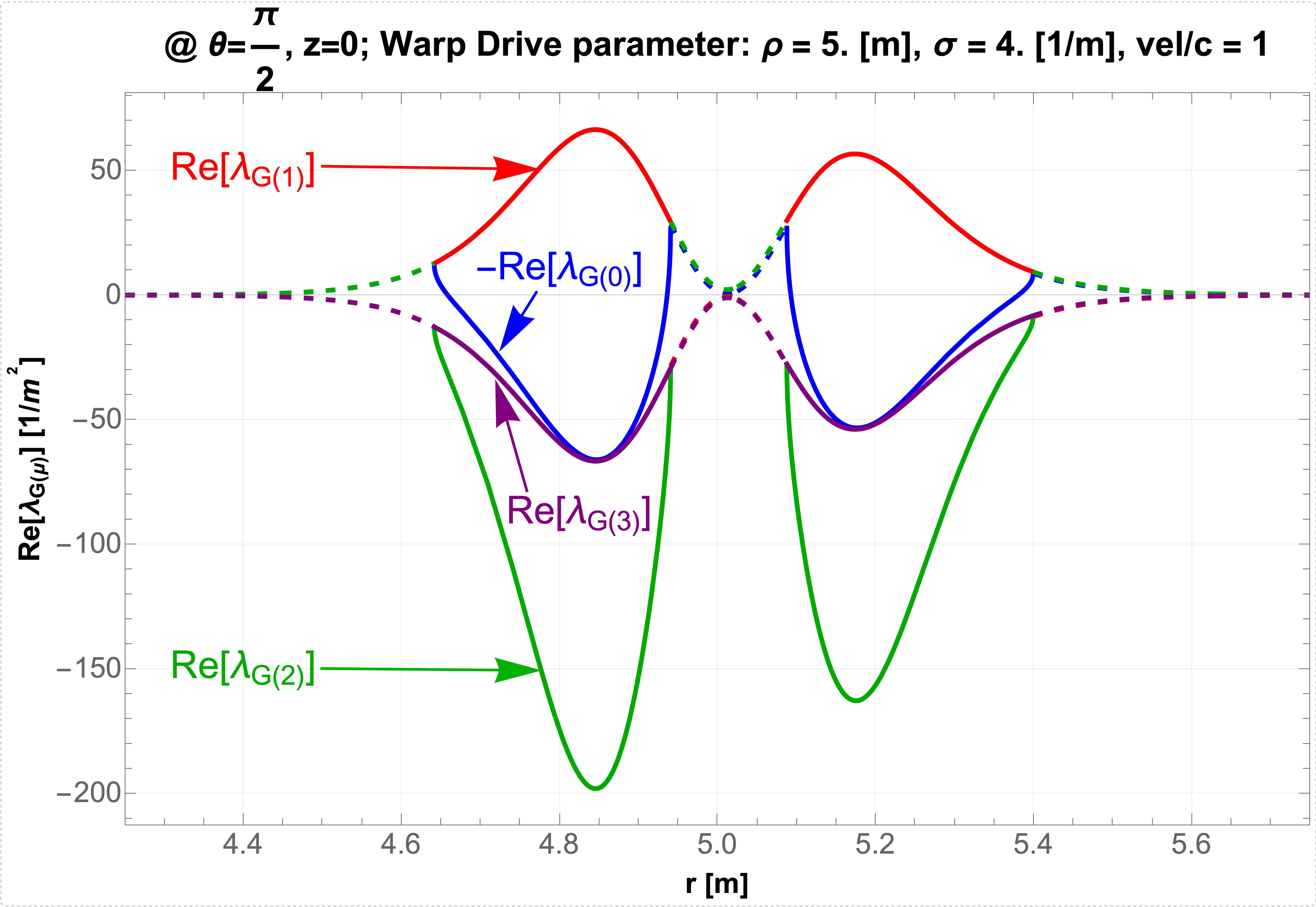}
        \caption{Real parts $\mathrm{Re}(\lambda_{G(\mu)})$ vs. $r$ at $\theta=\pi/2$.}
        \label{fig:natario_re_thetaPi2} 
    \end{subfigure}

    \vspace{0.5cm} 

    \begin{subfigure}[b]{0.48\textwidth} 
        \centering
        \includegraphics[width=\linewidth]{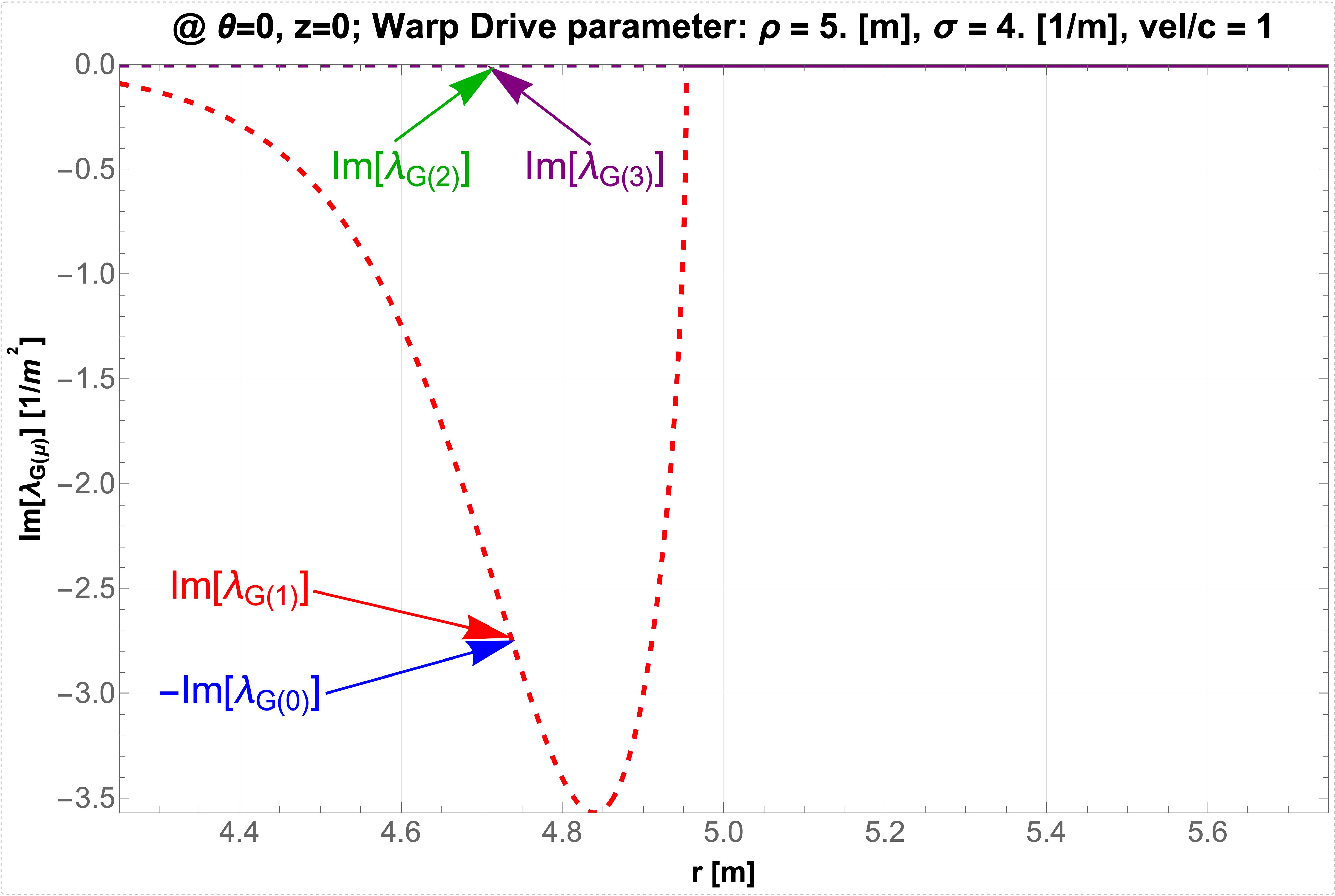}
        \caption{Imaginary parts $\mathrm{Im}(\lambda_{G(\mu)})$ vs. $r$ at $\theta=0$.}
        \label{fig:natario_im_theta0} 
    \end{subfigure}
    \hfill 
    \begin{subfigure}[b]{0.48\textwidth} 
        \centering
        \includegraphics[width=\linewidth]{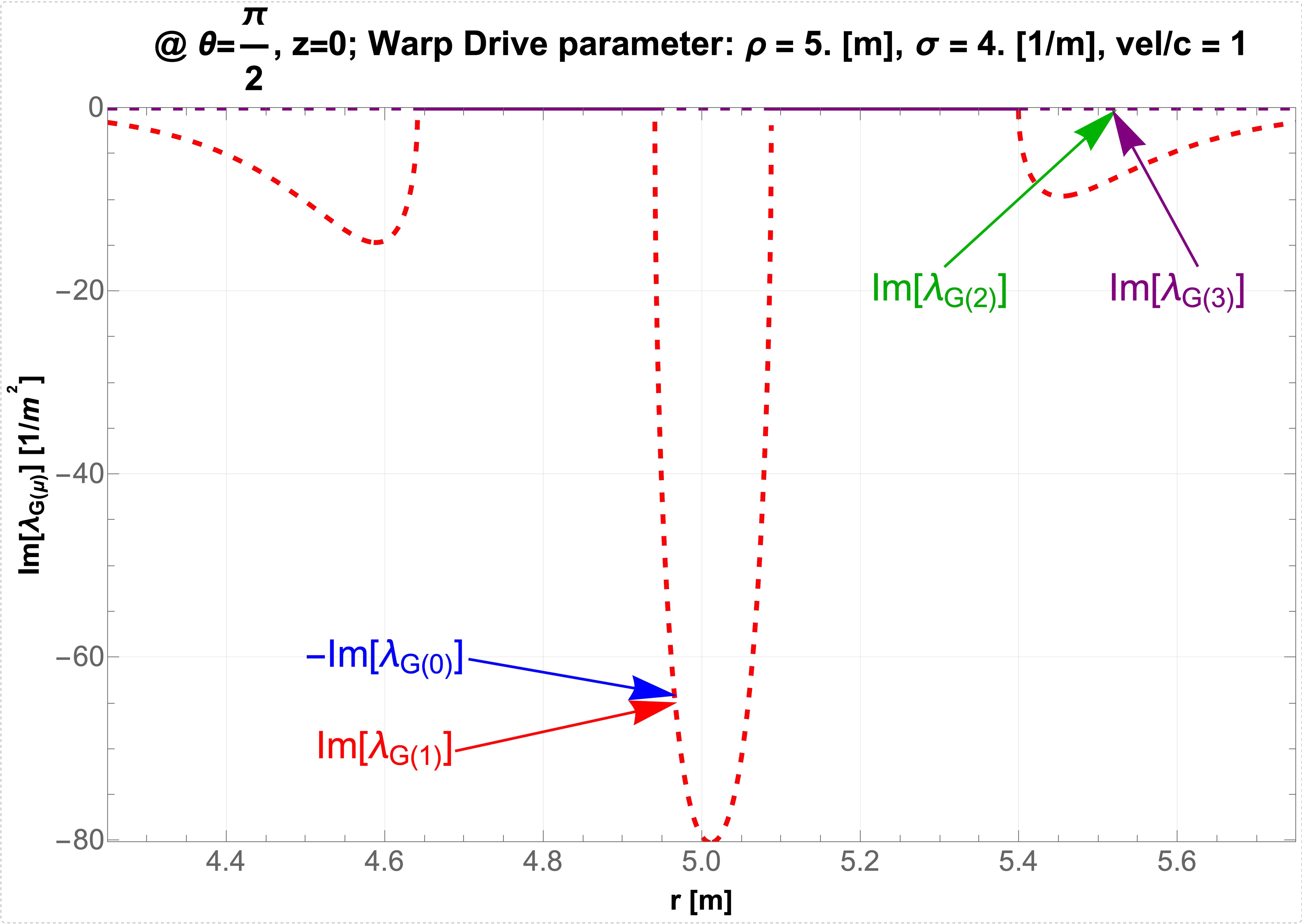}
        \caption{Imaginary parts $\mathrm{Im}(\lambda_{G(\mu)})$ vs. $r$ at $\theta=\pi/2$.}
        \label{fig:natario_im_thetaPi2} 
    \end{subfigure}

    \caption{Physically sorted eigenvalues $\{\lambda_{G(\mu)}\}_{\mu=0}^{3}$ of $G^{\hat{\alpha}}{}_{\hat{\beta}}$ versus radial distance $r$ for Natário's warp drive. Panels display real parts (top row) and imaginary parts (bottom row) of the eigenvalues along the axis of motion ($\theta=0$, left column) and in the equatorial plane  ($\theta=\pi/2$, right column). Line style indicates local Hawking–Ellis classification: \textbf{solid} for Type I (all eigenvalues real), \textbf{dashed} for Type IV (complex eigenvalues present).}
    \label{fig:Natario_eigenvalues_radial} 
\end{figure}

\subsection{Interpretation of eigenvalues and energy conditions from visualizations}
\label{sec:eigenvalue_interpretation}

Three‑dimensional \textit{Mathematica\textsuperscript{\textregistered}}~\cite{Mathematica} \texttt{RevolutionPlot3D} surfaces (Figs.~\ref{fig:main_3x3_grid A}–\ref{fig:main_3x3_grid C}) display the meridional distribution of selected combinations of the \emph{physically sorted eigenvalues} $\{\lambda_{G(\alpha)}\}_{\alpha=0}^{3}$ that enter the weak (WEC) and null (NEC) energy‑condition diagnostics, together with the scalar \emph{invariant} $G=-R=\sum_{\alpha=0}^{3}\lambda_{G(\alpha)}$.  
Each surface sits in the prime meridional $(x,y)$‑plane (the $x$‑axis points along the direction of motion; cf.\ Fig.~\ref{Fig1_SpherCoord}); with its vertical coordinate $z$ the amplitude of the plotted quantity, encoded by a rainbow color scale keyed to the $z$-values.

Alcubierre and irrotational results are plotted on a common vertical scale. Natário requires its own scale because the zero-expansion shift field forces amplitudes larger (by an order of magnitude), which underscores how suppressing local spatial volume expansion forces much larger curvature—and hence much larger stress–energy densities—to sustain the warp bubble.

To visualize the global Hawking–Ellis type distribution of the Alcubierre and Natário bubbles, Fig.~\ref{fig:main_2x3_grid} —also generated with \texttt{RevolutionPlot3D}— plots the timelike‑associated eigenvalue $-\Re[\lambda_{G(0)}]$ with color encoding algebraic type (light for Type I, dark for Type IV). Regions are classified as Type I \((\Im[\lambda_{G(0)}]=0\), hence \(-\Re[\lambda_{G(0)}]=\kappa\,\varrho_p\)) versus Type IV (\(\Im[\lambda_{G(0)}]\equiv\mu\neq0\)). Columns display (left) Type I only, (center) Type IV only, and (right) their overlay; rows distinguish (top) Alcubierre from (bottom) Natário.  

Complementary one‑dimensional radial cuts—generated with \texttt{Plot}—along the axis of motion ($\theta=0$) and in the equatorial plane ($\theta=\pi/2$) (Figs.~\ref{fig:main_figure lambda0}–\ref{fig:Natario_eigenvalues_radial}) plot each physically-sorted eigenvalue  $\{\lambda_{G(\alpha)}\}_{\alpha=0}^{3}$ across the warp-bubble walls. Line style indicates Hawking–Ellis classification (solid for Type I, dashed for Type IV), for direct comparison with Fig.~\ref{fig:main_2x3_grid}.

The 3-D type-structure surfaces and 1-D cuts collectively reveal the persistent Type IV regions in the Alcubierre and Natário drives versus the irrotational solution’s global Type I character.   As shown in \cref{prop:irrot_class}, the irrotational warp drive therefore admits a local rest frame, a unique timelike eigenvector, and an \emph{invariant} proper energy density.

The first two eigenvalues in the physically-sorted list, \(\lambda_{G(0)}\) and \(\lambda_{G(1)}\), correspond, in the case of Type IV regions, to the complex conjugate pair \(\Lambda_G = \xi \pm i\mu\).
This implies that their Type IV regions, as discussed within \cref{sec:IVgeometry_and_physical_meaning}, possess no real causal eigenvectors, admit no preferred local rest frame, and lack a uniquely defined invariant proper energy density (cf.~Table~\ref{tab:type_comparison}). These characteristics are consequences of the non-zero parameter \(\mu\), which corresponds to an irreducible energy flux component \(q_s = \mu/\kappa\) within the invariant 2-plane associated with the complex eigenvalues. 
The null energy condition (NEC), expressed in terms of the mixed Einstein tensor tetrad components as \(G^{\hat{\sigma}}{}_{\hat{\beta}} \eta_{\hat{\alpha}\hat{\sigma}} k^{\hat{\alpha}} k^{\hat{\beta}} \ge 0\) for all null vectors \(k^{\hat{\alpha}}\) (satisfying \( \eta_{\hat{\gamma}\hat{\delta}} k^{\hat{\gamma}} k^{\hat{\delta}} = 0 \)), is inherently violated in regions where \(\mu \neq 0\). 
This violation occurs because the Type IV algebraic structure guarantees the condition fails for some null directions \(k^{\hat{\alpha}}\). Since the weak energy condition (WEC) implies the NEC, the WEC is also necessarily violated in these Type IV regions.

\subsubsection{Timelike eigenvalue and energy density interpretation}

We now detail the interpretation of the eigenvalue \(-\lambda_{G(0)}\) visualizations (Figs.~\ref{fig:main_3x3_grid A}–\ref{fig:Natario_eigenvalues_radial}), starting with the irrotational drive. 

\paragraph*{Irrotational warp drive} The irrotational warp drive's global Type I classification ensures its eigenvalues are real (e.g., $\operatorname{Im}[\lambda_{G(0)}]=0$, Fig.~\ref{fig:sub_r2_c2 A}). The quantity \(-\lambda_{G(0)}\), related to the invariant proper energy density by \(-\lambda_{G(0)} = \kappa \varrho_p\) and depicted in Fig.~\ref{fig:sub_r1_c2 A}, must accordingly be non-negative (\(-\lambda_{G(0)} \ge 0\)) for the weak energy condition (WEC, requiring \(\varrho_p \ge 0\)) to be satisfied. Calculations show this quantity \(-\lambda_{G(0)}\) is predominantly positive (reaching a maximum of \(0.46\,[\text{m}^{-2}]\) on the axis of motion, $\theta=0$), but becomes slightly negative (reaching a minimum of \(-0.026\,[\text{m}^{-2}]\)) in a localized region near \(\theta=\pi/2\) (perpendicular to the direction of travel). This indicates a small, localized WEC violation (\(\varrho_p < 0\)) at \(\theta=\pi/2\), undiscernible at the scale shown in the figure.

\paragraph*{Alcubierre warp drive}\label{subpar:Alcubierre}

\textit{Global algebraic features.}  
Both the Alcubierre spacetime (Figs.~\ref{fig:HEtyp_r1_c2},
\ref{fig:HEtyp_r1_c3}) and the Natário spacetime
(Figs.~\ref{fig:HEtyp_r2_c2}, \ref{fig:HEtyp_r2_c3}) contain extended
Type IV regions, where the mixed Einstein tensor
\(G^{\hat\mu}{}_{\hat\nu}\) possesses complex eigenvalues
\(\Lambda_{G}=\xi\pm i\mu\) with \(\mu\neq0\)
(Figs.~\ref{fig:sub_r2_c1 A}, \ref{fig:sub_r2_c3 A}).  This algebraic structure \emph{necessarily violates} the null and weak
energy conditions (see Table~\ref{tab:type_comparison}): no real
timelike eigenvector exists, the observer‑dependent energy density in
the invariant two‑plane is
\(\varrho_{\text{obs}}=-\xi/\kappa<0\), and the irreducible energy flux
is \(q_{s}=\mu/\kappa\).

\smallskip
\textit{Axis of motion cut (\(\theta=0\)).}\quad
Along the axis of motion the tensor is Type IV for all radii (Figs.~\ref{fig:sub1 lambda0}, \ref{fig:sub3 lambda0}).  The
imaginary part dominates, \(|\mu|>|\xi|\), but its absolute value is
\emph{not} maximal on this cut (Fig.~\ref{fig:sub3 lambda0}).

\smallskip
\textit{Equatorial plane cut (\(\theta=\pi/2\)).}\quad
Along this cut the algebraic type varies with radius.  
In the bubble interior ($r\lesssim\rho-\Delta$) and the far field
($r\gg\rho+\Delta$), where $\Delta \approx 0.014 \rho$ the tensor remains Type IV; the energy flux
parameter \( q_{s} = \mu/\kappa\) attains its \emph{largest} magnitude,
\(\mu\simeq -3.0\;[\text{m}^{-2}]\) at $r\approx0.96\rho$ (Fig.~\ref{fig:sub4 lambda0}).
Within an annulus of width \(\sim2\Delta\) about the wall
($r\approx\rho$) we have (\(\mu=0\)) and the tensor
switches to Type I.  A unique rest frame exists, and the
timelike eigenvalue $\kappa\,\varrho_{p}(r)=-\lambda_{G(0)}(r)$ acquires a pronounced \emph{U}‑shape: it peaks at
\(\kappa\,\varrho_{p}\simeq0.9\;[\text{m}^{-2}]\) near
$r\approx\rho\pm\Delta$ and dips to
\(\kappa\,\varrho_{p}\simeq-1.0\;[\text{m}^{-2}]\) at $r\approx\rho$ (Fig.~\ref{fig:sub2 lambda0}), producing
a local WEC violation.

\smallskip

\textit{Comparative energy‐density profiles of Alcubierre and irrotational drives.}\quad
The irrotational solution is \emph{everywhere}\ Type\,I, with its peak positive proper energy density 
\(\varrho_p\simeq0.46\,[\text{m}^{-2}]/\kappa\) on the axis of motion  (\(\theta=0,\;r\approx\rho\)).  
By contrast, the Alcubierre drive’s Type\,I region with peak negative proper density $\varrho_p\simeq-1.0\,[\text{m}^{-2}]/\kappa$ occurs at the equatorial cut (at \(\theta=\pi/2,\;r\approx\rho\)).  Therefore the irrotational peak positive proper energy density is 46\,\% of that Alcubierre magnitude.  
Furthermore, the irrotational peak \(\varrho_p\) is only about \(15\%\) of the Alcubierre drive’s peak irreducible flux magnitude \(\lvert q_s\rvert\simeq3.0\,[\text{m}^{-2}]/\kappa\) 
in its equatorial Type IV region; this flux magnitude itself is roughly 3 times larger than the magnitude of the Alcubierre drive’s peak negative \(\varrho_p\) in its Type I zone. Comparing the proper energy density \(|\varrho_p|\) of the Type I irrotational drive with the irreducible flux magnitude \(|q_s|\) from Alcubierre's Type IV regions provides a useful gauge of relative \emph{exotic matter cost}, despite their distinct physical roles. Both quantities are scalar invariants with the same physical dimensions (\([\text{m}^{-2}]\)). As established in \cref{prop:mu_invariance}, for Type IV regions, \(|\mu|\) (and thus \(|q_s|\)) is invariant, while \(\xi = \operatorname{Re}[\lambda_{G(0)}]\) is an invariant scalar that relates to an observer-dependent energy density \(\varrho_{\mathrm{obs}}=-\xi/\kappa\) for observers within the null two-plane \(\mathcal{N}\); such regions inherently lack a unique rest frame and proper energy density. In contrast, the irrotational drive's Type I structure yields a true frame-invariant proper energy density \(\varrho_p=-\lambda_{G(0)}/\kappa\).

Thus, eliminating vorticity (\(\omega_{ij}=0\)) while allowing non-zero expansion leads to a dramatically different energy-density profile for the irrotational drive. The requirements are substantially smaller in absolute magnitude (e.g., its peak \(\varrho_p\) is \(\approx 15\%\) of Alcubierre's peak \(|q_s|\), a reduction factor of nearly 7), and the required stress-energy is predominantly positive and algebraically Type I  (classical energy-density) everywhere, unlike the large negative densities and Type IV (complex-null) zones characteristic of the Alcubierre drive.

\paragraph*{Natário warp drive}\label{subpar:Natario} Figure~\ref{fig:Natario_eigenvalues_radial} displays the eigenvalue behavior for the Natário warp drive along radial cuts at \(\theta=0\) (axis of motion) and \(\theta=\pi/2\) (equatorial plane). The Natário drive exhibits radially varying Hawking–Ellis classifications along both cuts shown.

Along the axis ($\theta=0$) (Figs.~\ref{fig:natario_re_theta0} and Fig.~\ref{fig:natario_im_theta0}), the geometry transitions from Type IV for radii approximately \(r \lesssim \rho\) (indicated by a non-zero imaginary component \(\mu = \operatorname{Im}[\lambda_{G(0)}]\)) to Type I for \(r \gtrsim \rho\) (where \(\mu = 0\) and eigenvalues become purely real).

The equatorial cut (\(\theta=\pi/2\)), reveals a more complex sequence of alternating classifications:
\begin{itemize}
  \item \textbf{Type IV}: (\(\mu\neq0\)), from the origin \(r=0\) to the inner wall, 
  \item \textbf{Type I}:  (\(\mu=0\)), inner wall, with the proper energy density peaking at $\varrho_p \;=\;-\,\xi/\kappa \simeq -67\,[\text{m}^{-2}]/\kappa$(Fig.~\ref{fig:natario_re_thetaPi2}), i.e.\  \(146\) times the irrotational drive’s peak magnitude.
  \item \textbf{Type IV}: (\(\mu\neq0\)) the central region extending over \(r \approx \rho\), where $\mu=\operatorname{Im}[\lambda_{G(0)}]\simeq -80\,[\text{m}^{-2}]$ (Fig.~\ref{fig:natario_im_thetaPi2}), i.e.\  \(174\) times the irrotational drive’s peak.
  \item \textbf{Type I}: (\(\mu=0\)) outer wall.
  \item \textbf{Type IV}: (\(\mu\neq0\)) in the far field (\(r\to\infty\)).
\end{itemize}

The peak eigenvalue magnitudes  (both the real part \(\xi\) and imaginary part \(\mu\)) in the equatorial cut ($\theta=\pi/2$) exceed those along the axis of travel ($\theta=0$), and are significantly larger than those required for the Alcubierre and irrotational drives, reflecting the comparatively higher exotic-matter cost of the Natário solution, due to its zero volume expansion constraint.

\subsubsection{Null energy condition interpretation}

We now detail the interpretation of the null energy condition (NEC) based on the eigenvalue visualizations (Figs.~\ref{fig:main_3x3_grid A}–\ref{fig:main_3x3_grid C}). The NEC requires that \(T_{\mu\nu} k^\mu k^\nu \ge 0\), or equivalently using the mixed Einstein tensor tetrad components central to our analysis, \(G^{\hat{\sigma}}{}_{\hat{\beta}} k_{\hat{\sigma}} k^{\hat{\beta}} \ge 0\) for all null vectors \(k^{\hat{\beta}}\) (where \(k_{\hat{\sigma}} = \eta_{\hat{\sigma}\hat{\alpha}} k^{\hat{\alpha}}\) and the null condition is \(\eta_{\hat{\gamma}\hat{\delta}} k^{\hat{\gamma}} k^{\hat{\delta}} = 0 \)). For Type I regions, the NEC implies \(\varrho_p + p_k \ge 0\) for all principal pressures \(p_k\). Given the relations \(\kappa \, \varrho_p = -\lambda_{G(0)}\) and \(\kappa \, p_k = \lambda_{G(k)}\) adopted here (for \(k=1,2,3\)), the NEC condition for Type I regions becomes \(\lambda_{G(k)} - \lambda_{G(0)} \ge 0\).

The plots focusing on relevant eigenvalue combinations (Figs.~\ref{fig:sub_r2_c1 B}-\ref{fig:sub_r2_c3 B}, \ref{fig:sub_r1_c1 C}-\ref{fig:sub_r1_c3 C}) display the quantities \(\lambda_{G(k)} - \operatorname{Re}[\lambda_{G(0)}]\) (where \(\xi = \operatorname{Re}[\lambda_{G(0)}]\) and \(k=2, 3\)). For all three drives, the most significant NEC violations correspond to regions where these quantities become most negative. This typically occurs where the real eigenvalues \(\lambda_{G(2)}, \lambda_{G(3)}\), associated with principal pressures (transverse based on physical sorting), reach their most negative values relative to \(\xi\).

\paragraph*{Irrotational warp drive} Applying this to the globally Type I irrotational drive, Figs.~\ref{fig:main_3x3_grid B} and \ref{fig:main_3x3_grid C} indicate that the most significant NEC violations (the most negative values of \(\lambda_{G(k)} - \lambda_{G(0)}\)), approximately $-4. [\text{m}^{-2}]$, occur along the axis of motion ($\theta=0$) associated with \(\lambda_{G(2)} - \lambda_{G(0)}\) (Fig.~\ref{fig:sub_r2_c2 B}) and with \(\lambda_{G(3)} - \lambda_{G(0)}\) (Fig.~\ref{fig:sub_r1_c2 C}). The origin of this violation is clarified by Figs.~\ref{fig:sub31 lambda20} and \ref{fig:sub33 lambda30}, which show that the transverse pressure eigenvalues \(\lambda_{G(2)}\) and \(\lambda_{G(3)}\) themselves become significantly negative along this axis. These large transverse negative pressures (\(p_k = \lambda_{G(k)}/\kappa < 0\)) are responsible for \(\varrho_p + p_k < 0\), thus violating the NEC, even though \(\varrho_p\) is positive at \(\theta=0\) (consistent with \(-\lambda_{G(0)}\) peaking positively there).

\paragraph*{Alcubierre warp drive} In the Type IV regions characteristic of the Alcubierre and Natário drives, null energy condition (NEC) violation is inherently guaranteed by the algebraic structure associated with \(\mu \neq 0\) (cf. \cref{sec:classification_protocol_final_v5}). The observation from Figs.~\ref{fig:sub_r2_c1 B}-\ref{fig:sub_r2_c3 B} and \ref{fig:sub_r1_c1 C}-\ref{fig:sub_r1_c3 C} that the quantities \(\lambda_{G(k)} - \xi\) (where \(\xi = \operatorname{Re}[\lambda_{G(0)}]\) and \(k=2,3\)) become strongly negative indicates directions of particularly strong NEC violation. This arises from the interplay between terms related to the real eigenvalues \(\lambda_{G(k)}\) (associated with \(k=2,3\) transverse negative principal pressures) and \(\xi\).

For the Alcubierre drive, specifically, Fig.~\ref{fig:sub_r2_c1 B} indicates that significant NEC violations occur practically uniformly in the $xy$ prime meridian plane, both in the cut along the axis of motion ($\theta=0$) and in the equatorial cut  ($\theta=\pi/2$). At the \(\theta=0\) location, the geometry is Type IV, as confirmed by the eigenvalue behavior shown in the radial cuts in the direction of motion, Figs.~\ref{fig:sub31 lambda20} and \ref{fig:sub33 lambda30}. Here, the peak NEC violation of approximately $-4.3 [\text{m}^{-2}]$ is driven by the transverse pressure-related eigenvalues \(\lambda_{G(2)}\) and \(\lambda_{G(3)}\) becoming significantly negative relative to \(\xi\). The Type IV classification of the Alcubierre drive at this location contrasts with the irrotational and Natário drives, which are classified as Type I at the regions of largest NEC violation.

\paragraph*{Natário warp drive} Focusing on the Natário drive's equatorial cut (\(\theta=\pi/2\)), Fig.~\ref{fig:sub_r2_c3 B} reveals a maximum null energy condition (NEC) violation, where the relevant eigenvalue combination reaches a minimum \(\lambda_{G(2)} - \lambda_{G(0)} \approx -260\,[\text{m}^{-2}]\). This large negative value signifies a strong breach of the physical NEC requirement \(\varrho_p + p_k \ge 0\), specifically \(\varrho_p + p_2 \ll 0\) here. Interestingly, at this radial location of maximum violation, the geometry is locally Type I ($\mu=0$, Fig.~\ref{fig:natario_im_thetaPi2}). The violation \(\varrho_p + p_2 \ll 0\) occurs because a strongly negative principal pressure (\(p_2 = \lambda_{G(2)}/\kappa \approx -200 [\text{m}^{-2}]/\kappa\))  (Fig.~\ref{fig:natario_re_thetaPi2}) adds to the proper energy density, which is also negative here (\(\varrho_p = -\lambda_{G(0)}/\kappa \approx -60 [\text{m}^{-2}]/\kappa\), corresponding to \(\lambda_{G(0)} \approx +60\,[\text{m}^{-2}]\)) (Fig.~\ref{fig:natario_re_thetaPi2}). The peak NEC violation magnitude of $\approx -260\,[\text{m}^{-2}]$ for the Natário drive significantly exceeds (e.g., by a factor exceeding 60 times) the maximum violations found for the irrotational and Alcubierre drives.

\subsubsection{Weak, dominant, and strong energy conditions} For all three warp–drive geometries we found a violation of the null energy condition (NEC).  Since
\begin{equation}
\text{DEC}\;\Rightarrow\;\text{WEC}\;\Rightarrow\;\text{NEC},
\qquad
\text{SEC}\;\Rightarrow\;\text{NEC},
\end{equation}
their contrapositives give
\begin{equation}
\neg\text{NEC}\;\Longrightarrow\;
\neg\text{WEC}\;\wedge\;\neg\text{DEC}\;\wedge\;\neg\text{SEC}.
\label{eq:hierarchy}
\end{equation}
Therefore all three warp–drive spacetimes simultaneously violate the weak (WEC), dominant (DEC), and strong (SEC) energy conditions.

\subsubsection{Trace energy condition interpretation}

The trace energy condition (TEC) \cite{Curiel_Primer,Kontou_Energy} is imposed on the trace of the stress-energy tensor. This invariant scalar can be expressed as \(T \equiv T^{\hat{\alpha}}{}_{\hat{\alpha}}\). For our \((-,\,+,\,+,\,+)\) metric signature, the TEC requires this trace to be non-positive, i.e., \(T \le 0\). This is equivalent to the physical requirement that the proper energy density \(\varrho_p\) must be greater than or equal to the sum of the principal pressures \(p_k\), \(\varrho_p \ge \sum_k p_k\). Through Einstein's field equations, \(G^{\hat{\alpha}}{}_{\hat{\beta}} = \kappa \, T^{\hat{\alpha}}{}_{\hat{\beta}}\), this condition on \(T\) directly translates to \(G \equiv G^{\hat{\alpha}}{}_{\hat{\alpha}}\), via \(G = \kappa \, T\). Consequently, the TEC ($T \le 0$) implies a non-positive Einstein tensor scalar (\(G \le 0\)) and thus a non-negative Ricci scalar (\(R = -G \ge 0\)). The scalar invariant \(G = \sum_{\alpha=0}^{3}\lambda_{G(\alpha)}\) (equal to \(\mathrm{Tr}\,M\)) provided a valuable check on our numerical eigensystem calculations to confirm the precise cancellation of imaginary parts from any complex conjugate eigenvalues. 
Historically, the notion that matter might be constrained in a way consistent with the TEC was foreshadowed by proposals regarding its ultimate equation of state. J. von Neumann suggested (as noted in 1935~\cite{Chandrasekhar1935}) a limiting equation of state \(p = (1/3) \varrho\), which precisely saturates the TEC requirement that \(p \le \rho/3\).  The TEC was largely abandoned after 1961 when Zel’dovich \cite{Zeldovich:1961sbr} showed that plausible ultra-dense matter, characterized by stiff equations of state such as \(p \approx \varrho_p\), can violate the TEC. Bekenstein~\cite{Bekenstein} introduced in 2013 a subdominant trace energy condition (STEC) for ordinary (baryonic) matter, requiring its stress-energy to satisfy a TEC-like constraint (\(\rho \gtrsim \sum_k p_k\)) that dominates exotic contributions from quantum vacuum energy. Consequently, the STEC ensures that the positive mass-energy of any physical apparatus would outweigh localized negative vacuum energy it contains, precluding net negative total mass and demonstrating a role for modified TEC conditions in arguments for global mass positivity in semi-classical gravity.
We visualize the behavior of the Einstein tensor scalar invariant \(G\) (Figs.~\ref{fig:sub_r3_c1 C}–\ref{fig:sub_r3_c3 C}) in this TEC context and now proceed to discuss its implications to further distinguish the drives.

\paragraph*{Alcubierre and irrotational warp drives} Figures~\ref{fig:sub_r3_c1 C} and \ref{fig:sub_r3_c2 C} indicate that the Alcubierre and irrotational drives exhibit similar spatial distributions and peak magnitudes for \(G\). Both drives violate the TEC, with \(G\) attaining positive values exceeding \(\approx 5\,[\text{m}^{-2}]\), these maximum violations occurring along the axis of motion ($\theta=0$).

\paragraph*{Natário warp drive} 
Our analytic result (\cref{prop:TEC_Natario}) proves that the zero‑expansion Natário solution satisfies the trace–energy condition everywhere: with the lapse fixed to $\alpha=1$, a flat spatial metric ($R^{(3)}=0$), and vanishing expansion $K=0$, the Gauss–Codazzi scalar identity reduces to $R^{(4)}=K_{ij}K^{ij}\!\ge 0$ and hence $G=-R^{(4)}\le 0\quad\text{point‑wise}$.

The numerical profile (Fig.~\ref{fig:sub_r3_c3 C}) agrees,
$G$ ranging from $0$ at the bubble centre to about
\(-120\,[\mathrm{m}^{-2}]\); the most negative values 
(strongest TEC satisfaction) occur perpendicular to the direction
of motion (\(\theta=\pi/2\)).
The non‑positivity of \(G=\sum_{\sigma=0}^{3}\lambda_{G(\sigma)}\)
is governed by \(\lambda_{G(2)}\): as Fig.~\ref{fig:natario_re_thetaPi2}
shows, this eigenvalue becomes strongly negative, giving a substantial
negative pressure \(p_{(2)}=\lambda_{G(2)}/\kappa\) and thereby enforcing  
\(G\le 0\) throughout the spacetime.

\section{Limitations \& Outlook}
\label{sec:limits}

\textbf{Scope assumptions.}
We work with unit lapse $\alpha=1$, a flat and time–independent spatial metric, and an irrotational shift $\beta_i=\partial_i\Phi$ with constant $v$ in the maps and global measures. These idealizations yield a clean Type~I eigenstructure and enable direct model-to-model comparisons. Convergence as $r\to 0$ is analyzed separately in Appendix \ref{app:gr0}.

\textbf{Global budgets are mask–defined.}
All integrated energies and volumes use the disjoint regions $\mathcal R_-$, $\mathcal R_0$, $\mathcal R_+$ built from the Hessian scalar $\lambda_H=\tfrac12(\mathrm{tr}\,H^2-(\mathrm{tr}\,H)^2)$ with adaptive absolute thresholds (Appendix~\ref{app:global_energy}). With our sign convention
$T^{(0)}{}_{(0)}=-\varrho_p$ and $\kappa\,\varrho_p=-\lambda_H$, the reported magnitudes
$E_\pm=\int_{\mathcal R_\pm}|\varrho_p|\,dV=\kappa^{-1}\!\int_{\mathcal R_\pm}|\lambda_H|\,dV$
and volumes $V_\pm, V_0$ are numerically robust; the zero band $V_0$ quantifies the tolerance deadband rather than physics. We integrate over an extended window $r\in[0,12\rho]$ and report both finite–window ($R_{\rm integ}=12\rho$) totals and \emph{tail–corrected} ($R\!\to\!\infty$) estimates obtained via a two–point $1/R$ extrapolation using $R\in\{8\rho,12\rho\}$ (App.~\ref{app:global_energy}, Eq.~\eqref{eq:tail_two_point}). Jacobian and mask–partition checks pass at roundoff.

\textbf{Interpretation of “exotic matter”.}
In the canonical run ($\rho=5~[\mathrm{m}]$, $\sigma=4~[\mathrm{m}^{-1}]$, $v/c=1$) on the $R_{\rm integ}=12\rho$ window we find $E_{+}/E_{-}=1.07$ (3 s.f.) and, on that finite window, a negative–energy volume fraction near $70\%$ (Tables~\ref{tab:globals_SI}–\ref{tab:ratios_checks}). After applying the $1/R$ tail extrapolation (App.~\ref{app:global_energy}, Table~\ref{tab:baseline_vs_tail}), the \emph{tail–corrected} magnitudes satisfy
\[
E_-(\infty)=1.33\times10^{44}\ \mathrm{J},\qquad
E_+(\infty)=1.33\times10^{44}\ \mathrm{J},
\]
so that the normalized net is $|E_{\rm net}|/E_{\rm abs}(\infty)=0.04\%$ (four decimals in fractional units). Within our stated reporting precision (central values to 3\,s.f., percentages to two decimals), the net proper energy is therefore \emph{consistent with zero}; accordingly, we adopt \(E_{\rm net}(\infty)=0\) as the best estimate. Equatorial localization metrics show that the negative part is diffuse and concentrated near $\theta\simeq\pi/2$, whereas the positive part forms high–amplitude, crescent–like thin collars near $\theta\simeq 0,\pi$.

\textbf{What changes if assumptions relax.}
Letting $\alpha$ or $v$ vary in time introduces explicit time dependence in $K_{ij}$ and can modify averages and alignments, potentially reintroducing off–diagonal fluxes in mixed components and creating Type~IV pockets. Departing from irrotational flow (adding vorticity) would break the simple Hessian reduction and alter $\lambda_H$–based diagnostics.

\textbf{Outlook.}
(i) Extend to gently time–dependent $v(t)$ and $\alpha(t,\mathbf{x})$;
(ii) couple to matter models consistent with Type~I eigenstructure;
(iii) assess averaged and quantum energy–condition constraints in semiclassical regimes.

\section{Conclusions}

This work has provided a detailed analytical and numerical study of a novel \emph{kinematically irrotational} warp-drive spacetime within General Relativity. Building on the work of Santiago, Schuster, and Visser \cite{SantiagoVisser}---who showed that kinematically irrotational warp drives with unit lapse ($\alpha=1$) and flat spatial slices have vanishing momentum density and therefore fall into the Hawking–Ellis Type I class---this paper advances the field by explicitly constructing, analyzing, and cross‑comparing concrete warp‑drive spacetimes with significantly improved properties.

The key new contributions presented in this paper, ordered by importance and novelty, are as follows:
\begin{enumerate}

\item We present an explicit, closed-form analytical construction of a kinematically irrotational warp drive—specified by the scalar potential $\Phi(r,\theta,t)=v(t)\,r\,g(r)\cos\theta$ \eqref{eq:irrot_potential} and the associated \emph{continuous} shift components $\beta^i$ \eqref{eq:BetaIrr}, \eqref{g(r)}—satisfying $\boldsymbol{\omega}=* (d\boldsymbol{\beta}^{\flat})=0$ with \emph{smooth} boundary behavior, thereby avoiding the discontinuities of some prior proposals; to our knowledge, no previous work has provided a continuous, analytically derived irrotational shift of this form.

        \item Causal evidence from a fixed–smoothing ablation.
    Holding the baseline smoothing profile \(g(r)\) and mask tolerances fixed, we introduce a tunable vortical component and recompute the global budgets (Appendix~\ref{app:ablation}, Table~\ref{tab:abl_sweep}). Even modest vorticity collapses the near-cancellation of proper energy: on the \(R_{\rm integ}=12\rho\) window the ratio \(E_{+}/E_{-}\) falls below \(10^{-2}\) by \(\eta=0.25\), and at \(\eta=1\) the negative-energy magnitude \(E_{-}\) increases by a factor \(>500\).
    This confirms that the favorable energy profile of our irrotational model is \emph{caused} by its curl-free kinematics rather than by profile smoothing.

\item We develop a systematic, high-precision (50-digit) Cartan–tetrad numerical pipeline for pointwise eigenvalue–eigenvector analysis of the full Einstein tensor, incorporating Hawking–Ellis classification and energy-condition diagnostics. Running it on our closed-form irrotational solution provides an explicit, full-spacetime confirmation of a warp-drive geometry that is globally Type I—corroborating the demonstration of \cite{SantiagoVisser} and extending the zero-vorticity study of \cite{SchusterSantiagoVisser2023}. As an internal validation, the timelike eigenvalue obtained by direct diagonalization, \(\lambda_{G(0)}\), agrees to roundoff with the independent Hessian-invariant evaluation \(\lambda_H\) (Appendix~\ref{app:global_energy}; Eqs.~\eqref{eq:lambdaH_def}, \eqref{eq:rho_from_lambdaH}).

    \item We formally demonstrate (\cref{prop:mu_invariance}) that for Hawking–Ellis Type IV (complex null) regions the real part \(\xi\) and the magnitude \(\lvert\mu\rvert\) of the complex eigenvalues are Lorentz scalar invariants, thereby justifying invariant‐based comparisons of the irreducible flux \(\lvert q_s\rvert=\lvert\mu\rvert/\kappa\) and observer‐dependent density \(\lvert\varrho_{\mathrm{obs}}\rvert=\lvert\xi\rvert/\kappa\) with the Type I proper energy density \(\varrho_p\).

    \item We extend this numerical Hawking–Ellis classification analysis to the Alcubierre and Natário (zero-expansion) warp drives. While these spacetimes can exhibit localized Type IV regions, in our analysis none of the models investigated—including the newly constructed irrotational drive—contain any regions of Type II or Type III stress-energy classification.

    \item We demonstrate that the irrotational warp drive achieves a predominantly \emph{positive} proper energy density \(\varrho_p\) (peaking at \(0.46\,[\mathrm{m}^{-2}]/\kappa\) on its axis of motion, \(\theta=0\)), with only a small, localized dip below zero (\(\varrho_p\approx -0.026\,[\mathrm{m}^{-2}]/\kappa\)) at the equatorial cut (\(\theta=\pi/2\)).  Relative to the \emph{peak proper-energy deficit}, this minimum is reduced by a factor of \(\approx 38\) compared with Alcubierre (\(\min\varrho_p\approx -1/\kappa\)) and by \(\approx 2.6\times 10^{3}\) compared with Natário (\(\min\varrho_p\approx -67/\kappa\)). While these relative reductions are significant, the absolute stress–energy scales, \(T^{\hat{\mu}}{}_{\hat{\nu}}= G^{\hat{\mu}}{}_{\hat{\nu}}/\kappa\), remain astronomically large within standard General Relativity. Crucially, the irrotational drive requires classical Type I energy rather than the complex-null Type IV exotic energy of the other models.

     \item Whole–spacetime energy/volume summary. On the baseline window \(R_{\rm integ}=12\rho\) we find
\(E_{+}/E_{-}=1.07\) (3 s.f.) and \(V_{-}/(V_{-}{+}V_{+})\approx 0.732\).
After applying a two-point \(1/R\) tail extrapolation with \(R\in\{8\rho,12\rho\}\)
(App.~\ref{app:global_energy}, Eq.~\eqref{eq:tail_two_point}, Table~\ref{tab:baseline_vs_tail}),
the tail-corrected magnitudes are
\(E_-(\infty)=1.33\times10^{44}\ \mathrm{J}\) and \(E_+(\infty)=1.33\times10^{44}\ \mathrm{J}\),
with \(E_{\rm net}(\infty)=-1.12\times10^{41}\ \mathrm{J}\) and
\(|E_{\rm net}|/E_{\rm abs}(\infty)=0.04\%\) (four decimals in fractional units).
 Within our stated reporting precision, the net proper energy is \emph{consistent with zero}; our best estimate is \(E_{\rm net}(\infty)\approx 0\).
As visualized in Appendix~\ref{app:global_energy}, Fig.~C1 (prime meridional slice), the positive energy density support is compact and concentrated near the bubble wall (\(r\simeq \rho\)), whereas the negative energy density is diffuse and significantly lower in amplitude.

    \item All three warp–drive spacetimes violate the NEC—and hence the WEC, DEC, and SEC—due to large \emph{negative transverse principal pressures}.  Although superficially akin to the isotropic negative pressure of dark energy, these pressures are strongly anisotropic (and confined to the bubble wall).  Quantitatively, the Natário drive’s peak NEC violation is $65\times$ and $60\times$ larger than those in the irrotational and Alcubierre drives, respectively (nearly two orders of magnitude).  Moreover, the location of maximal violation differs: for the irrotational drive it lies along the axis of motion ($\theta=0$), whereas for both the Alcubierre and Natário drives it lies on the equatorial cut ($\theta=\pi/2$).

    \item We present a comparative analysis of the trace energy condition (TEC), proving (\cref{prop:TEC_Natario}) that, under the stated assumptions, the Natário zero-expansion warp drive satisfies \(G\le 0\) throughout the domain considered, while both the Alcubierre and the new irrotational drives exhibit regions with \(G>0\).

    \item We emphasize that, although one can mathematically impose both zero divergence (\(\nabla\!\cdot\!\boldsymbol{\beta}=0\)) and zero vorticity (\(\boldsymbol{\omega}=\nabla\times\boldsymbol{\beta}=0\))—which forces \(\boldsymbol{\beta}=-\nabla\Phi\) with \(\Phi\) harmonic (\(\nabla^2\Phi=0\))—this severely limits the potentials available for a localized warp bubble. In contrast, our irrotational drive employs a deliberately non-harmonic \(\Phi\) (so \(\nabla\!\cdot\!\boldsymbol{\beta}=-\nabla^2\Phi\neq0\)), thereby opening a broader solution space and yielding markedly improved energy-condition profiles.

\end{enumerate}

\smallskip

These results collectively advance the theoretical and numerical understanding of warp-drive spacetimes by delivering what is, to our knowledge, the first explicit, closed-form irrotational solution in an orthonormal-tetrad (Cartan) framework.  Built on a smooth \(C^\infty\) form-function \(f(r)\), our model yields globally continuous frame fields and spin coefficients, and exhibits markedly improved energy-condition profiles compared to both Alcubierre and Natário drives.  Moreover, we provide the first comprehensive, pointwise Hawking–Ellis classification across all three geometries, and explain the origin of the discrete jumps in classification seen in Alcubierre and Natário spacetimes:  although the Einstein tensor and its eigenvalues vary smoothly, the algebraic type flips abruptly whenever key thresholds are crossed (for example, when the imaginary part \(\mu\) of a complex eigenvalue vanishes or when an eigenvector’s causal character changes).  In particular, the Type IV subtype—requiring strictly null complex eigenvectors—is structurally unstable in the usual sense: arbitrarily small perturbations generically drive \(\mu\) away from zero.  By combining these conceptual insights with our high-precision analytic–numerical Cartan‐tetrad pipeline, we lay a robust foundation for future explorations of exotic spacetime geometries and the theoretical limits of faster‑than‑light travel in General Relativity.  All quantitative statements above hold under the stated assumptions of unit lapse (\(\alpha=1\)) and stationary flat spatial slices; relaxing these (e.g., during acceleration) can reintroduce momentum density and localized Type~IV regions.

\backmatter

\section*{Declarations}

\begin{itemize}
\item Funding
No funding was received for conducting this study.
\item Conflict of interest/Competing interests 
The author declares no competing interests.
\item Ethics approval and consent to participate
Not applicable
\item Consent for publication
Not applicable
\item Data availability
This is a theoretical study. All results are derived from the equations presented within the manuscript. Numerical data used to generate the figures are available from the author upon reasonable request.
\item Materials availability
Not applicable
\item Code availability 
The custom Wolfram Mathematica scripts used for numerical calculations and figure generation are available from the author upon reasonable request.
\item Author contribution
The sole author was responsible for all aspects of this work, including conceptualization, methodology, analysis, and manuscript preparation.
\end{itemize}

\begin{appendices}

\section{Notation, units, and coordinate system}\label{Notation}

 This section introduces the notation and units used throughout this paper and presents the coordinate system.
 
\smallskip

\begin{itemize}
    \item Spatial 3-tensors are denoted by the index $(3)$ on the top left of the tensor: $^{(3)}R_{i j}$. 
     \item Indices: Greek $\mu,\nu=0\ldots3$ (spacetime), Latin $i,j=1\ldots3$ (space). Signature $(-,+,+,+)$.
    \item Metrics: $g_{\mu\nu}={}^{(4)}g_{\mu\nu}$, $\gamma_{ij}={}^{(3)}g_{ij}$; flat spatial slices in Cartesian. 
    \item 3+1: lapse $\alpha$ (set to $1$), shift $\beta^i$, extrinsic curvature $K_{ij}$, with conventions as in~\cite{Gourgoulhon,AlcubierreBook,baumgarte_shapiro_2010}.
    \item Covariant differentiation in space (in 3+1 decomposition) is denoted by $D_i$ or $D^i$.
    \item The components of the Riemann curvature tensor $^{(4)}R^{\kappa}_{\; \mu \kappa \nu}$ and of the Ricci curvature tensor $^{(4)}R_{\mu \nu}$ are defined following the conventions in Misner et al. \cite{Misner1973}, where $R^{\alpha}{}_{\beta\mu\nu} = \partial_\mu \Gamma^{\alpha}_{\beta\nu} - \dots$ and the Ricci tensor is $^{(4)}R_{\mu \nu} = {^{(4)}R^{\kappa}}_{\; \mu \kappa \nu}$.
    \item Tetrad components of tensors are denoted with carets (hats) above the indices: $K_{\hat{i}\hat{j} }$, as in \cite{Misner1973}.
    \item Tetrad components aligned with the eigenvectors of a tensor and whose magnitudes equal the corresponding eigenvalues are distinguished by enclosing the indices in parentheses: \( T^{(\mu)}{}_{(\nu)} \). 
    \item Units are denoted by square brackets, e.g., $[\text{m}]$. Geometrized (natural) units are employed in the analysis. SI units are used in a few physical examples. Specifically:
    \begin{itemize}
       \item $x, y, z, r, \rho$: $[\text{m}]$ in both geometrized and in SI.
       \item $G, G^{00}, G_{i}^{\ 0}, G^{0}_{\ i}$: $[\text{m}^{-2}]$ in both geometrized units and in SI units.
       \item $\sigma$: $[\text{m}^{-1}]$ in both geometrized and in SI.
       \item $g_{\mu \nu},g^{\mu \nu},\gamma_{ij},\gamma^{ij}, f, \theta,\phi$: $[\text{1}]$ (dimensionless) in both geometrized and in SI.
       \item $v, c$: $[\text{1}]$ (dimensionless) in geometrized and $[\text{m}\, \text{s}^{-1}]$ in SI.
       \item $t$: $[\text{m}]$ in geometrized and $[\text{s}]$ in SI.
       \item $\kappa$: $=8\pi G/c^{4}$: $[\text{1}]$ (dimensionless) in geometrized; and $[\text{N}^{-1}]$ in SI.
       \item $T, T^{00}=\varrho, T_{i}^{\ 0}, T^{0}_{\ i}$: $[\text{m}^{-2}]$ in geometrized and $[\text{N}\, \text{m}^{-2}] = [\text{J} \, \text{m}^{-3}]$ in SI.
     \end{itemize}

      \item We adopt Natário's \cite{Natário} right-handed spherical coordinate system $(r, \theta, \phi)$ (Fig.~\ref{Fig1_SpherCoord}) with orthonormal basis vectors \( \mathbf{e_{r}} \), \( \mathbf{e_{\theta}} \), \( \mathbf{e_{\phi}} \). The polar axis is the direction of travel: the $x$-axis. The coordinates are the radial distance $r = \sqrt{x^2 + y^2 + z^2} \in [0, \infty)$, the polar angle $\theta = \arccos(x/r) \in [0, \pi]$ measured from the positive $x$-axis, and the azimuthal angle $\phi = \operatorname{atan2}(z, y) \in [-\pi, \pi]$. Surfaces $\theta=$ const define circular cones centered on the $x$-axis (vertex at origin), with $\theta=\pi/2$ being the $yz$ (equatorial) plane. Constant $r$ and $\theta$ define parallel circles centered on the $x$-axis, analogous to parallels of latitude on Earth. The cylindrical radius (distance from the $x$-axis) is $\mathsf{r} = r \sin \theta$. The spherical warp bubble boundary is defined by $r=\rho$. Its `equator' (intersection with the $yz$-plane, $\theta=\pi/2$) is the largest of these parallels on the bubble surface and has cylindrical radius $\mathsf{r} = \rho \sin(\pi/2) = \rho$. The two-argument function $\operatorname{atan2}(z, y)$ computes the angle $\phi$ in the $yz$-plane measured from the positive $y$-axis to the projection of the position vector. Surfaces $\phi=$const are half-planes containing the $x$-axis; their intersection with the bubble boundary $r=\rho$ yields circles of radius $\rho$ (analogous to meridians, i.e.\ lines of constant longitude). We refer to the \(xy\)-plane (\(\phi = 0\)) as the prime meridian plane.
\end{itemize}


\section{Regularity and asymptotics of the irrotational profile \texorpdfstring{$g(r)$}{g(r)} at the origin}
\label{app:gr0}

\subsection*{Axioms}
(1) \emph{Parameters:} $\sigma>0$ and $\rho>0$ are fixed. \;
(2) \emph{Analyticity:} Define
\begin{equation}
\label{eq:falc_def}
f_{\rm Alc}(r)
:= \frac{\tanh\!\big(\sigma(r+\rho)\big)-\tanh\!\big(\sigma(r-\rho)\big)}
           {2\,\tanh(\sigma\rho)} ,
\qquad
f(r):=1-f_{\rm Alc}(r),
\end{equation}
and
\begin{equation}
\label{eq:g_def}
g(r) \;:=\; \frac{1}{r}\int_{0}^{r} f(s)\,\mathrm{d}s \qquad (r\neq 0), 
\quad\text{with } g(0) \text{ defined by the limit } r\to0 .
\end{equation}
Since $f_{\rm Alc}$ is a composition/product of analytic functions, $f_{\rm Alc},f\in C^{\infty}(\mathbb{R})$.
Moreover,
\begin{equation}
\label{eq:falc_alt}
f_{\rm Alc}(r)
\;=\;
\frac{\cosh^{2}(\sigma\rho)}
     {\cosh\!\big(\sigma(r+\rho)\big)\,\cosh\!\big(\sigma(r-\rho)\big)} ,
\end{equation}
which is manifestly even in $r$.

\subsection*{Proposition}
$g$ is $C^{\infty}$ near $r=0$ and
\begin{equation}
\label{eq:g_limit}
\lim_{r\to 0} g(r) \;=\; 0,
\qquad
g(r) \;=\; \frac{\sigma^{2}\operatorname{sech}^{2}(\sigma\rho)}{3}\,r^{2} \;+\; O(r^{4}) 
\quad (r\to 0).
\end{equation}

\begin{proof}
\emph{Step 1 (model–independent limit).}
By the Fundamental Theorem of Calculus,
$\dfrac{\mathrm{d}}{\mathrm{d}r}\!\big(\int_{0}^{r} f(s)\,\mathrm{d}s\big)=f(r)$.
Thus L’Hôpital’s rule gives
\begin{equation}
\label{eq:lhopital}
\lim_{r\to 0} g(r)
=\lim_{r\to 0}\frac{\int_{0}^{r} f(s)\,\mathrm{d}s}{r}
=\lim_{r\to0} f(r) = f(0).
\end{equation}
From \eqref{eq:falc_def} and the oddness of $\tanh$, $f_{\rm Alc}(0)=1$ and hence
\begin{equation}
\label{eq:f0}
f(0)=1-f_{\rm Alc}(0)=0.
\end{equation}
Therefore $\lim_{r\to0}g(r)=0$.

\medskip
\emph{Step 2 (leading coefficient for $f_{\rm Alc}$).}
Set $a:=\sigma\rho$, $b:=\sigma$ and write
\[
F(r):=f_{\rm Alc}(r)=\frac{\cosh^{2}a}{\cosh(a+br)\cosh(a-br)}.
\]
Then
\[
(\ln F)'(r)=-b\tanh(a+br)+b\tanh(a-br),\qquad (\ln F)''(r)=-b^{2}\sech^{2}(a+br)-b^{2}\sech^{2}(a-br),
\]
so $(\ln F)'(0)=0$ and $(\ln F)''(0)=-2b^{2}\sech^{2}a$. Since $F'(0)=0$, we have
\begin{equation}
\label{eq:falc_second_at0}
f''_{\rm Alc}(0)=F''(0)=F(0)\,(\ln F)''(0)=-\,2\sigma^{2}\sech^{2}(\sigma\rho).
\end{equation}
Thus
\begin{equation}
\label{eq:falc_series}
f_{\rm Alc}(r)=1+\frac{f''_{\rm Alc}(0)}{2}\,r^{2}+O(r^{4})
=1-\sigma^{2}\sech^{2}(\sigma\rho)\,r^{2}+O(r^{4}),
\end{equation}
and hence
\begin{equation}
\label{eq:f_series}
f(r)=1-f_{\rm Alc}(r)=\sigma^{2}\sech^{2}(\sigma\rho)\,r^{2}+O(r^{4}).
\end{equation}
Integrating and dividing by $r$,
\begin{equation}
\label{eq:g_series}
\begin{aligned}
\int_{0}^{r} f(s)\,\mathrm{d}s \; &=\; \frac{\sigma^{2}\operatorname{sech}^{2}(\sigma\rho)}{3}\,r^{3} + O(r^{5}),\\
g(r) \; &=\; \frac{1}{r}\!\int_{0}^{r} f(s)\,\mathrm{d}s \;=\; \frac{\sigma^{2}\operatorname{sech}^{2}(\sigma\rho)}{3}\,r^{2} + O(r^{4}).
\end{aligned}
\end{equation}
which proves \eqref{eq:g_limit}.
\end{proof}
\subsection*{Corollary (General limit for any \texorpdfstring{$C^{1}$}{C¹} form function).}
If $f\in C^{1}$ near $0$ and $f(0)=0$, then $\displaystyle\lim_{r\to 0} \frac{1}{r}\int_{0}^{r} f(s)\,\mathrm{d}s=0$.
\emph{Proof.} Either by the Fundamental Theorem of Calculus and L’Hôpital's rule (as in Step~1) or by the integral mean value theorem:
$\frac{1}{r}\!\int_{0}^{r} f=\!f(\xi)\to f(0)=0$ for some $\xi\in(0,r)$. \qed

\subsection*{Remark 1 (Evenness and absence of a linear term)}
From \eqref{eq:falc_alt}, $f_{\rm Alc}$ is even; thus $f=1-f_{\rm Alc}$ is even, $f'(0)=0$, and $f(r)=O(r^{2})$, yielding $g(r)=O(r^{2})$ by \eqref{eq:g_series}.
This reflects axial symmetry about the direction of motion.

\subsection*{Remark 2 (Integration constant)}
Definition \eqref{eq:g_def} fixes $g$. Any primitive
$\tilde g(r)=\frac{1}{r}\big(\int_{0}^{r}f+C\big)$ is regular at $r=0$ only if $C=0$; otherwise $\tilde g(r)\sim C/r$.

\subsection*{Remark 3 (Hypothetical linear term and radial \texorpdfstring{$C^{1}$}{C¹} regularity)}
If one posits $f(r)=a\,r+O(r^{2})$ for small $r>0$, then
$\frac{1}{r}\!\int_{0}^{r} f(s)\,\mathrm{d}s=\frac{a}{2}r+O(r^{2})\to0$, so $g(0)=0$ still holds. However, such a linear term is \emph{incompatible with a $C^{1}$, rotationally symmetric scalar field on $\mathbb{R}^{3}$}: writing $u(\mathbf{x})=\phi(|\mathbf{x}|)$, differentiability at $\mathbf{0}$ forces $\phi'(0)=0$ (otherwise directional derivatives at $\mathbf{0}$ would be a nonzero constant, contradicting linearity of $v\mapsto\nabla u(0)\!\cdot\! v$ on the unit sphere). In our model, $f=1-f_{\rm Alc}$ is even by construction (cf.\ \eqref{eq:falc_alt}).

\makeatletter
\@ifundefined{marginparwidth}{\newlength{\marginparwidth}}{}
\setlength{\marginparwidth}{2cm}
\makeatother

\setcounter{equation}{0}
\numberwithin{equation}{section}

\section{Global Energy Budget, Volumes, Tail Model, and Localization Diagnostics}
\label{app:global_energy}

This appendix quantifies the global negative-- and positive--energy budgets and their spatial distribution on the numerical window used for integration, and then attaches a physics-- and numerics--based tail correction to estimate the contribution from $r>R_{\rm integ}$.

\subsection*{Field, masks, and nondimensionalization}

We compute the spatial Hessian \(H_{ij}=\partial_i\partial_j\Phi\) and the rotationally invariant scalar
\begin{equation}
\lambda_H \;=\; \tfrac12\!\left(\mathrm{Tr}[H^2]-\bigl(\mathrm{Tr}\,H\bigr)^2\right),
\label{eq:lambdaH_def}
\end{equation}
which, for an irrotational shift \(\beta_i=\partial_i\Phi\) on a stationary flat slice, obeys
\begin{equation}
K_{ij}K^{ij}-K^2 \;=\; 2\,\lambda_H \;= -2\kappa\,\varrho_p,
\label{eq:ADM_link}
\end{equation}
so that the proper energy density follows from Einstein’s equations as
\begin{equation}
\varrho_p \;=\; -\,\frac{\lambda_H}{\kappa},
\qquad
\kappa=8\pi\ \text{(geom.)},\ \ 
\kappa=2.07665\times10^{-43}\ [\mathrm{N}^{-1}]\ \text{(SI)} .
\label{eq:rho_from_lambdaH}
\end{equation}

\paragraph*{Basis and notation.}
Indices \(i,j\) are \emph{Cartesian} on the spatial slice. In an orthonormal tetrad,
\(T^{(0)}{}_{(0)}=-\varrho_p\).
Using \(G^{(a)}{}_{(b)}=\kappa\,T^{(a)}{}_{(b)}\) one has
\(\lambda_{G(0)}=\kappa\,\lambda_{T(0)}=-\kappa\,\varrho_p=\lambda_H\), hence
\(\varrho_p=-\lambda_H/\kappa\).

\paragraph*{Independent computation of the timelike eigenvalue.}
Besides extracting the timelike eigenvalue \(\lambda_{G(0)}\) in the main text by direct diagonalization of the mixed Einstein tensor,
\begin{equation}
G^{(a)}{}_{(b)}\,u_{(0)}^{(b)}=\lambda_{G(0)}\,u_{(0)}^{(a)},
\label{eq:tla_direct_corrected}
\end{equation}
we compute it here independently, from the Hessian route developed above. Using Einstein’s equations \(G^{(a)}{}_{(b)}=\kappa\,T^{(a)}{}_{(b)}\) and the invariant
\(\lambda_H=\tfrac12\big(\mathrm{Tr}[H^2]-(\mathrm{Tr}\,H)^2\big)\),
the timelike eigenvalues of \(G^{(a)}{}_{(b)}\) and \(T^{(a)}{}_{(b)}\) satisfy
\begin{equation}
\lambda_{G(0)}=\kappa\,\lambda_{T(0)}=-\kappa\,\varrho_p=\lambda_H.
\label{eq:map_lambdaH_to_eigs_corrected}
\end{equation}
Thus the timelike eigenvalue obtained by direct diagonalization
(Eq.~\eqref{eq:tla_direct_corrected}) is identically reproduced by the Hessian construction \(\lambda_H\)
without any eigenvector ordering or degeneracy issues. Numerically, over the reported grids and working precision,
the two evaluations of \(\lambda_{G(0)}\) agree to roundoff, providing a robust internal consistency check on \texttt{Eigensystem} and on the analytic–derivative implementation used to build \(H_{ij}\).

\paragraph*{Disjoint masks and tolerances.}
We partition space by the sign of \(\lambda_H\) with absolute deadbands:
\begin{equation}
\begin{aligned}
\mathcal R_- &= \bigl\{\lambda_H>\lambda_{\mathrm{tol}}^{(+)}\bigr\},\qquad
\mathcal R_+ = \bigl\{\lambda_H<-\lambda_{\mathrm{tol}}^{(-)}\bigr\},\\
\mathcal R_0 &= \bigl\{\,|\lambda_H|\le
\max\!\bigl(\lambda_{\mathrm{tol}}^{(+)},\lambda_{\mathrm{tol}}^{(-)}\bigr)\,\bigr\},
\end{aligned}
\label{eq:regions}
\end{equation}
\begin{equation}
\lambda_{\mathrm{tol}}^{(\pm)}
\;=\;
\max\!\bigl\{\lambda_{\rm floor},\ \varepsilon\,\lambda^{(\pm)}_{\max}\bigr\},
\qquad
\lambda_{\rm floor}=10^{-20},\ \ \varepsilon=10^{-6}.
\label{eq:tolerances}
\end{equation}

\subsection*{Grid, volume element, and numerical window}

With axial symmetry the azimuthal angle is integrated analytically, so
\begin{equation}
dV \;=\; 2\pi\,r^{2}\sin\theta \; dr \; d\theta .
\label{eq:vol_element}
\end{equation}
We use a uniform tensor grid \((n_{r},n_{\theta})=(1601,1441)\) at 40-digit working precision.
The integration window is \(r\in[0,R_{\rm integ}]\) with \(R_{\rm integ}=12\,\rho\).
A posteriori checks (Eq.~\eqref{eq:mask_jacobian_checks}) confirm the Jacobian normalization and mask disjointness to machine accuracy.

\paragraph*{Reporting precision.}
Consistent with common practice, we round \emph{central values} to \(\sim\)3 significant figures (s.f.) unless otherwise noted; \emph{percentages} are reported to two decimals; and \emph{uncertainties} are quoted with 1--2 s.f., with central values rounded to the same decimal place as their uncertainties.

\paragraph*{Nondimensionalization.}
For the irrotational profile, analytic structure gives \(\lambda_H \sim v^{2}\sigma^{2}\,F(\sigma r,\theta)\).
We therefore present the dimensionless field (for \(\sigma v\neq 0\))
\begin{equation}
\widetilde{\varrho}_p \;\equiv\;
\frac{\kappa\,\varrho_p}{\sigma^{2}(v/c)^{2}}
\;=\; -\,\frac{\lambda_H}{\sigma^{2}(v/c)^{2}},
\label{eq:rho_tilde}
\end{equation}
on dimensionless axes (for \(\rho\neq 0\))
\begin{equation}
(\tilde x,\tilde y) \;=\; \Bigl(\frac{x}{\rho},\,\frac{y}{\rho}\Bigr) ,
\label{eq:axes_dimless}
\end{equation}
so that runs with the same \(\sigma\rho\) therefore collapse onto a single nondimensional map, while the warp-bubble wall remains near \(r/\rho\simeq 1\).
.

\subsection*{Prime meridional plane energy density map}

\begin{figure}[t]
  \centering
  \includegraphics[width=0.82\linewidth]{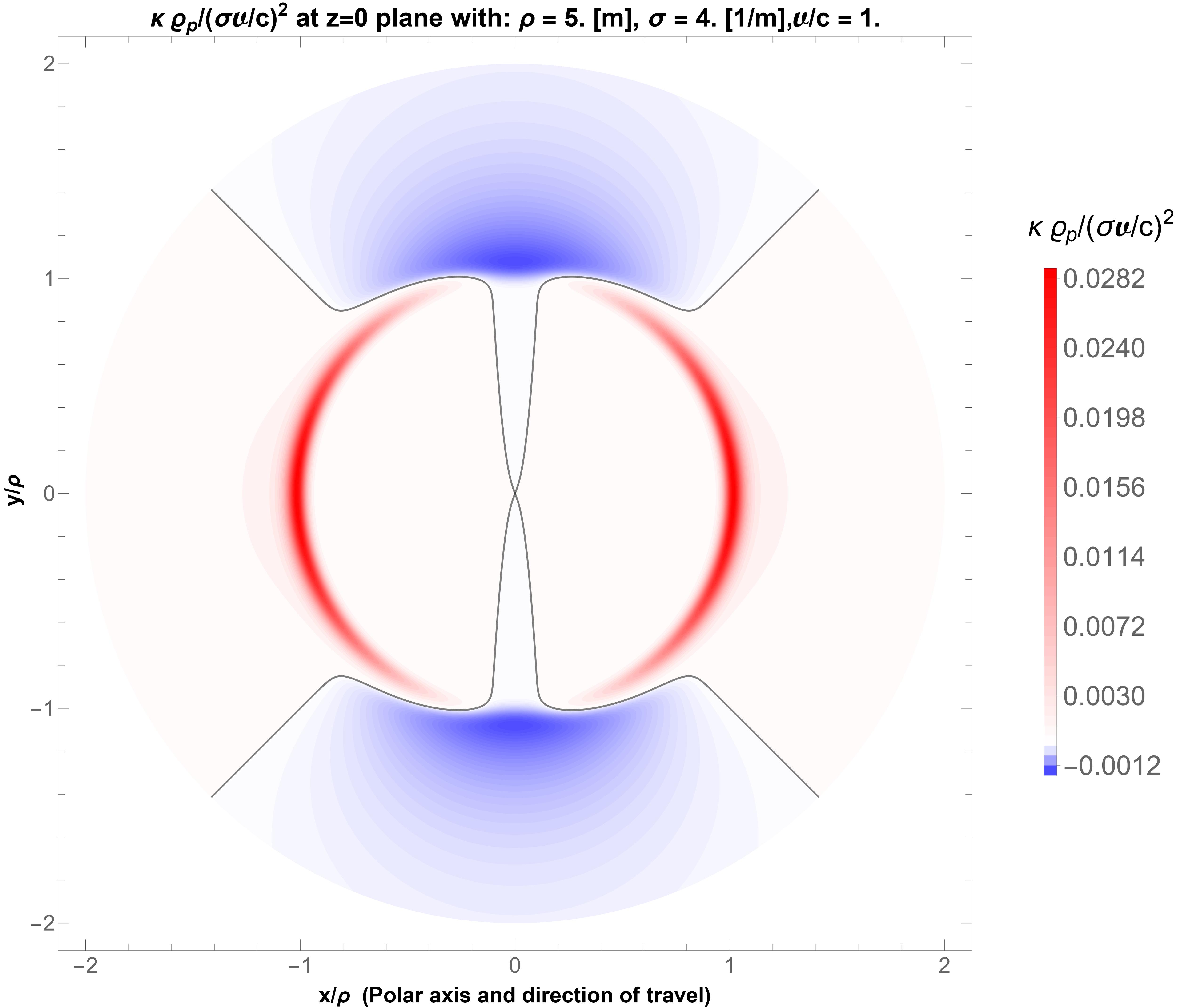}
\caption{Prime meridional slice (\(z=0\), \(xy\)-plane) of the
nondimensionalized proper energy density \(\kappa\,\varrho_p/(\sigma v/c)^{2}\) versus
nondimensional axes \(x/\rho\) (polar axis and direction of motion) and \(y/\rho\)
(transverse in–plane axis). The solid black curve is the null energy density \(\varrho_p=0\) contour
separating \(\varrho_p<0\) (blue) from \(\varrho_p>0\) (red). Parameters:
\(\rho=5~[\mathrm{m}]\), \(\sigma=4~[\mathrm{m}^{-1}]\), \(v/c=1\).}
  \label{fig:eq_density}
\end{figure}

On the plotted (zoomed) disk of Fig.~\ref{fig:eq_density} with \(R_{\rm plot}=2\rho\), the SI extrema are
\(\min\varrho_p=-1.29\times10^{41}\,[\mathrm{J\,m^{-3}}]\) and
\(\max\varrho_p=2.23\times10^{42}\,[\mathrm{J\,m^{-3}}]\).
These extrema are governed by the near–wall structure. By construction, blue indicates \(\varrho_p<0\) and red indicates \(\varrho_p>0\); the solid black curve marks
the zero locus \(\varrho_p=0\) (equivalently \(\lambda_H=0\)).

For the representative run \(\rho=5~[\mathrm{m}]\), \(\sigma=4~[\mathrm{m}^{-1}]\), \(v/c=1\) (so \(\sigma\rho=20\)),
the \emph{positive} energy density on the prime meridional plane appears as two \emph{disjoint, crescent–shaped} lobes concentrated near \(r\simeq\rho\).
The crescents are centered along the \(\pm x\) directions (axis of motion), where \(\varrho_p\) attains its in–plane maxima, and they taper toward the \(\pm y\) directions at the equator, where the \(\varrho_p=0\) contour pinches the would–be annulus into two separate lobes.
The configuration exhibits twofold (\(C_2\)) rotational symmetry about the bubble center, i.e.\ invariance under the half–turn \((x,y)\mapsto(-x,-y)\).
Each crescent has characteristic half–thickness \(\sim 1/\sigma\), which on the nondimensional axes corresponds to
\(\Delta(x/\rho)\sim\Delta(y/\rho)\sim 1/(\sigma\rho)\approx 0.05\).

The \emph{negative} energy density appears (i) outside the bubble ($r>\rho$ ) as weak exterior caps that occupy approximately quarter–disk sectors, whose \(\varrho_p=0\) boundaries are straight rays at \(\pm45^\circ\) to the \(x\) axis; and (ii) inside the bubble ($r<\rho$ )  as thin, very low amplitude slivers aligned with the \(y\) axis.

\noindent\textbf{3D interpretation.}
Revolving the two positive–energy crescents about the \(x\)–axis (polar axis) produces two \emph{disjoint polar collars} on the thin spherical shell \(r\simeq\rho\), with maximal thickness at the poles (\(\theta=0,\pi\)) and vanishing near the equator (\(\theta=\pi/2\)).
In contrast, the exterior negative–energy wedges lie at slightly larger radius (\(r>\rho\)) and, upon revolution, sweep out broad \emph{equatorial mantles} (annular sectors): they are low in amplitude but extend substantially in radius, peaking near the equator and decaying with increasing \(r\).

\paragraph*{Why the volumes differ.}
Although the $z=0$ meridional slice shows comparable positive and negative $\varrho_p$ areas—which approach equality as $r\to\infty$ because the $\varrho_p=0$ contours are straight rays at $\pm45^\circ$ to the $x$-axis for $r>\rho$—the axisymmetric 3-D volume element weights each $(r,\theta)$-ring by $2\pi r^{2}\sin\theta$ (Eq.~\eqref{eq:vol_element}). The positive support is confined to thin polar collars near $r\simeq\rho$ at small $\sin\theta$, whereas the negative support occupies broad equatorial mantles extending to larger $r$ where $\sin\theta\approx1$; hence $V_{-}>V_{+}$ even when the $z=0$ slice areas appear similar (cf.\ Fig.~\ref{fig:eq_density}, Tables~\ref{tab:band_captures}–\ref{tab:shell_captures}).

\subsection*{Global energy magnitudes and volumes (baseline window \(R_{\rm integ}=12\rho\))}

Let \(\lambda_H\) be as in Eq.~\eqref{eq:lambdaH_def}. By \(\kappa\,\varrho_p=-\lambda_H\),
\begin{equation}
\begin{aligned}
E_- \;&=\; \int_{\mathcal R_-} |\varrho_p|\,dV
      \;=\; \frac{1}{\kappa}\int_{\mathcal R_-} |\lambda_H|\,dV,\qquad
E_+ \;=\; \int_{\mathcal R_+} |\varrho_p|\,dV
      \;=\; \frac{1}{\kappa}\int_{\mathcal R_+} |\lambda_H|\,dV,\\[4pt]
E_{\rm neg}&=-E_-,\quad E_{\rm pos}=E_+,\quad E_{\rm net}=E_+-E_- .
\end{aligned}
\label{eq:Epm_magnitudes}
\end{equation}
Region volumes are
\begin{equation}
V_-=\!\int_{\mathcal R_-}\! dV,\qquad
V_0=\!\int_{\mathcal R_0}\! dV,\qquad
V_+=\!\int_{\mathcal R_+}\! dV ,\qquad
V_{\rm act}=V_-+V_0+V_+ .
\label{eq:volumes}
\end{equation}

\paragraph*{Baseline numerical totals (\(R_{\rm integ}=12\rho\), no tail).}
See Tables~\ref{tab:globals_SI} and \ref{tab:ratios_checks}.

\subsection*{Large--radius tail model and two--point estimator}
\label{sec:tail_model}

\paragraph*{Asymptotic scaling and why \(1/R\).}
At large \(r\), the analytic form used in the code gives
\(g(r)=1-\dfrac{I}{r}+\mathcal O\!\big(e^{-2\sigma r}/r\big)\) with 
\(I=\rho\,\coth(\sigma\rho)\).
Then \(\Phi=v\,r\,g(r)\cos\theta = v\,(x-I\cos\theta)+\mathcal O(e^{-2\sigma r})\),
so the Hessian obeys \(H_{ij}=\partial_i\partial_j\Phi=\mathcal O(r^{-2})\) and the invariant
\(\lambda_H=\tfrac12\!\big(\mathrm{Tr}[H^2]-(\mathrm{Tr}\,H)^2\big)=\dfrac{C(\theta)}{r^{4}}+\mathcal O(e^{-4\sigma r})\),
with \(C(\theta)\ge 0\) on a set of nonzero solid angle.
With an \emph{absolute} deadband (Eq.~\eqref{eq:tolerances}) and over the radii used for fitting
(\(8\rho\!\le r\!\le 12\rho\)), the active solid--angle fraction where \(|\lambda_H|>\lambda^{(+)}_{\rm tol}\) is
\emph{order--unity} (cf.\ Table~\ref{tab:band_captures}).
Hence a shell contributes \(dE_-(r)\propto (r^{-4})\times r^{2}\,dr = r^{-2}dr\),
and the far--field tail satisfies \(\displaystyle E_-(\infty)-E_-(R)\sim\int_R^\infty r^{-2}dr\sim 1/R\).
(For \(E_+\), support is compact near \(r\simeq\rho\), so the tail is much smaller—matching the data.)

\paragraph*{Two--point \(1/R\) estimator.}
If \(E_\pm(\infty)-E_\pm(R)=A_\pm/R+\mathcal O(R^{-2})\), then for \(R_1<R_2\)
\begin{equation}
E_\pm(\infty)\;\approx\;\frac{R_2\,E_\pm(R_2)-R_1\,E_\pm(R_1)}{R_2-R_1}
\;=\;E_\pm(R_2)\;+\;\frac{R_1}{R_2-R_1}\,\Big(E_\pm(R_2)-E_\pm(R_1)\Big).
\label{eq:tail_two_point}
\end{equation}
We use \(R_1=8\rho\), \(R_2=12\rho\), so the tail increment is
\(\Delta E_\pm^{\rm tail}=E_\pm(\infty)-E_\pm(12\rho)=2\,[E_\pm(12\rho)-E_\pm(8\rho)]\).

\paragraph*{Tail--corrected totals (two--point).}
Using Eq.~\eqref{eq:tail_two_point} with \(R_1=8\rho\), \(R_2=12\rho\) and the measured
\(\{E_\pm(8\rho),E_\pm(12\rho)\}\), we obtain
\[
\Delta E_-^{\rm tail}=1.19\times10^{43}\ [\mathrm{J}],\qquad
\Delta E_+^{\rm tail}=3.42\times10^{42}\ [\mathrm{J}],
\]
and
\begin{equation}
E_-(\infty)=1.33\times10^{44}\ [\mathrm{J}],\quad
E_+(\infty)=1.33\times10^{44}\ [\mathrm{J}],\quad
E_{\rm net}(\infty)=-1.12\times10^{41}\ [\mathrm{J}].
\label{eq:tail_results}
\end{equation}
Thus \(E_{\rm abs}(\infty)=E_-(\infty)+E_+(\infty)=2.66\times10^{44}\ [\mathrm{J}]\) and
\(|E_{\rm net}|/E_{\rm abs}(\infty)=0.04\%\).

\begin{table}[htbp]
\centering
\caption{Baseline (\( R_{\rm integ}=12\rho\)) versus tail--corrected (\(R\!\to\!\infty\)) global energies using a two--point \(1/R\) extrapolation with \(R_1=8\rho\), \(R_2=12\rho\).}
\label{tab:baseline_vs_tail}
\begin{tabular}{lcc}
\hline
Quantity & Baseline at \(12\rho\) & Tail--corrected \(R\!\to\!\infty\) \\
\hline
\(E_-\) [J] & \(1.21\times10^{44}\) & \(1.33\times10^{44}\) \\
\(E_+\) [J] & \(1.29\times10^{44}\) & \(1.33\times10^{44}\) \\
\(E_{\rm net}=E_+-E_-\) [J] &
\(8.37\times10^{42}\) &
\(-1.12\times10^{41}\) \\
\(\displaystyle \frac{|E_{\rm net}|}{\int |\varrho_p|\, dV}\) [\%] &
\(3.34\) &
\(0.04\) \\
\hline
\end{tabular}
\end{table}

Numerically, the tail raises both magnitudes and drives the net toward zero:
\(|E_{\rm net}|/\!\int|\varrho_p|\,dV\) drops from \(3.34\%\) (baseline window) to \(0.04\%\) (tail--corrected).

\textbf{The near--cancellation \(E_+\approx E_-\) thus becomes extremely tight once the far--field is included.}

\paragraph*{Physical remark on near cancellation of net energy.}
The \emph{proper} energy magnitudes nearly cancel: on the tail--corrected totals we find
\(|E_{\rm net}|/E_{\rm abs}=0.04\%\) (two--point \(1/R\) extrapolation; percentages reported to two decimals).
This statement concerns the \emph{proper} energy density defined by the timelike eigenvalue of \(T^\mu{}_\nu\);
it does not, by itself, establish a vanishing ADM or Komar mass, which require global geometric fluxes/constraints beyond the local energy density.

\subsection*{Localization diagnostics (baseline window)}

\paragraph*{Equatorial band captures about \(\theta=\pi/2\).}
For half--width \(\Delta\),
\begin{equation}
\begin{aligned}
f^{\mathrm{band}}_{V_\pm}(\Delta)
&=\frac{\displaystyle
\int_{\mathcal R_\pm}\!\!\int_{\pi/2-\Delta}^{\pi/2+\Delta}
2\pi r^2\sin\theta\,dr\,d\theta}{V_\pm},\\[6pt]
f^{\mathrm{band}}_{E_\pm}(\Delta)
&=\frac{\displaystyle
\int_{\mathcal R_\pm}\!\!\int_{\pi/2-\Delta}^{\pi/2+\Delta}
\frac{|\lambda_H|}{\kappa}\,2\pi r^2\sin\theta\,dr\,d\theta}{E_\pm}.
\end{aligned}
\label{eq:band_captures}
\end{equation}

\paragraph*{Bubble wall--shell captures.}
For half--width \(w/\sigma\) about \(r=\rho\) (clipped to
\([0,R_{\rm integ}]\)),
\begin{equation}
\begin{aligned}
f^{\mathrm{shell}}_{V_\pm}(w)
&=\frac{\displaystyle
\int_{\mathcal R_\pm}\!\!\int_{\rho-w/\sigma}^{\rho+w/\sigma}
2\pi r^2\sin\theta\,dr\,d\theta}{V_\pm},\\[6pt]
f^{\mathrm{shell}}_{E_\pm}(w)
&=\frac{\displaystyle
\int_{\mathcal R_\pm}\!\!\int_{\rho-w/\sigma}^{\rho+w/\sigma}
\frac{|\lambda_H|}{\kappa}\,2\pi r^2\sin\theta\,dr\,d\theta}{E_\pm}.
\end{aligned}
\label{eq:shell_captures}
\end{equation}

\subsection*{Numerical validation}

Unmasked signed integral reproduces the masked net energy to within roundoff:
\begin{equation}
E_{\rm net}^{(\mathrm{unmasked})}
=\int\!\left(-\frac{\lambda_H}{\kappa}\right)dV
\;\approx\;
E_{\rm net}^{(\mathrm{masked})}.
\label{eq:direct_check}
\end{equation}
Masks partition the grid and the discrete volume form is consistent:
\begin{equation}
\max_{\text{grid}}
\bigl|\mathbf{1}_{\mathcal R_-}+\mathbf{1}_{\mathcal R_0}+\mathbf{1}_{\mathcal R_+}-1\bigr|=0,
\qquad
\sum_{\text{grid}} dV \big/ V_{\mathrm{plot}}=1.000937,
\label{eq:mask_jacobian_checks}
\end{equation}
indicating a relative volume-normalization error of $9.37\times10^{-4}$.

\subsection*{Numerical summary for \(\rho=5~[\mathrm{m}]\), \(\sigma=4~[\mathrm{m}^{-1}]\), \(v/c=1\)}
\label{sec:num_summary}

\noindent\textbf{Tail--corrected net.}
Using the two--point \(1/R\) extrapolation, \textbf{the net proper energy satisfies}
\(|E_{\rm net}(\infty)|/E_{\rm abs}(\infty)=0.04\%\), i.e.\ \textbf{effectively zero at the reported precision.}

\begin{table}[htbp]
\centering
\caption{Global energy magnitudes and volumes (baseline window \(R_{\rm integ}=12\rho\)).}
\label{tab:globals_SI}
\begin{tabular}{l c}
\hline
Quantity & Value \\
\hline
Negative energy \(E_-\) [J] & \(1.21\times10^{44}\) \\
Positive energy \(E_+\) [J] & \(1.29\times10^{44}\) \\
Net energy \(E_{\rm net}=E_+-E_-\) [J] & \(8.37\times10^{42}\) \\
\(\displaystyle \int |\varrho_p|\,dV\) [J] & \(2.50\times10^{44}\) \\
Negative volume \(V_-\) [m\(^3\)] & \(6.38\times10^{5}\) \\
Positive volume \(V_+\) [m\(^3\)] & \(2.34\times10^{5}\) \\
Zero volume \(V_0\) [m\(^3\)] & \(3.40\times10^{4}\) \\
\hline
\end{tabular}
\end{table}

\begin{table}[htbp]
\centering
\caption{Ratios and consistency checks (baseline window).}
\label{tab:ratios_checks}
\begin{tabular}{l c}
\hline
Metric & Value \\
\hline
Energy ratio \(E_+/E_-\) & \(1.07\) \\
\(|E_{\rm net}|/\!\int |\varrho_p|\,dV\) & \(3.34\%\) \\
\(V_-/(V_-+V_++V_0)\) & \(70.44\%\) \\
\(V_+/(V_-+V_++V_0)\) & \(25.81\%\) \\
\(V_0/(V_-+V_++V_0)\) & \(3.75\%\) \\
\(V_-/(V_-+V_+)\) & \(73.19\%\) \\
\(V_+/(V_-+V_+)\) & \(26.81\%\) \\
\(V_-/V_+\) & \(272.95\%\) \\
Direct signed sum \(E_{\rm net}^{\rm direct}\) [J] & \(8.40\times10^{42}\) \\
\(|E_{\rm net}-E_{\rm net}^{\rm direct}|/\!\int |\varrho_p|\,dV\) & \(0.0138\%\) \\
\hline
\end{tabular}
\end{table}

\begin{table}[htbp]
\centering
\caption{Prime meridional band captures about \(\theta=\pi/2\) (baseline window).}
\label{tab:band_captures}
\begin{tabular}{c c c c c}
\hline
\(\Delta\) & \(V_-\) & \(E_-\) & \(V_+\) & \(E_+\) \\
\hline
\(5^\circ\)  & \(12.52\%\) & \(19.53\%\) & \(0.00\%\)  & \(0.00\%\) \\
\(10^\circ\) & \(24.79\%\) & \(37.77\%\) & \(0.00\%\)  & \(0.02\%\) \\
\(15^\circ\) & \(36.87\%\) & \(54.24\%\) & \(0.01\%\)  & \(0.28\%\) \\
\(20^\circ\) & \(48.67\%\) & \(68.50\%\) & \(0.02\%\)  & \(1.18\%\) \\
\hline
\end{tabular}
\end{table}

\begin{table}[htbp]
\centering
\caption{Bubble wall--shell captures for \(|r-\rho|\le w/\sigma\) (baseline window).}
\label{tab:shell_captures}
\begin{tabular}{c c c c c}
\hline
\(w\) & \(V_-\) & \(E_-\) & \(V_+\) & \(E_+\) \\
\hline
0.5 & \(0.00\%\) & \(0.33\%\) & \(0.03\%\)  & \(34.24\%\) \\
1.0 & \(0.00\%\) & \(1.09\%\) & \(0.06\%\)  & \(58.43\%\) \\
2.0 & \(0.01\%\) & \(3.43\%\) & \(0.11\%\)  & \(71.21\%\) \\
3.0 & \(0.02\%\) & \(7.52\%\) & \(0.14\%\)  & \(74.42\%\) \\
\hline
\end{tabular}
\end{table}

\section{Causal Proof: A Fixed-Smoothing Vortical Ablation Study}
\label{app:ablation}

\noindent
\textbf{Goal.} To isolate the causal role of \emph{irrotationality}, we conduct a rigorous ablation study. We hold the baseline ``smoothing'' profile—derived from the scalar potential \(\Phi\) and its associated integral map \(g(r)\)—\emph{fixed}, and then introduce a tunable vortical component. By recomputing the energetics with the \emph{same} numerical grid, masks, and tolerances, any observed degradation in the energy metrics is attributable solely to the loss of the curl-free condition, not to a change in smoothing.

\subsection*{Shared assumptions and numerics}
We work on a stationary, flat spatial slice with unit lapse \(\alpha=1\). All physical parameters (\(\rho=5~[\mathrm{m}]\), \(\sigma=4~[\mathrm{m}^{-1}]\), \(v/c=1\)), the integration window (\(r\in[0,12\rho]\)), and the numerical grid \(\big(n_r,n_\theta\big)=(1601,1441)\) are identical to the baseline irrotational case. Crucially, the masks \(\mathcal R_\pm,\mathcal R_0\) and their absolute tolerances \(\lambda^{(\pm)}_{\mathrm{tol}}\) are \emph{frozen} to the baseline values (see Appendix~\ref{app:global_energy}, Eqs.~\eqref{eq:regions}--\eqref{eq:tolerances}), ensuring a consistent and fair comparison across the entire ablation sweep. Central values are reported to \(3\) significant figures and percentages to two decimals, matching the conventions used elsewhere.

\subsection*{Kinematics and tensors}

\subsubsection*{Baseline irrotational structure}
The baseline metric is purely irrotational, with \(\beta_i=\partial_i\Phi\) and \(\Phi=v\,r\,g(r)\cos\theta\). The extrinsic curvature is the Hessian of the potential,
\begin{equation}
K_{ij} \;=\; H_{ij}:=\partial_i\partial_j\Phi,
\label{eq:K_irrot_app}
\end{equation}
\noindent
As established in Appendix~\ref{app:global_energy}, the Hessian invariant \(\lambda_H\) and the proper energy density \(\varrho_p\) are linked through Eqs.~\eqref{eq:lambdaH_def} and~\eqref{eq:rho_from_lambdaH}.

\subsubsection*{Vortical ablation with fixed smoothing}
We add a vortical component to the baseline shift vector,
\begin{equation}
\beta \;=\; \nabla\Phi \;+\; \nabla\times \vA ,
\label{eq:abl_beta_appendix}
\end{equation}
where the vorticity is introduced via a toroidal vector potential
$\vA = \vAphi(r,\theta)\,\mathbf{e}_\phi$ with
\begin{equation}
\vAphi(r,\theta) \;=\; \psi(r)\,\sin\theta .
\label{eq:A_choice_appendix}
\end{equation}
The radial profile \(\psi(r)\) generates a vortical perturbation localized in a compact shell around the bubble wall:
\begin{equation}
\psi(r)\;=\;\eta\,v\,\Big(\tfrac{r}{\rho}\Big)^2\,
\exp\!\left[-\Big(\tfrac{r-\rho}{w/\sigma}\Big)^2\right],
\qquad
w=1,\quad
\eta\in\{0,\,0.25,\,0.5,\,0.75,\,1\}.
\label{eq:psi_profile_appendix}
\end{equation}
In an orthonormal spherical tetrad \((\hat r,\hat\theta,\hat\phi)\), the nonzero components of the full extrinsic curvature \(K\) are
\begin{align}
K_{\hat r\hat r}(r,\theta) \;&=\; 
\Big[\,v\,(2g'(r)+r\,g''(r)) + 2\Big(\tfrac{\psi'(r)}{r}-\tfrac{\psi(r)}{r^2}\Big)\,\Big]\;\cos\theta ,
\label{eq:Krr_app}\\[3pt]
K_{\hat\theta\hat\theta}(r,\theta) \;&=\; 
\Big[\,v\,g'(r) - \tfrac{\psi'(r)}{r}+\tfrac{\psi(r)}{r^2}\,\Big]\;\cos\theta ,
\qquad
K_{\hat\phi\hat\phi}(r,\theta)=K_{\hat\theta\hat\theta}(r,\theta),
\label{eq:Ktt_app}\\[3pt]
K_{\hat r\hat\theta}(r,\theta) \;&=\; 
-\,\Big[\,v\,g'(r)+\tfrac12\,\psi''(r)\,\Big]\;\sin\theta ,
\qquad
K_{\hat r\hat\phi}(r,\theta)=K_{\hat\theta\hat\phi}(r,\theta)=0.
\label{eq:Krt_app}
\end{align}
The quadratic invariant for the ablation, which generalizes Eqs.~\eqref{eq:lambdaH_def} and~\eqref{eq:rho_from_lambdaH}
(and reduces to those relations when \(\eta=0\)), is
\begin{equation}
\lambda_{\mathrm{blend}}(r,\theta)\;=\;\tfrac12\!\left(K_{ij}K^{ij}-K^2\right),
\qquad
K:=K^i{}_i=K_{\hat r\hat r}+K_{\hat\theta\hat\theta}+K_{\hat\phi\hat\phi},
\label{eq:lambda_blend_appendix}
\end{equation}
and in terms of the orthonormal components it reads
\begin{equation}
\lambda_{\mathrm{blend}}
\;=\;
\tfrac12\Big(
K_{\hat r\hat r}^2 + 2\,K_{\hat\theta\hat\theta}^2 + 2\,K_{\hat r\hat\theta}^2
- \big(K_{\hat r\hat r} + 2\,K_{\hat\theta\hat\theta}\big)^2
\Big),
\qquad
\kappa\,\varrho_p \;=\; -\,\lambda_{\mathrm{blend}} .
\label{eq:lambda_blend_explicit_appendix}
\end{equation}

\subsubsection*{Baseline equality at \texorpdfstring{\(\eta=0\)}{eta=0}}
To ensure bit-for-bit consistency, the \(\eta=0\) case reuses the stored baseline grid \(\lambda_H(r,\theta)\), guaranteeing that all baseline totals are reproduced to roundoff:
\begin{equation}
\eta=0 \quad\Longrightarrow\quad \lambda_{\mathrm{blend}}(r,\theta)\equiv \lambda_H(r,\theta).
\label{eq:eta0_identity_appendix}
\end{equation}

\subsection*{Global measures}
Global magnitudes and volumes follow Appendix~\ref{app:global_energy}, Eqs.~\eqref{eq:Epm_magnitudes}--\eqref{eq:volumes}, with \(\lambda\) replaced by \(\lambda_{\mathrm{blend}}\) in the ablation.

\subsection*{Numerical results (3 s.f.)}

\subsubsection*{Windowed sweep at \texorpdfstring{\(R_{\mathrm{integ}}=12\rho\)}{R\_integ = 12 rho}}
Table~\ref{tab:abl_sweep} reports the global magnitudes
\(\displaystyle E_-=\int_{\mathcal R_-}\!|\varrho_p|\,dV\),
\(\displaystyle E_+=\int_{\mathcal R_+}\!|\varrho_p|\,dV\),
and the ratio \(E_+/E_-\) for \(\eta\in\{0,0.25,0.5,0.75,1\}\).
The mask tolerances are held fixed to their baseline values.

\begin{table}[htbp]
\centering
\caption[Fixed-smoothing vortical ablation sweep]{Fixed-smoothing vortical ablation sweep on the baseline window \(R_{\mathrm{integ}}=12\rho\). Central values to \(3\) s.f.}
\label{tab:abl_sweep}
\begin{tabular}{c c c c}
\hline
\(\eta\) & \(E_-\) [J] & \(E_+\) [J] & \(E_+/E_-\) \\
\hline
\textbf{\(0.00\)} & \textbf{\(1.21\times10^{44}\)} & \textbf{\(1.29\times10^{44}\)} & \textbf{\(1.07\)} \\
\(0.25\) & \(3.92\times10^{45}\) & \(6.00\times10^{43}\) & \(1.53\times10^{-2}\) \\
\(0.50\) & \(1.55\times10^{46}\) & \(4.91\times10^{43}\) & \(3.16\times10^{-3}\) \\
\(0.75\) & \(3.49\times10^{46}\) & \(4.13\times10^{43}\) & \(1.18\times10^{-3}\) \\
\(1.00\) & \(6.19\times10^{46}\) & \(3.62\times10^{43}\) & \(5.84\times10^{-4}\) \\
\hline
\end{tabular}
\end{table}

\subsection*{Interpretation}
Because \(g(r)\) is held fixed and the mask tolerances are frozen to the baseline values, Table~\ref{tab:abl_sweep} shows that even a modest vortical admixture (\(\eta=0.25\)) increases the negative-energy magnitude \(E_-\) by a factor of \(\approx 32\) relative to baseline, while collapsing the balance \(E_+/E_-\) by nearly two orders of magnitude (from \(1.07\) to \(1.53\times10^{-2}\)). As \(\eta\) increases to \(1\), \(E_-\) grows by a factor \(>5\times10^{2}\) and \(E_+/E_-\) falls to \(5.84\times10^{-4}\).
\textbf{This demonstrates that the favorable energy properties reported in the main text—particularly the near-cancellation of integrated proper energy—are a direct consequence of the irrotational, curl-free kinematics and not an artifact of profile smoothing.}


\section{Boundary behaviour details}
\label{app:boundary-details}

The limiting values of the shift components are summarized in Table~\ref{tab:boundary-behaviour}.
For Alcubierre, the interior shift equals $+v(t)$ along $+x$ and decays for large $r$.
For Nat\'ario and for the irrotational model (Sec.~\ref{secIrrotational}), the bubble center is at rest while the external universe drifts with velocity $-v(t)$, consistent with the adopted convention.

\begin{table}[htbp]
\centering
\setlength{\tabcolsep}{6pt}
\renewcommand{\arraystretch}{1.2}
\begin{tabular}{|l|c|c|c|c|c|}
\hline
Model & Angle
& $\displaystyle \lim_{r\to 0}\beta^{\hat r}$ 
& $\displaystyle \lim_{r\to 0}\beta^{\hat\theta}$
& $\displaystyle \lim_{r\to \infty}\beta^{\hat r}$
& $\displaystyle \lim_{r\to \infty}\beta^{\hat\theta}$ \\
\hline
Alcubierre & $\theta=0$ & $v(t)$ & $0$ & $0$ & $0$ \\
           & $\theta=\dfrac{\pi}{2}$ & $0$ & $-\!v(t)$ & $0$ & $0$ \\
\hline
Nat\'ario  & $\theta=0$ & $0$ & $0$ & $-\!v(t)$ & $0$ \\
           & $\theta=\dfrac{\pi}{2}$ & $0$ & $0$ & $0$ & $v(t)$ \\
\hline
Irrotational & $\theta=0$ & $0$ & $0$ & $-\!v(t)$ & $0$ \\
           & $\theta=\dfrac{\pi}{2}$ & $0$ & $0$ & $0$ & $v(t)$ \\
 \hline

\end{tabular}

\caption{Boundary behaviour of the shift tetrad components for the Alcubierre, Nat\'ario, and irrotational drives: limiting values of $(\beta^{\hat r},\beta^{\hat\theta})$ at $r\to0$ and $r\to\infty$ for two representative polar angles. (In all three, $\beta^{\hat\phi}=0$ by axial symmetry.) }

\label{tab:boundary-behaviour}
\end{table}

\section{Canonical Type-IV blocks and parameter maps—Relation to the literature}
\label{app:typeIV_maps}

Table~\ref{tab:typeIV_maps} aligns the $2\times2$ blocks used in the literature with the mixed canonical form
\(T^{a}{}_{b}=\begin{psmallmatrix}\xi & \mu\\[-1pt]-\mu & \xi\end{psmallmatrix}\).
Apparent differences come only from index position (covariant/contravariant), basis ordering, and signature; after the appropriate raising/lowering and basis change, all reduce to the same mixed block with conjugate eigenvalues \(\xi\pm i\mu\) and an invariant null two–plane \(g(z,\bar z)=0\).

\emph{Mappings (see Table~\ref{tab:typeIV_maps}):}
We use the mixed canonical form \(T^{a}{}_{b}=\begin{psmallmatrix}\xi & \mu\\[-1pt]-\mu & \xi\end{psmallmatrix}\) in an orthonormal \((t,s)\) basis with signature \((-,+,+,+)\).  Hawking–Ellis \cite[p.\,90]{hawking_ellis_1973} start from \(T^{ab}=\begin{psmallmatrix}-\varkappa & \nu\\ \nu & 0\end{psmallmatrix}\) (orthonormal \((s,t)\); \((+,+,+,-)\)).
Lowering yields \(T^{a}{}_{b}=\begin{psmallmatrix}-\varkappa & -\nu\\ \nu & 0\end{psmallmatrix}\) with eigenvalues \(\tfrac12\!\left(-\varkappa \pm i\sqrt{4\nu^{2}-\varkappa^{2}}\right)\), hence \(\xi=-\varkappa/2\), \(\mu=\sqrt{\nu^{2}-\varkappa^{2}/4}\); the Type~IV condition is \(\varkappa^{2}<4\nu^{2}\).
Maeda–Martínez \cite{Maeda2020} use \(T^{ab}=\begin{psmallmatrix}\rho & \nu\\ \nu & -\rho\end{psmallmatrix}\) ((\(t,s\)); \((-,+,+,+)\)), giving \(T^{a}{}_{b}=\begin{psmallmatrix}-\rho & \nu\\ -\nu & -\rho\end{psmallmatrix}\) with \(-\rho\pm i\nu\), so \(\xi=-\rho\), \(\mu=\nu\).
Martín–Moruno–Visser \cite{Martín-Moruno_2018} have \(T^{ab}=\begin{psmallmatrix}\rho & f\\ f & -\rho\end{psmallmatrix}\) (same signature), hence \(T^{a}{}_{b}=\begin{psmallmatrix}-\rho & f\\ -f & -\rho\end{psmallmatrix}\) with \(-\rho\pm i f\), so \(\xi=-\rho\), \(\mu=f\).
Petrov \cite[§47]{Petrov} writes the covariant block \(T_{ab}=\begin{psmallmatrix}\alpha-\beta & \alpha+\beta\\ \alpha+\beta & -\alpha+\beta\end{psmallmatrix}\) in a quasi–orthonormal \((k,l)\) frame with \(g_{ab}=\begin{psmallmatrix}1&1\\ 1&-1\end{psmallmatrix}\); raising gives \(T^{a}{}_{b}=\begin{psmallmatrix}\alpha & \beta\\ -\beta & \alpha\end{psmallmatrix}\) with \(\alpha\pm i\beta\), so \(\xi=\alpha\), \(\mu=\beta\).

All four descriptions thus encode the same algebraic geometry. The invariant null two–plane does \emph{not} imply a real null eigenvector, consistent with \cite{hawking_ellis_1973}. The “invariant-null-plane $\Rightarrow$ null-eigenvector’’ lemma of Stephani et al.\ \cite[p.~58, item (4)]{Stephani_Kramer_MacCallum_Hoenselaers_Herlt_2003} is proved for \(R^{a}{}_{b}\); since \(G^{a}{}_{b}=R^{a}{}_{b}-\tfrac12 R\,\delta^{a}{}_{b}\) differs only by a trace, eigenvectors and invariant planes are unchanged, and the conclusion transfers verbatim to \(G^{a}{}_{b}\propto T^{a}{}_{b}\).

\begin{table*}[t]
\caption{Canonical Type–IV $2\times2$ blocks across the literature and mapping to our mixed canonical form \(T^{a}{}_{b}=\begin{psmallmatrix}\xi & \mu\\[-1pt]-\mu & \xi\end{psmallmatrix}\).}
\label{tab:typeIV_maps}
\centering
\small
\setlength{\tabcolsep}{5.5pt}
\resizebox{\textwidth}{!}{%
\begin{tabular}{@{}l l l l l l@{}}
\toprule
Source
& Initial tensor block
& Basis / signature
& Mixed block
& Eigenvalues
& Mapping to \((\xi,\mu)\) \\
\midrule
This work
& \(T^{a}{}_{b}=\begin{psmallmatrix}\xi & \mu\\[2pt]-\mu & \xi\end{psmallmatrix}\)
& Ort.\,(t,s) / \((-,+,+,+)\)
& \(T^{a}{}_{b}=\begin{psmallmatrix}\xi & \mu\\[2pt]-\mu & \xi\end{psmallmatrix}\)
& \(\xi \pm i\mu\)
& \(\xi=\xi,\;\mu=\mu\) \\
Hawking–Ellis
& \(T^{ab}=\begin{psmallmatrix}-\varkappa & \nu\\[2pt]\nu & 0\end{psmallmatrix}\)
& Ort.\,(s,t) / \((+,+,+,-)\)
& \(T^{a}{}_{b}=\begin{psmallmatrix}-\varkappa & -\nu\\[2pt]\nu & 0\end{psmallmatrix}\)
& \(-\varkappa/2\)
& \(\xi=-\varkappa/2\) \\
& & & & \(\pm \tfrac{i}{2}\sqrt{4\nu^{2}-\varkappa^{2}}\)
& \(\mu=\sqrt{\nu^{2}-\varkappa^{2}/4}\) \\
Maeda–Martinez
& \(T^{ab}=\begin{psmallmatrix}\rho & \nu\\[2pt]\nu & -\rho\end{psmallmatrix}\)
& Ort.\,(t,s) / \((-,+,+,+)\)
& \(T^{a}{}_{b}=\begin{psmallmatrix}-\rho & \nu\\[2pt]-\nu & -\rho\end{psmallmatrix}\)
& \(-\rho \pm i\nu\)
& \(\xi=-\rho,\;\mu=\nu\) \\
Moruno–Visser
& \(T^{ab}=\begin{psmallmatrix}\rho & f\\[2pt]f & -\rho\end{psmallmatrix}\)
& Ort.\,(t,s) / \((-,+,+,+)\)
& \(T^{a}{}_{b}=\begin{psmallmatrix}-\rho & f\\[2pt]-f & -\rho\end{psmallmatrix}\)
& \(-\rho \pm i f\)
& \(\xi=-\rho,\;\mu=f\) \\
Petrov
& \(T_{ab}=\begin{psmallmatrix}\alpha-\beta & \alpha+\beta\\[2pt]\alpha+\beta & -\alpha+\beta\end{psmallmatrix}\)
& Qs.\,(k,l) / \(g_{ab}=\begin{psmallmatrix}1&1\\[1pt]1&-1\end{psmallmatrix}\)
& \(T^{a}{}_{b}=\begin{psmallmatrix}\alpha & \beta\\[2pt]-\beta & \alpha\end{psmallmatrix}\)
& \(\alpha \pm i\beta\)
& \(\xi=\alpha,\;\mu=\beta\) \\
\bottomrule
\end{tabular}}
\end{table*}

\end{appendices}


\bibliography{sn-bibliography}

\end{document}